\documentclass[aip,graphicx]{revtex4-1}
  \usepackage{amsthm}
  \usepackage{amsmath}
  \usepackage{amsfonts}
  \usepackage{amssymb}
  \usepackage{amscd}
  \usepackage{bm} 
  \usepackage[mathscr]{eucal}
  \usepackage{mathtools} 
  \usepackage{enumitem}
  \usepackage{physics}
\usepackage{relsize} 
  \usepackage{color}
  \usepackage{fancyhdr}
  \usepackage{hyperref} 
  \usepackage{tikz-cd}
  \usepackage{framed}
  \usepackage[noabbrev]{cleveref} 
 \usepackage[utf8]{inputenc}
 \usepackage{blkarray}
\draft 

\theoremstyle{definition}  
   \newtheorem{defn}{Definition}[section]
   \newtheorem{eg}[defn]{Example}
   \newtheorem{egs}[defn]{Examples}

   \newtheorem{rmk}[defn]{Remark}
   \newtheorem{rmks}[defn]{Remarks}

  \theoremstyle{plain}  
   \newtheorem{thm}[defn]{Theorem}
   \newtheorem{lem}[defn]{Lemma}
   \newtheorem{prop}[defn]{Proposition}
   \newtheorem{cor}[defn]{Corollary}

  \theoremstyle{remark} 
  \newtheorem*{notn}{Notation}

   \newcommand{\B}[1]{\mathscr{B}({#1})}
   \newcommand{\Bo}[1]{\mathscr{B}_1({#1})}
   \newcommand{\Bt}[1]{\mathscr{B}_2({#1})}

   \newcommand{\cv}[1]{\mathbf{#1}}

   \newcommand{\bbell}{\bm{\ell}}
   \newcommand{\bk}{\mathbf{k}}
   \newcommand{\br}{\mathbf{r}}
   \newcommand{\bs}{\mathbf{s}}
   \newcommand{\bt}{\mathbf{t}}
   \newcommand{\bu}{\mathbf{u}}
   \newcommand{\bv}{\mathbf{v}}
   \newcommand{\q}{\mathbf{q}}
   \newcommand{\m}{\mathbf{m}}
   
   \newcommand{\x}{\mathbf{x}}
   \newcommand{\y}{\mathbf{y}}
   \newcommand{\z}{\mathbf{z}}
   \newcommand{\CE}{\mathcal{E}}
   \newcommand{\CH}{\mathcal{H}}
   \newcommand{\CK}{\mathcal{K}}
   \newcommand{\GH}{\Gamma(\mathcal{H})}
   \newcommand{\GC}{\Gamma(\mathbb{C})}
   \newcommand{\GL}{\Gamma(L)} 
   \newcommand{\GOL}{\Gamma_{0}(L)}

   \newcommand{\GCn}{\Gamma(\mathbb{C}^{n})}
   
   \newcommand{\BC}{\mathbb{C}}
   \newcommand{\BR}{\mathbb{R}}
   \newcommand{\BF}{\mathbb{F}}
   \newcommand{\BZ}{\mathbb{Z}}

 \newcommand{\numberthis}{\refstepcounter{equation}\tag{\theequation}} 

\let\tr\relax 
\DeclareMathOperator{\tr}{Tr}
\DeclareMathOperator{\diag}{Diag}

\DeclareMathOperator{\ran}{Ran}
\DeclareMathOperator{\re}{Re}
\DeclareMathOperator{\im}{Im}

\DeclareMathOperator{\spn}{span}

\begin{document}


\title
{A Common Parametrization for Finite Mode Gaussian States, their Symmetries and associated Contractions with some Applications} 
\author{Tiju Cherian John}
\email{tijucherian@gmail.com}
\author{K. R. Parthasarathy}
\email{krp@isid.ac.in}
\affiliation{Indian Statistical Institute, Delhi Centre, 7, SJSS Marg, New Delhi, India, 110059}





\begin{abstract}
\begin{center}
    {In memory of V. S. Varadarajan.}
\end{center}
Let $\Gamma(\mathcal{H})$ be the boson Fock space over a finite dimensional Hilbert space $\mathcal{H}$. It is shown that every gaussian symmetry admits a Klauder-Bargmann integral representation in terms of coherent states. Furthermore,  gaussian states, gaussian symmetries, and second quantization contractions belong to a weakly closed, selfadjoint semigroup $\mathcal{E}_2(\mathcal{H})$ of bounded operators in $\Gamma(\mathcal{H})$. This yields a common parametrization for these operators. It is shown that the new parametrization for gaussian states is a fruitful alternative to the customary parametrization by position-momentum mean vectors and covariance matrices. This leads to a rich harvest of corollaries: (i) every gaussian state $\rho$ admits a factorization $\rho= Z_1^\dagger Z_1$, where $Z_1$ is an element of $\mathcal{E}_2(\mathcal{H})$ and has the form $Z_{1} = \sqrt{c}\Gamma(P)\exp{\sum_{r=1}^{n} \lambda_ra_r+\sum_{r,s=1}^{n} \alpha_{rs}a_{r}a_{s}}$ on the dense linear manifold generated by all exponential vectors, where $c$ is a positive scalar, $\Gamma(P)$ is the second quantization of a positive contractive  operator $P$ in $\mathcal{H}$, $a_r$, $1\leq r\leq n$ are the annihilation operators corresponding to the $n$ different modes in $\Gamma(\mathcal{H})$, $\lambda_r\in \mathbb{C}$ and $[\alpha_{rs}]$ is a symmetric matrix in $M_n(\mathbb{C})$; (ii) an explicit particle basis expansion of an arbitrary mean zero pure gaussian state vector along with a density matrix formula for a general gaussian state in terms of its $\mathcal{E}_2(\mathcal{H})$-parameters; (iii)  a class of examples of pure $n$-mode gaussian states which are completely entangled; (iv) tomography of an unknown gaussian state in $\Gamma(\mathbb{C}^n)$ by the estimation of its $\mathcal{E}_2(\mathbb{C}^n)$-parameters using $O(n^2)$ measurements with a finite number of outcomes.

\end{abstract}

\pacs{020000, 0367a, 4250p}

\maketitle 
\section{Introduction}\label{sec:introduction}
The principal aim of this paper is an analysis of gaussian states and their symmetries through a new scheme of parametrization. It replaces the customary mean values and covariances of position and momentum observables which assume all values on the real line.  To this end we consider the $n$-mode boson Fock space $\GH$ over an $n$-dimensional complex Hilbert space $\CH$ with a chosen and fixed orthonormal basis $\{e_j, 1\leq j\leq n\}$, where the index $j$ stands for the $j$-th mode and $n$ for the total number of modes. The study is based on three already well-known tools and a fourth one which is not widely used. We shall repeatedly use the gaussian integral formula, properties of exponential vectors and coherent states, the Weyl displacement operator and the quantum characteristic function and, finally, the elementary idea of generating function of a bounded operator on $\GH$.

Section \ref{sec:2} contains a brief summary of some well-known properties of exponential vectors and coherent states as well as the definition of generating function of a bounded operator on $\GH$. In order to make the exposition fairly self-contained, the Klauder-Bargmann isometry from the Hilbert space $\GH$ into $L^2(\BC^n)$ is described along with a short proof. The Klauder-Bargmann formula for the resolution of the identity operator on $\GH$ as an integral of coherent states with respect to a suitably normalized Lebesgue measure in $\GCn$ is given. This is repeatedly used in our analysis. 

In Section \ref{sec:weyl-oper-wign}, the Weyl displacement operators are presented as a projective unitary representation of the additive group $\CH$ in the Hilbert space $\GH$ and the associated quantum characteristic function of a state in $\GH$ is defined. Also, we provide a proof of a well known result called the Wigner isomorphism between the Hilbert space $\Bt{\GH}$ of all Hilbert-Schmidt operators on $\GH$ and $L^2(\BC^n)$ in this section.

A unitary operator $U$ in $\GH$ is called a gaussian symmetry if $U\rho U^\dagger$ is a gaussian state whenever $\rho$ is a gaussian state. In Section \ref{sec:KB-gaussian}, every gaussian symmetry $U$ is realized as a Klauder-Bargmann integral in terms of coherent states with respect to the Lebesgue measure in $\BC^n$. This construction yields a strongly continuous, projective unitary and irreducible representation of the Lie group which is the semidirect product of the additive group $\CH$ and the group $Sp(\CH)$ of all symplectic linear transforms of $\CH$. It is also shown that the generating function of a gaussian symmetry admits an exponential formula.

In Section \ref{sec:semigroup}, we construct the central object of our paper, namely, the operator semigroup $\CE_2(\CH)$ contained in the algebra $\B{\GH}$ of all bounded operators on $\GH$ by using the idea of generating function. To this end, we identify $\CH$ with $\BC^n$ through the $n$-mode basis mentioned at the beginning of Section \ref{sec:2}. For $\x\in \CH$, let $e(\x)\in \GH$, denote the exponential vector (cf. Section \ref{sec:2}) at $\x$. We say that a bounded operator $Z$ on $\GH$ is in the class $\CE_2(\CH)$ if, for all $\bu, \bv$ in $\CH$, the following holds: 
\begin{equation*}
\mel{e(\bar{\bu})}{Z}{e(\bv)}= c \exp{\cv{u}^T\bm{\alpha}+ \bm{\beta}^T\cv{v}+ \cv{u}^TA\cv{u} +\cv{u}^T\Lambda\cv{v}+\cv{v}^TB\cv{v}}, 
\end{equation*}
for some ordered $6$-tuple $(c, \bm{\alpha}, \bm{\beta}, A, \Lambda, B)$ consisting of a scalar $c \neq 0$, vectors $\bm{\alpha}, \bm{\beta}$ in $\BC^n$ and $n\times n$ complex matrices $A, \Lambda, B$ with $A$ and $B$ being symmetric. Here $e(\bar{\bu})$ and $e(\bv)$ are the exponential vectors in $\GH$ associated with $\bar{\bu}$ and $\bv$ respectively, bar indicating  complex conjugation. We say that this $6$-tuple are the $\CE_2(\CH)$-parameters of the operator $Z$. By the properties of the exponential vectors in $\GH$ summarised in Section \ref{sec:2}, this parametrization is unambiguous. If $Z$ is a selfadjoint element of $\CE_2(\CH)$, $c$ is real, $\bm{\beta} = \bar{\bm{\alpha}}$, $B = \bar{A}$ and $\Lambda$ is hermitian. Thus the $\CE_2(\CH)$-parameters of a selfadjoint element in $\CE_2(\CH)$ reduce to a quadruple $(c, \bm{\alpha}, A, \Lambda)$. Furthermore, $\Lambda\geq 0$ if $Z\geq0$. The class $\CE_2 (\CH)$ is shown to enjoy the following properties: 
\begin{enumerate}
    \item \label{item:47} $\CE_2(\CH)$ is a $\dagger$-closed multiplicative semigroup closed in the weak operator topology. 
    \item \label{item:48} A unitary operator $U$ is in $\CE_2(\CH)$ if and only if it is a gaussian symmetry.
\item \label{item:49} A density operator $\rho$ is in $\CE_2(\CH)$ if and only if $\rho$ is a gaussian state. The element $\Lambda$ in the quadruple of $\CE_2(\CH)$-parameters of a gaussian state is a positive and contractive matrix operator in $\CH$. 
\item \label{item:50} Let $Z$ be a positive element  in $\CE_2(\CH)$  (or, in particular, a gaussian state) with parameters $(c, \bm{\beta}, A, \Lambda)$, where $\bm{\beta} = [\beta_1,\dots,\beta_n]^T\in \BC^n$, $A=[\alpha_{rs}]\in M_n(\BC)$, $\Lambda \in M_n(\BC), \Lambda\geq 0$. Then $Z$ admits a factorization \[Z = Z_{1}^{\dagger}Z_{1},\] where $Z_{1}$ is an element of $\CE_2(\CH)$ and has the form \[Z_{1} = \sqrt{c}\Gamma(\sqrt{\Lambda})\exp{\sum_{r=1}^{n} \beta_ra_r+\sum_{r,s=1}^{n} \alpha_{rs}a_{r}a_{s}}\] on the dense linear manifold generated by all exponential vectors, $\Gamma(\sqrt{\Lambda})$ being the second quantization of the positive operator $\sqrt{\Lambda}$ in $\CH$, $a_{r}, 1\leq r \leq n$ are the basic annihilation operators corresponding to the $n$ different modes in $\GH$.
\item \label{item:51} Mean zero pure gaussian states are parametrized by a complex symmetric matrix  $A$ of order $n$ and general mean zero gaussian states are  parametrized by a pair of $n\times n$ complex matrices $(A,\Lambda)$, where $A$ is  symmetric and $\Lambda$ is a positive contraction.   An explicit particle basis expansion of an arbitrary mean zero  pure gaussian state vector along with a density matrix formula for a general mean zero gaussian state is  obtained in terms of the $\CE_2(\CH)$-parameters.  
\end{enumerate}
Some of the proofs spill over to other sections; 
items \ref{item:47}, \ref{item:48} and \ref{item:49} above are proved in Section \ref{sec:semigroup} while \ref{item:50} is obtained in Section \ref{sec:positive-operators} and \ref{item:51}  in Section  \ref{sec:arch-gauss-stat}.

 Section \ref{sec:gaussian} is devoted entirely to gaussian states. Suppose $\rho$ is a gaussian state with $\CE_2(\CH)$-parameters $(c,\bm{\alpha},A,\Lambda)$ and position-momentum parameters $\m, S$ where $\m$ is the mean annihilation vector in $\CH$ and $S$ is the $2n\times 2n$ real covariance matrix. Then formulae for $\m$ and $S$ in terms of $(\bm{\alpha}, A, \Lambda)$ and vice versa are obtained. This also shows that the scalar parameter $c$ is a function of $(\bm{\alpha}, A, \Lambda)$. As a corollary of the results described above, we obtain the uncertainty relation $S+i/2J\geq 0$ in terms of the $\CE_2(\CH)$-parameters. Furthermore, we prove the necessary and sufficient conditions on a pair $(A,\Lambda)$ consisting of a complex symmetric matrix $A$ and a positive matrix $\Lambda$ to be the $\CE_2(\CH)$-parameters of a gaussian state. Let  $C:\BC^n\rightarrow \BC^n$ denote the complex conjugation map $\z\mapsto\bar{\z}$, it is shown that a pair $(A,\Lambda)$ as described above are the $\CE_2(\CH)$-parameters of a gaussian state if and only if 
\begin{align*}
I-\Lambda -2AC>0,
\end{align*} in the sense of positive definiteness of real linear operators on $\BC^n$. Matrix version of the  inequality above is provided in Theorem \ref{sec:gauss-stat-uncert-6}. This is equivalent to the uncertainty relations in terms of the new parametrization.

Section \ref{sec:positive-operators} is dedicated to the study of positive operators in $\CE_2(\CH)$. The main result here is  item \ref{item:50} above.

In Section \ref{sec:arch-gauss-stat}, we prove item \ref{item:51} described above. Furthermore, the architecture of a gaussian state is described. 
If $U$ is an arbitrary unitary matrix operator in $\CH$ and $\Gamma(U)$ its second quantization in $\GH$, then the transformed gaussian state $\rho'= \Gamma(U)\rho\Gamma(U)^{\dagger}$ has $\CE_2(\CH)$-parameters given by $(c,U\bm{\alpha}, UAU^T, U\Lambda U^{\dagger})$. This shows that a conjugation by a Weyl operator followed by a second quantized unitary operator transforms $\rho$  to a mean zero gaussian state of the form $\rho(A, D_{\bm{\lambda}})$ with $\CE_2(\CH)$-parameters $(c, 0, A, D_{\bm{\lambda}})$, where the symmetric matrix $A$ can be different from the one we started with and $D_{\bm{\lambda}}$ is a diagonal matrix with diagonal entries $\lambda_1\geq \lambda_2\geq\cdots\geq \lambda_n\geq 0$. Furthermore, the pair $(A,0)$ is the $\CE_2(\CH)$-parameters of a mean zero pure gaussian state $\rho(A,0) = \ketbra{\psi_A}$. Remarkably, the density matrix of $\rho(A,D_{\bm{\lambda}})$ admits a canonical expansion in terms of $\ketbra{\psi_A}$ and $\bm{\lambda}.$ This completes the architecture of an arbitrary gaussian state.

In Section \ref{sec:gps-particle}, we  study some important examples of gaussian states using the new parametrization proposed in this paper. A whole class of completely entangled $n$-mode pure gaussian states is constructed. This yields examples of such entangled states which are also
invariant under the action of the permutation group $S_n$ on the set of all the $n$ modes. 

In Section \ref{sec:tomography}, we show how the tomography of an unknown gaussian state $\rho$ is essentially the  tomography of the truncated finite level state
 $(\tr \rho \mathcal{P})^{-1}\mathcal{P}\rho\mathcal{P}$, where $\mathcal{P}$ is the orthogonal projection onto the subspace $\BC\oplus \CH \oplus \mathcal{H}^{{\small{\text{\textcircled{s}}}}^{2}}$ in $\GH$.

Finally, a note on the notational conventions used in this paper.  We consider all vectors in $\BC^{n}$ or $\BZ_+^n$ as column vectors and the notation with round bracket $(z_1,z_2,\dots,z_n)$ is often used to denote the column vector $[z_1,z_2,\dots,z_n]^T$.  Bold letters like $\x,\y,\z$ etc. are used to indicate vectors in euclidean spaces $\BR^n$ and $\BC^n$.  Similarly, when we use the multi-index notation, bold letters like $\bk, \bt, \bs,$ etc. denote vectors in $\BZ_+^n$, meaning,  all  entries in the vector are non-negative integers.  In the multi-index convention, if $\bk=(k_1,k_2,\dots,k_n), \bm{\ell} =(l_1,l_2,\dots,l_n)\in \BZ_+^n$ and $\z= (z_1,z_2,\dots,z_n) \in \BC^n$, then
\begin{align*}
\label{eq:112}
&\cv{k}! = k_1!k_2!\cdots k_n!, \phantom{...}\abs{\bk} = k_1+k_2+\cdots +k_n, 
\phantom{...}\cv{z^k} = z_1^{k_1}z_2^{k_2}\cdots z_n^{k_n},\\ 
 &\bk\leq \bm{\ell}, \textnormal{ if } k_j\leq l_j, \forall j, \phantom{...} \binom{\bbell}{\bk} = \begin{cases}
  \frac{\bbell !}{\bk ! (\bbell-\bk)!} & \textnormal{ if } \bk\leq \bm{\ell}\\
  0 &\textnormal{ otherwise,} 
 \end{cases} \\&\bk \wedge \bm{\ell} = (k_1\wedge l_1,k_2\wedge l_2,\dots, k_n\wedge l_n), \textnormal{ where } k_j\wedge l_j=\min\{k_j,l_j\}, 1\leq j\leq n. 
\end{align*} 

Furthermore,  for any Hilbert space $\CH$ (finite or infinite dimensional, real or complex) with scalar product $\braket{\cdot}{\cdot}$, in the space of all bounded selfadjoint operators as well as real symmetric or Hermitian matrices, introduce the partial ordering $A\leq B$ to imply that $\mel{u}{B-A}{u}\geq 0$, for all $u\in \CH$. 
If $0\leq B$, we say that $B$ is a positive operator (or matrix, accordingly). We write $A<B$, if $A\leq B$ and $B-A$ has $0$ null space. If $0<B$ we say that $B$ is a strictly positive operator (or matrix). The symbols $\geq$ and $>$ are also used with obvious meanings arising from the discussion above.

\section{Exponential Vectors, Coherent States and the Klauder-Bargmann Isometry}\label{sec:2}
\label{sec:fock-space}

In this section and the next, we review the basics of the subject under consideration and fix some notations. All statements and proofs in these sections except possibly those concerning \emph{generating functions} [cf. (\ref{eq:gen-funct})] are well known. While we strive to keep the exposition as much self contained as possible, those results which are readily available in references \cite{Par12, Par10, serafini2017quantum} are stated only. 

Let $\GH$ denote the boson (symmetric) Fock space over a finite dimensional complex Hilbert space $\CH$ of dimension $n$. Choose and fix an  orthonormal basis, $\{e_j| 1\leq j\leq n\}$ in $\CH$ and identify $\CH$ with $\BC^n$, such that $z = \sum_{j=1}^nz_je_j \in \CH$ is identified with $\z = (z_1,z_2, \dots, z_n) \in \BC^n$. Then $\GH = \Gamma(\BC^n)$ = $\otimes_{j=1}^n\Gamma(\mathbb{C}e_{j})$. Let  $\mathbf{k} = (k_1,k_2,\dots,k_n)$, with nonnegative integer entries  $k_j, 1\leq j\leq n$. Denote by $\ket{k_1, k_2,\dots, k_n}$, the normalized symmetric tensor product of $e_r$ taken $k_r$ times, $r =1, 2, \dots, n$ so that $\ket{k_1, k_2,\dots, k_n}$ is a unit vector in the subspace $\mathcal{H}^{\small{\text{\textcircled{s}}}^{\abs{\bk}}}$ of $\GH$. 
Write \begin{equation*}
\label{eq:85}
\ket{\mathbf{k}} 
=\ket{k_1, k_2,\dots, k_n}.  
\end{equation*} 
The set $\{\ket{\mathbf{k}}| \cv{k}\in \BZ_+^n \}$ is a complete orthonormal basis for $\GH$, called a \emph{particle basis} with reference to the choice  $\{\ket{e_j}| 1\leq j\leq n\}$ in $\CH$. When $\CH = \BC^n$ we choose its canonical basis so that $e_j$ is the column vector with 1 in the position $j$ and $0$ elsewhere. The Fock space $\GH$ as well as $\GCn$ is also called an $n$-mode Fock space describing a quantum system of an arbitrary finite number of boson particles but in $n$ different modes. Any unit vector $\ket{\phi}$ as well as the corresponding $1$-dimensional projection $\ketbra{\phi}{\phi}$ is called a \emph{pure state}. The pure state $\ket{\mathbf{k}}$, thus defined is understood to be a state in which there are $k_r$ particles in the $r$-th mode for $r=1, 2, \dots, n$. For any $z\in \GH$, define $\ket{e(z)}, \ket{\psi(z)}$ in $\GH$ respectively by  \begin{align*}
    \ket{e(z)} &= \oplus_{k=0}^{\infty}\frac{z^{\otimes^k}}{\sqrt{k!}}, \\
    \ket{\psi(z)} & =  e^{-\norm{z}^2/2}\ket{e(z)}.
\end{align*} 
Then \begin{align*}
    \braket{e(z)}{e(z')} & = e^{\braket{z}{z'}},\\
    \norm{\psi(z)}& =1.
\end{align*}
We call $\ket{e(z)}$ and  $\ket{\psi(z)}$ respectively the \emph{exponential vector} and \emph{ coherent state} with parameter $z$;  $\ket{e(0)}$ is known as the \emph{vacuum state} and is denoted by $\ket{\Omega}$. Then $\ket{\psi(z)}$ is a superposition of finite particle states $\{\ket{\mathbf{k}}|\mathbf{k}\in \BZ_+^n\}$,  $\ket{\mathbf{k}}$ having probability amplitude $\Pi_r z_r^{k_r}/\sqrt{\Pi_r k_r!}$. 
Using the multi-index conventions described in Section \ref{sec:introduction}, we have
\begin{equation}\label{eq:2.kb}
    \ket{e(\z)} = \sum\limits_{\cv{k}\in \BZ^n_+} \frac{\cv{z^k}}{\sqrt{\cv{k}!}}\ket{\cv{k}} = \sum\limits_{0\leq j < \infty}^{}\sum\limits_{\abs{\bk} = j}^{} \frac{\cv{z^k}}{\sqrt{\cv{k}!}}\ket{\cv{k}}
\end{equation}
and for the particle basis measurement in the coherent state $\ket{\psi(z)}$, the probability of observing $k_r$ particles in the $r$-th mode for $1\leq r \leq n$ is \[e^{-\abs{\z}^2}\frac{\abs{\z^{\mathbf{k}}}^2}{\mathbf{k}!} = \Pi_{r=1}^n e^{-\abs{z_r}^2} \frac{\abs{z_r^{k_r}}^2}{k_r!}. \]
In other words, number of particles in different modes have independent Poisson distributions with mean values $\abs{z_r}^2, r= 1, 2, \dots, n$. Later, in our exposition, we shall meet $\ket{\psi(z)}$ as an example of a pure gaussian state.

 We now recall some of the well-known properties of exponential vectors and coherent states:

\begin{enumerate}
    \item The map $\cv{z}\mapsto \ket{e(\cv{z})}$ from $\BC^n$ into $\GCn$ is analytic where as $\cv{z}\mapsto \ket{\psi(\cv{z})}$ is real analytic.
    \item \label{item:li-total} For any nonempty finite set $F\subset \CH$, the set $\{\ket{e(z)}| z\in F\}$ is linearly independent and the linear span of all coherent states is dense in $\GH$.  In other words, the coherent states form a linearly independent and \emph{total} set.
    \item \label{item:1.kb} Let $\CH = \oplus_j S_j$, where $S_j$'s are mutually orthogonal subspaces. Furthermore, let $P_j$ denote the orthogonal projection of $\CH$ onto $S_j$. Then the map \begin{equation*}
        \ket{e(z)} \mapsto \otimes_j\ket{e(P_jz)}
    \end{equation*}
    extends to an isomorphism between $\GH$ and $\otimes_j \Gamma (S_j)$. In particular, if $\{v_1, v_2, \dots, v_n\}$ is any orthonormal basis of $\CH$ and $z = \sum_j z_jv_j$, then $\ket{\psi(z)} = \otimes_j \ket{\psi(z_jv_j)}$. In other words, in any orthonormal basis of $\CH$, any coherent state $\ket{\psi(z)}$ can be viewed as a product state.
    \item \label{finite-mode-l2-identifi} Suppose for some function $f\in L^2(\BR)$, $\int_{\BR}\overline{f(x)}e^{-\frac{1}{2}x^2+ux} \dd x = 0, \forall u\in \BC$, then this implies, in particular, that the Fourier transform of the function $\bar{f}e^{-\frac{1}{2}x^2}$ vanishes and so does $f$. Hence the set of functions $\{e^{-\frac{1}{2}x^2+ux}| u\in \BC\}$ is total in $L^{2}(\BR)$. Putting 
       \begin{equation}\label{eq:defn-gu}
        g_u(x) =\pi^{-1/4}e^{-\frac{1}{2}(x^{2}+u^2) + \sqrt{2}ux}, \forall x \in \mathbb{R},
    \end{equation}
we conclude that the set  $\{g_u|u\in \BC\}$ is total in $L^2(\BR)$.
 
When   $\CH$ is one dimensional, it may be identified with $\BC$ and then the map \begin{equation}\label{eq:eu-goesto-gu}
        e(u) \mapsto g_u
    \end{equation}
    is scalar product preserving between total sets in $\GH$ and $L^2(\BR)$. So it extends uniquely to an isomorphism between $\GH$ and $L^2(\BR)$. In general, when $\CH$ is $n$-dimensional, by Property \ref{item:1.kb} above, we see that $\GH$ is isomorphic to $L^2(\BR^n)$ via the mapping \begin{equation}\label{eq:eu-goesto-gu-2}
        e(u) \mapsto \otimes_jg_{u_j}
    \end{equation}
    where $u=\oplus_{j=1}^{n} u_je_j$.
    \item \label{bargmann-transform} \textbf{Klauder-Bargmann Isometry}. The map $\ket{\phi} \mapsto \pi^{-n/2}\braket{\psi(\cv{z})}{\phi}, \cv{z}\in \BC^n$ 
is an isometry from $\GH$ into $L^2(\BC^n)$, where $\BC^n$ is equipped with the $2n$-dimensional Lebesgue measure.
\end{enumerate}

 We indicate a proof of Property \ref{bargmann-transform}:
 \begin{proof}
   Since $\GCn$ is the $n$-fold tensor product of copies of $\GC$, which, by definition, is $\ell^2(\BZ_+)$ and by Property \ref{item:1.kb}, $\ket{\psi(\z)} =  \otimes_{j=1}^n\ket{\psi(z_j)}$, it is enough to prove Property \ref{bargmann-transform} when $n = 1$. In this case, \[\psi(z) = \{e^{-|z|^2/2}\frac{z^k}{\sqrt{k!}}, k = 0, 1, 2, \dots\}, z \in \BC.\] Let $\phi$ and $\phi'$ be two sequences given by $\{a_j\}$ and $\{b_j\}$, $j\in \BZ_+$. Putting $z= re^{i\theta}$, and using polar co-ordinates along with Parseval's identity we get
\begin{align*}
\frac{1}{\pi}\int\limits_{\BC}^{}\overline{\braket{\psi(z)}{\phi}}\braket{\psi(z)}{\phi'} \dd z & = 2\int\limits_{0}^{\infty}e^{-r^2}\int\limits_{0}^{2\pi} \overline{\left(\sum\limits_{j=0}^{\infty}\frac{a_jr^j}{\sqrt{j!}}e^{ij\theta}\right)}\left(\sum\limits_{j=0}^{\infty}\frac{b_jr^j}{\sqrt{j!}}e^{ij\theta}\right)\frac{\dd \theta}{2\pi} r \dd r\\
& =  \int\limits_{0}^{\infty}\sum\limits_{j=0}^{\infty}\frac{\bar{a}_jb_j}{j!}r^{2j}e^{-r^2}2r \dd r\\
&= \sum\limits_{j=0}^{\infty}\bar{a}_jb_j = \braket{\phi}{\phi'}.
\end{align*}
The last line above is obtained by integrating term by term and using the fact that $\int_0^\infty x^je^{-x}\dd x = j!$.
 \end{proof}
   We have a few important corollaries from the Klauder-Bargmann isometry. Before that, let us recall the weak operator integrals which we use in this article.
   Let $\left( X, \mathscr{F}, \mu \right)$ be a measure space and $\mathcal{K}$ a complex separable Hilbert space. $A \colon X \to \mathscr{B} (\mathcal{K})$ is said to be \textit{weakly measurable}  if the function $A_{\phi, \phi'} \colon X \to \mathbb{C}$ defined by $A_{\phi,\phi'}(x):= \left\langle \phi| A\left( x \right)|\phi' \right\rangle$ is measurable for every $\phi,\phi' \in \mathcal{K}$. A weakly measurable $A$ is said to be \textit{weakly integrable} (or just integrable when it is clear from the context) if for some $M \geq 0$, 
\begin{equation}\label{eq:sec:w-op-int-1} 
\int\limits_{X} \left| \left\langle \phi| A(x)|\phi' \right\rangle \right| \mu(\dd x) \leq M\|\phi\|\|\phi'\|.
\end{equation}
In this case we can make sense of $\int_{X}A(x)\mu(\dd x)$ 
as an element of $\mathscr{B}(\mathcal{K})$ using the  equation
\begin{equation}
\label{eq:sec:w-op-int-2} 
\langle \phi|\int\limits_{X}A(x)\mu(\dd x)| \phi' \rangle = \int\limits_{X}\left\langle \phi|A(x)|\phi'\right\rangle \mu(\dd x),  \forall \phi,\phi' \in \mathcal{K}.
\end{equation}
Existence and uniqueness of $\int_{X}A(x)\mu(\dd x)$ are given by the Riesz-representation theorem for linear functionals.
   
\begin{cor} \begin{enumerate}
    \item \emph{[Klauder-Bargmann formula]} The coherent states yield a resolution of identity $I$ into $1$-dimensional projections:  
\begin{equation}\label{eq:8-lec5}
      \frac{1}{\pi^n}\int \limits_{\mathbb{C}^{n}} \ketbra{\psi(\cv{z})}{\psi(\cv{z})} \dd  \cv{z} = I, 
  \end{equation}
 where the left hand side integral is a weak operator integral with respect to the $2n$-dimensional Lebesgue measure on $\BC^n$. In particular, for any element $\ket{\phi}$ in $\GH$ the following holds: 
\begin{equation}
\label{eq:3}
\ket{\phi} = \frac{1}{\pi^n}\int\limits_{\BC^{n}} \braket{\psi(\z)}{\phi}\ket{\psi(\z)}\dd \z,
\end{equation}
which has the interpretation that $\{\ket{\psi(\z)}|\z\in \CH\}$ is an 'overcomplete basis' for $\GH$.
  \item Any bounded operator $Z$ admits the representation 
  \begin{equation}
      Z = \frac{1}{\pi^n}\int \limits_{\BC^n} Z\ketbra{\psi(\cv{z})}{\psi{(\cv{z})}}\dd \cv{z}.
  \end{equation}
  \item The positive operator valued measure (POVM) $m$ defined by \begin{equation*}
      m(E) := \frac{1}{\pi^n}\int \limits_{E} \ketbra{\psi(\cv{z})}{\psi(\cv{z})} \dd  \cv{z},
  \end{equation*}
E a Borel set, yields a $\BR^{2n}$-valued continuous measurement in $\GH$.
\end{enumerate}
 
\end{cor}

\begin{rmk}
  The formula in (\ref{eq:8-lec5}) was first discovered in the present form by Klauder in Page 125-126 of \cite{KLAUDER1960123}, with a  heuristic proof. A rigorous proof of this first appeared in Page 194 of \cite{Barg1961}, where he proved it for a slightly different version of exponential vectors called \emph{principal vectors} in the Segal-Bargmann space. Equation (\ref{eq:8-lec5}) later appeared separately in the works of Glauber (again with a heuristic proof \cite{Glauber1963feb, Glauber1963sep}) and Sudarshan (who refers to Bargmann \cite{Sudarshan1963apr, Sudarshan-Klauder}) and was used by them to prove various results in quantum optics including the well-known Glauber-Sudarshan P representation. We call (\ref{eq:8-lec5}), \emph{Klauder-Bargmann formula}.
\end{rmk}

 Now we turn our attention to the description of an arbitrary bounded operator on $\GH$ in the particle basis.  Any bounded operator $Z$ in $\GH$, admits the following matrix representation in the particle basis: \begin{equation}\label{eq:kb-particle-1}
    Z = \sum\limits_{\br, \bs \in \BZ^n_+} Z_{\br\bs}\ketbra{\br}{\bs},
\end{equation}
where $Z_{\br\bs} = \mel{\br}{Z}{\bs}$. Define the \emph{generating function} of the operator $Z$, 
 \begin{equation}\label{eq:gen-funct}
 G_Z(u,v):= \mel{e(\bar{u})}{Z}{e(v)},
 \end{equation} where $\bar{u}$ is understood using the identification of $\CH$ with $\BC^n$. If $\bu = (u_1,u_2,\dots,u_n)$ and $\bv= (v_1,v_2,\dots, v_n)$, then by (\ref{eq:2.kb}), 
$G_Z(u,v)$ is a power series in the $2n$ variables $u_1,u_2,\dots,u_n,v_1,v_2,\dots,v_n$: 
 \begin{equation}\label{eq:4}
     G_Z(u,v) = \sum\limits_{0\leq k,\ell < \infty}^{}\sum\limits_{\underset{\abs{\bs} = \ell}{\abs{\br} = k}}^{}\frac{\cv{u^rv^s} }{\sqrt{\cv{r!s!}}}\mel{\cv{r}}{Z}{\cv{s}}.
 \end{equation}
The constant term in the power series above is $\mel{\cv{0}}{Z}{\cv{0}} = \mel{\Omega}{Z}{\Omega}$. We shall adopt the following notations to identify terms up to the second degree,
\begin{align}
\label{eq:126}
\ket{\chi_j}&:=|0, \cdots, 0,\underset{\mathclap{\substack{\uparrow\\ j\textnormal{-th}}}}{1},0,\cdots,0\rangle,
&\ket{\chi_{ij}}:=|0, \cdots, 0,\underset{\mathclap{\substack{\uparrow\\ i\textnormal{-th}}}}{1},0,\cdots, 0,\underset{\mathclap{\substack{\uparrow\\ j\textnormal{-th}}}}{1},0,\cdots,0\rangle,\nonumber\\
\ket{\chi_{jj}}&:=|0, \cdots, 0,\underset{\mathclap{\substack{\uparrow\\ j\textnormal{-th}}}}{2},0,\cdots,0\rangle,      
&1\leq i,j \leq n, i\neq j. \phantom{...................................}
\end{align}
It may be noticed that $\ket{\chi_{ij}}=\ket{\chi_{ji}}$ correspond to the vector $\br\in \BZ_+^n$ with $1$ at both  $i$-th and $j$-th positions and $0$ everywhere else.
The linear terms on the right side of (\ref{eq:4}) are 
\begin{equation}
\label{eq:5}
\sum\limits_{j=1}^nu_j\mel{\chi_j}{Z}{\Omega} + \sum\limits_{j=1}^nv_j\mel{\Omega}{Z}{\chi_j} = \bu^T\bm{\lambda}_Z+\bv^T\bm{\mu}_Z,
\end{equation} where $\bm{\lambda}_Z$ and  $\bm{\mu}_Z$ are vectors in  $\BC^n$ with $j$-th coordinate $\mel{\chi_j}{Z}{\Omega}$ and $\mel{\Omega}{Z}{\chi_j}$ respectively, $j= 1, 2, \dots, n$.
We call $\bm{\lambda}_Z$ and $\bm{\mu}_Z$ the $1$-\emph{particle annihilation } and \emph{creation amplitude vector} respectively of $Z$. The quadratic terms in the power series are 
\begin{equation}
\label{eq:16}
\begin{split}
\sum\limits_{j=1}^n \frac{u_j^2}{\sqrt{2}}\mel{\chi_{jj}}{Z}{\Omega}+\sum\limits_{j< k} u_ju_k \mel{\chi_{jk}}{Z}{\Omega}+\sum\limits_{j,k = 1}^n u_jv_k \mel{\chi_j}{Z}{\chi_k}+\sum\limits_{j=1}^{n}\frac{v_j^2}{\sqrt{2}}\mel{\Omega}{Z}{\chi_{jj}}\\+\sum\limits_{j< k} v_jv_k \mel{\Omega}{Z}{\chi_{jk}}= \bu^TA_Z\bu+ \bu^T\Lambda_Z\bv+  \bv^TB_Z\bv,
\end{split}
\end{equation} 
where $A_{Z} = [\alpha_{jk}]$, $\Lambda_Z = [\lambda_{jk}]$ and $B_Z=[\beta_{jk}]$ are  in $M_n(\BC)$ with,
\begin{equation*}\begin{split}
\alpha_{jk} = \begin{cases}\phantom{.}\frac{1}{2}\mel{\chi_{jk}}{Z}{\Omega}, & j\neq k,\\\frac{1}{\sqrt{2}} \mel{\chi_{jj}}{Z}{\Omega}, & j=k\end{cases},\lambda_{jk} = \mel{\chi_j}{Z}{\chi_k}, \beta_{jk} = \begin{cases}\phantom{.}\frac{1}{2}\mel{\Omega}{Z}{\chi_{jk}}, & j\neq k,\\\frac{1}{\sqrt{2}} \mel{\Omega}{Z}{\chi_{jj}}, & j=k\end{cases}\end{split}.
\end{equation*} We call $A_Z$, $\Lambda_Z$ and $B_Z$  the $2$-\emph{particle annihilation, exchange} and \emph{creation  amplitude matrix} respectively of $Z$. 
Notice that $A_Z$ and $B_Z$ are complex symmetric matrices by construction.
\section{Weyl Operators, quantum characteristic function and  the Wigner  Isomorphism}\label{sec:weyl-oper-wign}

As mentioned at the beginning of the previous section, we continue with an exposition of fundamental notions of our subject in this section also.

The correspondence $\ket{\psi(z)}\mapsto  e^{-i \im \braket{u}{z}}\ket{\psi(u+z)}, z\in \CH$
is a scalar product preserving map for any fixed $u\in \CH$. Since the coherent states constitute a total set (Property \ref{item:li-total}, Section \ref{sec:fock-space})  in $\GH$, it follows that there exists a unique unitary operator $W(u)$  on $\GH$ satisfying the relation
 \begin{equation}
        W(u) \ket{\psi(z)} = \exp{-i \im \braket{u}{z}}\ket{\psi(u+z)}, z\in \CH. 
     \end{equation}
We call $W(u)$ the \emph{Weyl operator} at  $u \in \CH$. It is also known as the \emph{displacement operator} at $u$. The Weyl operators obey the multiplication  relations 
\begin{align}\begin{split}
W(u)W(v) &= \exp(-i \im \left\langle u|v \right\rangle)W(u+v),\forall u, v \in \CH,\\
W(u)W(v)& = \exp{-2i\im\braket{u}{v}}W(v)W(u), u, v \in \CH. 
\label{eq:2.kb1}
\end{split}
\end{align}
Equations in (\ref{eq:2.kb1}) are known as \emph{Weyl commutation relations} or canonical commutation relations (\emph{CCR}) and the $C^{*}$-algebra generated by the Weyl operators denoted $CCR(\CH)$ is called the $CCR$-algebra. We recall a few basic properties of the Weyl operators. 
\begin{enumerate}
\item\label{item:1.kb2} The map $u\mapsto W(u)$ is a \emph{strongly continuous, projective, unitary} and \emph{irreducible}  representation of the additive group $\CH$ known as the Weyl representation in $\GH$.  Furthermore, it follows from the irreducibility that the von Neumann algebra generated by Weyl operators is all of $\B{\GH}$, i.e., \begin{equation} \label{eq:item:1.kb2}
    \overline{CCR(\CH)}^{\textnormal{ sot}} = \B{\GH},
\end{equation}
where $^{\overline{\phantom{.......}}{\textnormal{ sot}}}$ indicates the closure in the strong operator topology.
\item\label{item:2.kb2} The  Weyl representation enjoys the factorizability property: if $\CH = \CH_1 \oplus \CH_2 \oplus \cdots \oplus \CH_k$, $u = u_1\oplus u_2 \oplus \cdots \oplus u_k, u_{j}\in \CH_j$ for each $j$, then $W(u) = W(u_1)\otimes W(u_2) \otimes \cdots \otimes W(u_k)$.

\item \label{item:5.kb2} In this item we shall  recall some facts about certain well known unbounded operators in quantum theory. We shall not discuss matters concerning their domains and refer to \cite{Par12} for details.  For every fixed $u\in \CH$, the set $\{W(tu): t\in \mathbb{R}\}$ is a strongly continuous, one parameter unitary group and hence has the form 
\begin{equation}\label{prel-eq:5}
    W(tu) = e^{-it \sqrt{2} p(u)}, t \in \mathbb{R}, u\in \mathcal{H},
\end{equation}
where $p(u)$ is a self-adjoint operator in $\GH$. Define 
\begin{align*}
q(u)&=p(-iu),\\
a(u)&= \frac{1}{\sqrt{2}}( q(u)+ip(u)),\\ 
 a^{\dagger}(u)&= \frac{1}{\sqrt{2}}( q(u)-ip(u)).
\end{align*}
Then $a(u)$ and $a^{\dagger}(u)$ are the well-known annihilation and creation operators at $u$. Observe that 
\begin{align}
  \label{eq:a-action-exp}  a(u)e(v)&= \braket{u}{v}e(v),\\
    a^{\dagger}(u)e(v) & =  \sum\limits_{k=1}^{\infty}\frac{1}{\sqrt{k!}}\sum\limits_{r=0}^{k-1}v^{\otimes^r}\otimes u \otimes v^{\otimes^{k-r-1}}.\numberthis \label{eq:a-daggar-action-exp} 
\end{align}
 Changing $v$ to $sv,s\in \mathbb{R}$, and identifying coefficients of $s^k$ on both sides of 
  equations above, we get
\begin{align}\label{eq:35.kb}
    a(u)v^{\otimes^k}& = \sqrt{k}\braket{u}{v}v^{\otimes^{k-1}},\\
\label{eq:36.kb}
    a^{\dagger}(u)v^{\otimes^k}& = \frac{1}{\sqrt{k+1}}\sum\limits_{r=0}^{k}v^{\otimes^r}\otimes u \otimes v^{\otimes^{k-r}}  \end{align}
for all $v \in \CH, k\in \mathbb{N}$. Furthermore, it may  be noted from (\ref{eq:a-daggar-action-exp}) that     \begin{equation}\label{a-a-dagger-eq:4}
    a^{\dagger}(u)e(v)= {\frac{d}{ds}}_{|_{s=0}}e(v+su).
    \end{equation}


When $\CH = \BC^n = \BR^n\oplus i\BR^n$ the families $\{q(\x)|\x\in \BR^n\}$ and $\{p(\x)|\x\in \BR^n\}$ are commuting families of self-adjoint operators  or observables and the CCR in (\ref{eq:2.kb1}) becomes
\begin{equation*}
  [q(\x), p(\y)] = i\x^{T}\y, \forall \x, \y \in \BR^n.
\end{equation*}
These are the well-known \emph{Heisenberg commutation relations}, again called CCR. It is also expressed as
\begin{align*}
   [a(u),a(v)] & = 0,\\
   [a^{\dagger}(u),a^{\dagger}(v)] & = 0,\\
  [a(u), a^{\dagger}(v)] &= \braket{u}{v}, \forall u,v \in \CH \textnormal{ or } \BC^n.
\end{align*}
It may also be noted that the map $u\mapsto a(u)$ and $u\mapsto a^{\dagger}(u)$ are respectively antilinear and linear in the variable $u$. 
Going back to Weyl operators we have
\begin{equation*}
  W(u) = e^{(a^{\dagger}(u)-a(u))}, \forall u \in \CH  \textnormal{ or } \BC^n.
\end{equation*}
    Write
\begin{align}
  p_j&= p(e_j),& q_j &= q(e_j)= -p(ie_j)\\
a_j&=a(e_j)=\frac{1}{\sqrt{2}}(q_j+ip_j) & a_j^{\dagger}&=a^{\dagger}(e_j) =\frac{1}{\sqrt{2}}(q_j-ip_j)
\end{align}
for each $1\leq j\leq n$. The operators $p_j$, $q_j$,  $a_j$ and $a_j^{\dagger}$ are respectively called the  \emph{\index{momentum}momentum}, \emph{\index{position}position}, \emph{\index{annihilation}annihilation} and \emph{\index{creation}creation} operators of the $j$-th mode. In particular, the observable $a_{j}^{\dagger}a_j$ is the \emph{number operator} describing the number of particles in the $j$-th mode. 

\item \label{item:3.kb2} \textbf{Stone-von Neumann Theorem}. If $\CK$ is a complex separable Hilbert space and $u\mapsto W'(u)$ is a strongly continuous, projective and  unitary representation of $\CH$ in $\CK$ satisfying the relations (\ref{eq:2.kb1}) with $W$ replaced by $W'$, then there exists a Hilbert space 
$\mathtt{k}$ and a unitary isomorphism $\Gamma: \CK\rightarrow \GH\otimes \mathtt{k}$ such that 
\begin{equation*}
  \Gamma W'(u)\Gamma^{-1} = W(u)\otimes I_{\mathtt{k}}, \forall u \in \CH,
\end{equation*}
where $I_{\mathtt{k}}$ is the identity operator in $\mathtt{k}$. In particular, if $W'$ is also irreducible then $\mathtt{k} = \BC$, the $1$-dimensional Hilbert space and $\Gamma$ is a unitary isomorphism from $\CK$ to $\GH$.
\item \label{item:4.kb2} Let $L$ be a  real linear transformation of $\CH$ 
satisfying  \[\im \braket{Lu}{Lv} =\im \braket{u}{v},\forall u,v\in \CH. \] Such a transformation is said to be symplectic. Define \[W_L(u)= W(Lu), u \in \CH.\] The map $u\mapsto W_{L}(u)$ is a strongly continuous, projective, unitary and irreducible representation of $\CH$ in $\GH$ obeying  (\ref{eq:2.kb1}). Hence by the Stone-von Neumann theorem in Property \ref{item:3.kb2} there exists a unitary operator $\GL$ in $\GH$ satisfying 
  \begin{equation}\label{shale-unitary-prop}
    \Gamma(L)W(u)\Gamma(L)^{-1} = W(Lu), \forall u \in \CH.
\end{equation}
Such a unitary operator $\GL$ is unique up to multiplication by a scalar of modulus unity. The operator $\GL$ is said to intertwine the representations $W$ and $W_L$.
\end{enumerate}

Let $\mathscr{B}_j(\GH) \subset \B{\GH}$, for $j=1$ and $2$ denote the ideal of trace class operators and  Hilbert-Schmidt operators respectively on $\GH$. Then $\Bo{\GH}$ is a Banach space with $\norm{\rho}_1 = \tr \sqrt{\rho^{\dagger}\rho}, \rho\in \Bo{\GH}$, $ \Bt{\GH}$ is a Hilbert space with scalar product $\braket{\rho_1}{\rho_2}_2 = \tr \rho_1^{\dagger}\rho_2$ and $\Bo{\GH}\subset\Bt{\GH}$  as a linear manifold. 
\begin{defn}\label{defn:qft} 
 If $\rho \in\Bo{\GH} $,  then the complex valued function
  \begin{equation*}
    \hat{\rho}(z) := \tr \rho W(z), z \in \CH
  \end{equation*}
is called the \emph{quantum characteristic function} (or \emph{Wigner transform}) of $\rho$.
\end{defn}

 We summarize a few properties of the quantum characteristic function: \begin{enumerate}
    \item \label{item:6.kb.2} The function $\hat{\rho}$ is  \emph{bounded} and \emph{continuous} on $\CH$.
    
    Since  $\Bo{\GH}$ is the predual of $\B{\GH}$ and $W(z)$ is a unitary operator,
     \[\abs{\tr \rho W(z)} \leq \|\rho\|_1.\]  The continuity of $\hat{\rho}$ follows from the strong continuity of the Weyl representation.
      
    \item \label{item:7.kb.2} The \emph{correspondence} $\rho \rightarrow \hat{\rho}$ is \emph{injective}. 
    
    Let $\rho_1,\rho_2 \in \Bo{\GH}$. The equation $\hat{\rho_1}=\hat{\rho_2}$ implies that $\tr(\rho_1 - \rho_2)W(z) = 0$ for all $z \in \CH$ and by (\ref{eq:item:1.kb2}),  $\tr(\rho_1 - \rho_2)X = 0$ for any $X\in \B{\CH}$. 
    \item \label{item:8.kb.2} The quantum characteristic function is \emph{factorizable}. 

    Indeed, for $\rho_j \in \Bo{\Gamma(\CH_j)}, j =1,2$, by property \ref{item:2.kb2} of Weyl operators,  \begin{equation*}
        \widehat{\rho_1\otimes\rho_2} (u\oplus v) = \hat{\rho_1}(u)\hat{\rho_2}(v). 
    \end{equation*}

    \item  \label{item:9.kb.2} A positive operator $\rho$ of unit trace in $\GH$ is called an $n$-mode \emph{state}. For such a state $\rho$, by Property \ref{item:5.kb2}, of Weyl operators, the function $\hat{\rho}(t\z), t\in \BR$ is the characteristic function of the probability distribution of the observable $-\sqrt{2}p(\z) = i(a(\z) - a^{\dagger}(\z)) = \sqrt{2}(q(\y)-p(\x))$ for any fixed $\z = \x+i\y$.
\item \textbf{Quantum Bochner Theorem} \cite{Srinivas-Wolf-75, Par10}. A complex valued function $f$ defined on $\CH$ is the quantum characteristic function of an $n$-mode state if and only if the following are satisfied:
  \begin{enumerate}
  \item $f(0) = 1$ and $f$ is continuous at $0$.
  \item The kernel $\mathscr{K}(z, w) = e^{i\im \braket{z}{w}}f(w-z)$ is positive definite.
  \end{enumerate}
\end{enumerate}

We now state the gaussian integral in a form which we frequently use in the rest of this paper and refer to Appendix A of \cite{Fol89} for a proof.
\begin{prop}[Gaussian integral formula]
\label{thm:gaussian-int}  Let $A$ be an $n\times n$ complex matrix such that $A$ is symmetric ($A = A^T$) and $\re A$ is a strictly positive  matrix. Then for any $\m\in \mathbb{C}^n$, \begin{equation}\label{eq:gif}
  \int\limits_{\mathbb{R}^n} \exp{-\x^TA \x + \m^T\x}\dd \x = \sqrt{\frac{\pi^n}{\det A}}\exp{\frac{1}{4}\m^TA^{-1}\m},
\end{equation}
where the branch of the square root is determined in such a way that $\det^{-1/2}{A}>0$ when $A$ is real and strictly positive.
\end{prop}
\begin{rmk}
  In the context of Proposition \ref{thm:gaussian-int}, it may be noted that a symmetric matrix $A$ with $\re A>0$ is invertible. Furthermore, $A=A^T$ implies that $\re\mel{\z}{A}{\z} = \mel{\z}{\re A}{\z}$, $\forall \z \in \BC^n$, now  $\re A>0$ implies that the eigenvalues $\lambda_1, \lambda_2, \dots, \lambda_n$ of $A$ have strictly positive real part. Then $\det^{-1/2}{A} = \Pi_{j=1}^n\lambda_j^{-1/2}$, where $\lambda_j^{-1/2}$ is the square root of $\lambda_j^{-1}$ with positive real part.
\end{rmk}
We defined the quantum characteristic function on the trace class ideal. Now we proceed to extend this definition to the Hilbert-Schmidt class in the same spirit as in the classical theory of Fourier transforms. Let $\mathcal{F} = \{\ketbra{e(u)}{e(v)}: u, v\in \CH\}$, then  $\mathcal{F} \subset \Bo{\GH} \subset \Bt{\GH}$. Since exponential vectors form a total subset (Property \ref{item:li-total}, Section \ref{sec:fock-space})  of $\GH$, $\mathcal{F}$ is a total set in $\Bt{\GH}$. The following example illustrates an important property of the quantum characteristic function of elements of $\mathcal{F}$.
\begin{eg}\label{sec:weyl-oper-wign-eg-qft}
 For $\bu, \bv\in \BC$, consider $\rho= \ketbra{e(\bu)}{e(\bv)} \in \B{\GC}$. Then 
\begin{align*}
\hat{\rho}(\z) &= \tr \ketbra{e(\bu)}{e(\bv)}W(\z)\\
&= \mel{e(\bv)}{W(\z)}{e(\bu)}\\
&= e^{\bar{\bv}\bu}e^{-\frac{1}{2}|\z|^2+\bar{\bv}\z-\bar{\z}\bu}. \numberthis \label{sec:weyl-oper-wign-eg-eq} 
\end{align*}
Thus  $\hat{\rho}\in L^1(\BC)\cap L^2(\BC)$. Since $L^2(\BC) = L^2(\BR)\otimes L^2(\BR)$  and (by Property \ref{finite-mode-l2-identifi} in Section \ref{sec:fock-space}) the set  $\{e^{-\frac{1}{2}x^2+\zeta x}| \zeta\in \BC\}$ is total in $L^{2}(\BR)$ it follows that $\{\hat{\rho}|\rho\in \mathcal{F}\}$ is a total set in $L^2(\BC)$.  
\end{eg}    
    \begin{thm}[\textbf{Wigner isomorphism}]\label{thm:wigner-iso}
       For $\rho \in \Bo{\GH}$, let $\mathbb{F}_n(\rho)$ be the function defined on $\BC^n$ such that \begin{equation}\label{eq:defn-Fn}
        \mathbb{F}_n(\rho)(\z) = \pi^{-n/2}\hat{\rho}(\z), \z \in \BC^n.
    \end{equation}
    Then $\mathbb{F}_n$ extends uniquely to a Hilbert space isomorphism denoted again by $\mathbb{F}_n$ from $\Bt{\GH}$ onto $L^2(\BC^n)$.
    \end{thm}
\begin{rmk} The identification $z\mapsto \z$ is  exploited in  (\ref{eq:defn-Fn}),
strictly speaking, it should be $\hat{\rho}(z)$ not  $\hat{\rho}(\z)$ on the right side of (\ref{eq:defn-Fn}).
\end{rmk}
    \begin{proof}
      First we prove the theorem when $n = 1$, i.e., $\CH = \BC$. Let $\rho_j = \ketbra{e(\bu_j)}{e(\bv_j)}$, $j = 1,2$. Using (\ref{sec:weyl-oper-wign-eg-eq}) and the gaussian integral formula (Proposition \ref{thm:gaussian-int}), we get
\begin{align*}
\frac{1}{\pi}\int\limits_{\BC}^{}\overline{\hat{\rho_1}(\z)}\hat{\rho_2}(\z)\dd \z & = e^{\braket{\bu_1}{\bv_1}+\braket{\bv_2}{\bu_2}}\frac{1}{\pi}\int\limits_{\BC}^{}\exp{-\abs{\z}^2 +\braket{\bv_2-\bu_1}{\z}+ \braket{\z}{\bv_1-\bu_2}}\dd \z\\
& = e^{\braket{\bu_1}{\bv_1}+\braket{\bv_2}{\bu_2}}e^{\braket{\bv_2-\bu_1}{\bv_1-\bu_2}}\\
&= e^{\braket{\bu_1}{\bu_2}+\braket{\bv_2}{\bv_1}}\\
& = \tr \rho_1^{\dagger}\rho_2.
\end{align*}
Now Example \ref{sec:weyl-oper-wign-eg-qft}, shows that $\mathbb{F}_1$ is a  scalar product preserving map between total sets and thus extends uniquely to a Hilbert space isomorphism. Hence $\BF_1^{\otimes{n}}$ is an isomorphism from $\Bt{\GH}$ onto $L^2(\BC^n)$. Observe that $\BF_1^{\otimes{n}}$ coincides with $\BF_n$ on $\Bo{\GH}$, hence $\BF_1^{\otimes{n}}$ extends $\BF_n$ initially defined on $\Bo{\GH}$ to $\Bt{\GH}$. 
    \end{proof}

\begin{cor}\begin{enumerate} 
    \item \label{cor:1} The map $e(u\oplus v) \mapsto \ketbra{e(u)}{e(\bar{v})}$  extends as an isomorphism $\eta_1$ from $\Gamma(\CH\oplus\CH)$ onto $\Bt{\GH}$.
    \item \label{cor:2} Let $u = \sum_j u_je_j \in \CH$ define $g_u =\otimes_{j}g_{u_j} $, where $g_{u_j} \in L^2(\BR)$ is as defined by equation (\ref{eq:defn-gu}). The map $\ketbra{e(u)}{e(\bar{v})}\mapsto g_u\otimes g_v$ extends as an isomorphism $\eta_2$ from $\Bt{\GH}$ onto $L^2(\BC^n) = L^2(\BR^n)\otimes L^2(\BR^n)$.
    \item \label{cor:3} Let $v=\sum_j v_je_j\in \CH$, define $\bar{v}:= \sum_j \bar{v_j}e_j$. The Wigner isomorphism satisfies $ \BF_n\left(\ketbra{e(u)}{e(\bar{v})}\right) = g_{u'} \otimes g_{v'},$
    where %
\begin{equation}\label{eq:S} \begin{pmatrix}
    u'\\v'
  \end{pmatrix} = S
  \begin{pmatrix}
    u\\v
  \end{pmatrix}, 
 S=   \frac{1}{\sqrt{2}}
  \begin{bmatrix}
    -I&I\\iI &iI
  \end{bmatrix},
\end{equation}
\end{enumerate}
\end{cor}
\begin{proof}
\ref{cor:1}. This follows from a direct computation showing that $\eta_1$ is scalar product preserving.

\ref{cor:2}.
We know from Property \ref{finite-mode-l2-identifi} in Section \ref{sec:fock-space} that the map $e(\bu)\mapsto g_{\bu}$  extends to a Hilbert space isomorphism, the required result follows from part \ref{cor:1}).

\ref{cor:3}. Again it is enough to prove this when $\CH=\BC$. We have
\begin{align*}
\BF_1(\ketbra{e(u)}{e(\bar{v})})(x,y) 
&=\pi^{-1/2} \tr \ketbra{e(u)}{e(\bar{v})}W(x+iy)\\
&=\pi^{-1/2} \bra{e(\bar{v})}W(x+iy)\ket{e(u)}\\
&=\pi^{-1/2}\braket{e(\bar{v})}{e^{-\frac{1}{2}(x^2+y^2)-(x-iy)u}e(u+x+iy)}\\
&=\pi^{-1/2} \exp{-\frac{1}{2}(x^2+y^2)-(x-iy)u+ v(u+x+iy)}\\
&=\pi^{-1/2}  \exp{-\frac{1}{2}(x^2+y^2)+(v-u)x +i(v+u)y +vu}\\
&=\pi^{-1/2} \exp{-\frac{1}{2}x^2+\left(\frac{v-u}{\sqrt{2}}\right)\sqrt{2}x - \frac{1}{2}\left(\frac{v-u}{\sqrt{2}}\right)^2 }e^{ \frac{1}{2}\left(\frac{v-u}{\sqrt{2}}\right)^2}\\
&\phantom{...}\times \exp{-\frac{1}{2}y^2+\left(i\frac{v+u}{\sqrt{2}}\right)\sqrt{2}y - \frac{1}{2}\left(i\frac{v+u}{\sqrt{2}}\right)^2 }e^{ \frac{1}{2}\left(i\frac{v+u}{\sqrt{2}}\right)^2}\times e^{vu}\\
&=g_{u'}(x)g_{v'}(y),
\end{align*}
where  \begin{equation}\label{eq:u'-v'}
     \begin{pmatrix}
    u'\\v'
  \end{pmatrix} = \frac{1}{\sqrt{2}}
  \begin{bmatrix}
    -1&1\\i &i
  \end{bmatrix}
  \begin{pmatrix}
    u\\v
  \end{pmatrix}.
\end{equation}
\end{proof}

\begin{rmks}
 \begin{enumerate}
    \item All the  isomorphisms above are described by the Figure  \ref{fig:id-mappings} via the mappings in Figure \ref{fig:id-actions}.

    \item The quantum characteristic function on $\Bt{\GH}$ can be viewed as the second quantization $\Gamma(S)$ on $\Gamma(\CH\oplus\CH)$, where $S$ is the unitary matrix given by (\ref{eq:S}). The eigenvalues of $S$ are $\lambda_1 = ie^{-i\pi/12}$
    and $\lambda_2=-e^{i\pi/12} $
    with  multiplicities $n$ each. Then $\lambda_j^{12} = -1, j=1,2$. Write  $\BF_n^{'} = \tilde{\Gamma}(S) =\eta_2^{-1}\BF_n $, then  $\BF_n^{'}: \Bt{\GH} \rightarrow \Bt{\GH} $ is an isomorphism and satisfies the property \begin{equation}\label{eq:fourier-isomo}
        (\BF_n^{'} )^{12} = -I.
    \end{equation}
    In the classical theory, the Fourier transform defined on $L^1(\BR^n)$ extends to a unitary $\BF$ on  $L^2(\BR^n)$, furthermore $\BF^2 = -I$. Equation (\ref{eq:fourier-isomo}) can be viewed as a noncommutative analogue of this fact.
\end{enumerate}
\end{rmks}
  \begin{figure}[t]
    \centering
    \begin{tikzcd}
\Gamma(\CH \oplus\CH) \arrow[r, "\eta_1"] \arrow[d, "\Gamma(S)"] 
& \Bt{\GH} \arrow[d, "\tilde{\Gamma}(S)" ] \arrow[r, "\eta_2"] \arrow[rd, "\BF_n"]& L^2(\BC^n) \arrow[d, "\hat{\Gamma}(S)"]\\
\Gamma(\CH \oplus \CH)  \arrow[r, "\eta_1"]
& \Bt{\GH} \arrow[r, "\eta_2"] & L^2(\BC^n)
\end{tikzcd}
    \caption{The maps $\tilde{\Gamma}(S)$ and $\hat{\Gamma}(S)$ are the corresponding compositions.}
    \label{fig:id-mappings}
\end{figure}

\begin{figure}[t]
    \centering
    \begin{tikzcd}
e(u\oplus v ) \arrow[r, "\eta_1"] \arrow[d, "\Gamma(S)"] 
& \ketbra{e(u)}{e(\bar{v})} \arrow[d, "\tilde{\Gamma}(S)" ] \arrow[r, "\eta_2"] \arrow[rd, "\BF_n"]& g_u\otimes g_v \arrow[d, "\hat{\Gamma}(S)"]\\
e(u'\oplus v') \arrow[r, "\eta_1"]
& \ketbra{e(u')}{e(\bar{v}')} \arrow[r, "\eta_2"] & g_{u'}\otimes g_{v'}
\end{tikzcd}
    \caption{$u'\oplus v' = S(u\oplus v)$ where $S$ as in equation (\ref{eq:S}).}
    \label{fig:id-actions}
\end{figure}

\begin{thm}[Inversion for quantum characteristic functions]
If $\rho \in \Bo{\GH}$ then,
 \begin{equation}
\label{eq:5.kb2}
\rho = \frac{1}{\pi^n}\int\limits_{\CH} \hat{\rho}(z)W(-z) \dd z,
\end{equation}
where the integral is a weak operator integral with respect to the $2n$-dimensional Lebesgue measure (inherited from $\BC^n$) on $\CH$ .
\end{thm}
\begin{proof}
If $\phi, \phi' \in \GH$, by Theorem \ref{thm:wigner-iso}
\begin{align*}
    \mel{\phi}{\rho}{\phi'}& = \tr \ketbra{\phi}{\phi'}^{\dagger}\rho = \int\limits_{\BC^n} \overline{\BF_n(\ketbra{\phi}{\phi'})(\z})\BF_n(\rho)(\z)\dd \z\\ 
    & = \frac{1}{\pi^n}\int\limits_{\CH}\mel{\phi}{W(-z)}{\phi'}\hat{\rho}(z) \dd z\\
    & = \mel{\phi}{\frac{1}{\pi^n}\int_{\CH}\hat{\rho}(z)W(-z) \dd z}{\phi'}.
\end{align*}
This is same as (\ref{eq:5.kb2}).
\end{proof}
We conclude this section with two results connecting the notion of generating function defined in Section \ref{sec:2} with the quantum characteristic function.
\begin{prop}\label{prop:qft-in-terms-generator}
Let $\rho \in \Bo{\GH}$, then the quantum characteristic function of $\rho$ can be expressed in terms of the generating function (equation (\ref{eq:gen-funct})) of $\rho$ as \begin{equation}
    \label{eq:qft-in-terms-generator}
    \hat{\rho}(u)= \frac{e^{-\frac{1}{2}\abs{u}^2}}{\pi^n}\int\limits_{\CH} \exp{-\abs{z}^2 - \braket{u}{z}}G_{\rho}(\bar{z}, u+z)\dd z.
\end{equation}
\end{prop}
\begin{proof}
By the Klauder-Bargmann formula,
\begin{align*}
    \hat{\rho}(u) &= \tr \rho W(u)  = \frac{1}{\pi^n}\int\limits_{\CH}\mel{\psi(z)}{\rho e^{-i \im \braket{u}{z}}}{\psi(u+z)}\dd z\\
        & = \frac{1}{\pi^n}\int\limits_{\CH} \exp{-\abs{z}^2-\frac{1}{2}\abs{u}^2-\braket{u}{z}}\mel{e(z)}{\rho}{e(u+z)}\dd z,
\end{align*}
which is same as (\ref{eq:qft-in-terms-generator}).
\end{proof}

The following lemma is a corollary of Klauder-Bargmann formula.
\begin{lem}\label{lem:semigr-ensur}
If $\rho$ is a trace class operator  then 
\begin{equation}
\label{eq:18}
\tr \rho=\frac{1}{\pi^n}\int\limits_{\CH}^{}\tr \rho \ketbra{\psi(z)}{\psi(z)} \dd z.
\end{equation} Furthermore, a positive operator $\rho$ is trace class if and only the quantity on the right side of the equation above is finite.
\end{lem}
\begin{proof} Assume that $\rho$ is trace class. 
  If $\rho = \ketbra{\phi}{\phi}$ then (\ref{eq:18}) is immediate from the Klauder-Bargmann isometry. If $\rho$ is a positive trace class operator then (\ref{eq:18}) follows from the preceding case by an application of the spectral theorem. Now the required result  follows from the fact that any trace class operator is a linear combination of four positive trace class operators. 

Now suppose $\rho\geq 0$. 
Let $\{f_k\}$ be any orthonormal basis of $\GH$, then we can interchange the summation and integral in the following computation because all the quantities involved are nonnegative,
\begin{align*}
\sum\limits_k^{}\mel{f_k}{\rho}{f_k}& = \sum\limits_k^{}\frac{1}{\pi^n} \int\limits_{\CH}^{}\mel{f_k}{\sqrt{\rho}\ketbra{\psi(z)}\sqrt{\rho}}{f_k}\dd z\\
&= \frac{1}{\pi^n}\int\limits_{\CH}^{}\sum\limits_k^{}\mel{f_k}{\sqrt{\rho}\ketbra{\psi(z)}\sqrt{\rho}}{f_k}\dd z\\
&= \frac{1}{\pi^n}\int\limits_{\CH}^{}\tr \rho \ketbra{\psi(z)}{\psi(z)} \dd z.
\end{align*}
Hence $\rho\in \Bo{\CH}$ if and only if the integral above converges.
\end{proof} 
\begin{thm}Let $\CH = \CH_0\oplus \CH_1$ with $\dim \CH_0 = n_0$  and $\dim \CH_1 = n_1$ so that $\GH = \mathfrak{h}_0\otimes \mathfrak{h}_1$ where $\mathfrak{h}_i = \Gamma(\CH_{i}), i= 0, 1$.
Let $\rho$ be any state on $\GH$ and let $\rho_{1-i} := \tr_{i}\rho \in \B{\mathfrak{h}_{1-i}}$ denote the $\mathfrak{h}_{1-i}$ marginal of $\rho$, $i = 0,1$. For any $u_0, v_0 \in \CH_0$, 
\begin{equation}
\label{eq:17}
\frac{1}{\pi^{n_1}}\int\limits_{\CH_1}^{}\tr \rho\ketbra{\psi(u_0\oplus z)}{\psi(v_0\oplus z)} \dd z = G_{\rho_0}(\bar{v}_0,u_0)e^{-\frac{1}{2}(\norm{u_0}^{2}+\norm{v_0}^2)}.
\end{equation}
In other words, \begin{equation}\label{eq:marginal-gen-fn}
    G_{\rho_0}(u_0, v_0) = \frac{1}{\pi^{n_1}}\int\limits_{\CH_1}   G_{\rho}(\bar{v}_0+\bar{z}, \bar{u}_0+z) e^{-\norm{z}^2}\dd z.
\end{equation}
\end{thm}
\begin{proof}
 Take $\CH = \BC^n$ and $\CH_i=\BC^{n_i}, i =0,1$. Let $I_i$ denote the identity operator in $\mathfrak{h}_i, i = 0,1$. Using  Lemma \ref{lem:semigr-ensur} and the general property that,  $\tr A(I_0\otimes B) = \tr (\tr_0A)B$ and $\tr C (D\otimes I_1) = \tr (\tr_1C)D$ for operators $B, D$ in $\B{\mathfrak{h}_1}$, $\B{\mathfrak{h}_0}$ respectively, we have 
\begin{align*}
&\frac{1}{\pi^{n_1}}\int\limits_{\BC^{n_1}}^{}\tr\rho \ketbra{\psi(\bu_0\oplus \z)}{\psi(\bv_0\oplus \z)} \dd \z \\
& = \frac{1}{\pi^{n_1}}\int\limits_{\BC^{n_1}}^{}\tr\rho \ketbra{\psi(\bu_0)}{\psi(\bv_0)}\otimes \ketbra{\psi(\z)} \dd \z \\
& = \frac{1}{\pi^{n_1}}\int\limits_{\BC^{n_1}}^{}\tr\rho (\ketbra{\psi(\bu_0)}{\psi(\bv_0)}\otimes I_1)(I_0\otimes \ketbra{\psi(\z)}) \dd \z \\
& =  \frac{1}{\pi^{n_1}}\int\limits_{\BC^{n_1}}^{}\tr (\tr_0(\rho \ketbra{\psi(\bu_0)}{\psi(\bv_0)}\otimes I_1)\ketbra{\psi(\z)}) \dd \z  \\
& = \tr \rho (\ketbra{\psi(\bu_0)}{\psi(\bv_0)}\otimes I_1)\\
& = \tr (\tr_1\rho)\ketbra{\psi(\bu_0)}{\psi(\bv_0)}\\
& = \tr \rho_0\ketbra{\psi(\bu_0)}{\psi(\bv_0)}\\
& =  G_{\rho_0}(\bar{\bv}_0,\bu_0)e^{-\frac{1}{2}(\abs{\bu_0}^{2}+\abs{\bv_0}^2)}.
\end{align*}
\end{proof}

\section{Klauder-Bargmann Representation for Gaussian Symmetries}\label{sec:KB-gaussian}
 We continue our discussions with  a finite dimensional Hilbert space  $\CH$  with an orthonormal basis $\{e_j|1\leq j\leq n\}$ and the identification of $z = \sum_{j=1}^nz_je_j \in \CH$ with $\z = (z_1,z_2, \dots, z_n) \in \BC^n$. Let $\CH_{\BR}$ be the real linear span of the orthonormal basis, then $\CH = \CH_{\mathbb{R}}+i\CH_{\BR}$, i.e, if $z\in \CH, z = x+iy, x, y\in \CH_{\BR}$. Furthermore, let $\mathfrak{L}_{\BR}(\CH)$ be the real algebra of real linear operators on $\CH$. Then  for $L\in \mathfrak{L}_{\BR}(\CH)$, $Lz = (Ax+By)+i(Cx+Dy)$, where $A, B, C, D$ are operators in $\CH_{\BR}$ with respective real matrices denoted again by $A, B, C, D$.  Write 
\begin{equation}
\label{eq:28}
L_0 = \bmqty {A&B\\C&D},
\end{equation} where $L_0 \in M_{2n}(\BR)$, i.e., a ${2n\times 2n}$ real matrix. 
 We now have
\begin{equation}
\label{eq:2.kb3}
L\z = \bmqty{I & iI}L_0\bmqty{\x\\\y}, \z=\x+i\y \in \BC^n,\x, \y \in \BR^n.
\end{equation} where $I$ is the identity matrix of order $n$.

\begin{lem}\label{lem:integral-kernel-2}
Let $L, M \in \mathfrak{L}_{\BR}(\CH)$ and $z = x +iy,z' =x'+iy'$, where $x,x',y,y'\in \CH_{\BR}$. Then 
\begin{equation}\label{eq:19}
  \braket{Lz}{Mz'} =  \bmqty{\x^{T}&\y^{T}}L_0^{T}(I+iJ)M_0\bmqty{\x'\\\y'},
\end{equation} where $J = \bmqty{0 & I\\ -I & 0}$.
In particular, 
\begin{equation}
\label{eq:7}
\abs{Lz}^{2}=\bmqty{\x^{T}&\y^{T}}L_0^{T}L_0\bmqty{\x\\\y}.
\end{equation}
\end{lem}
\begin{proof} 

By (\ref{eq:2.kb3}),
\begin{align*}
    \braket{Lz}{Mz'} 
    & = \bmqty{\x^T&\y^T}L_0^T\bmqty{I \\ -iI}\bmqty{I & i I}M_0\bmqty{\x'\\\y'}\\
    & = \bmqty{\x^T&\y^T}L_0^T(I+iJ)M_0\bmqty{\x'\\\y'}.
\end{align*}
\end{proof}
\begin{eg}\label{eg:complex-linear}
  Consider $\Lambda \in M_n(\BC)$ as a real linear transformation on $\BC^n$ then 
\begin{equation}
\label{eq:29}
\Lambda_0 = \bmqty{\re \Lambda& -\im \Lambda\\ \im \Lambda& \phantom{-}\re \Lambda}.
\end{equation} 
In this case $J$ commutes with $\Lambda_0$. Furthermore, $\Lambda$ is selfadjoint if and only if $\Lambda_0$ is  symmetric. Thus $J\Lambda_0$ is skew symmetric in this case. Now by Lemma \ref{lem:integral-kernel-2} $\Lambda \geq 0$ if and only if $\Lambda_0\geq 0$. 
\end{eg}
\begin{defn}
  A real linear operator $L \in \mathfrak{L}_{\BR}(\CH)$ is called a \emph{symplectic transformation} of $\CH$ if  \[\im \braket{Lz}{Lz'} =\im \braket{z}{z'},\forall z,z' \in \CH. \]
By Lemma \ref{lem:integral-kernel-2}, this is equivalent to 
\[\bmqty{\x^T&\y^T}L_0^TJL_0\bmqty{\x'\\\y'} = \bmqty{\x^T&\y^T}J\bmqty{\x'\\\y'},\forall \x,\y, \x',\y' \in \BR^n,\]
or equivalently, 
\begin{equation}
\label{eq:4.kb3}
L_0^TJL_0 = J.
\end{equation}
Any $L_0 \in M_{2n}(\BR)$  satisfying (\ref{eq:4.kb3}) is called a \emph{symplectic matrix}.
\end{defn}
Suppose $L_0$ is any symplectic matrix. Since $J$ is an orthogonal matrix, taking determinants on both sides of (\ref{eq:4.kb3}) we get $(\det L_0)^2 = 1$. Thus $L_0$ is nonsingular. Furthermore, (\ref{eq:4.kb3}) shows that $L_0$ and $(L_0^{-1})^T$ are orthogonally equivalent through $J$. 
Thus $a$ is an eigenvalue of $L_0$ if and only if $a^{-1}$ is so, hence $\det(L) = 1$. Multiplying by $(L_0^{-1})^T$ on the left and by $L_0^{-1}$ on the right on both sides of (\ref{eq:4.kb3}) shows that  $L_0^{-1}$ is symplectic. Thus symplectic matrices form a group under multiplication. Indeed, it is a unimodular Lie group, denoted $Sp(2n,\mathbb{R})$ and known as the \emph{symplectic real matrix group} of order $2n$. From our discussions it is clear that all symplectic transformations of $\CH$ constitute a group, denoted $Sp(\CH)$, isomorphic to the Lie group $Sp(2n, \BR)$.

We now make a detailed analysis of the unitary operators $\GL, L\in Sp(\CH)$ occurring in Property \ref{item:4.kb2} of Weyl operators in Section \ref{sec:weyl-oper-wign}.

\begin{prop} \label{prop:alphaL}
  Let $L \in Sp(\CH)$. Then \begin{equation*}
      \abs{\mel{\Omega}{\GL}{\Omega}} = \alpha(L)^{-1/4}, 
  \end{equation*}
  where \begin{equation*} 
      \alpha(L) = \det \frac{1}{2}(I + L_0^TL_0).
  \end{equation*}
\end{prop}
\begin{proof}
  Let $\rho = \ketbra{\Omega}{\Omega}$. Its quantum characteristic function (Definition \ref{defn:qft}) is given by, 
\begin{align*}
\hat{\rho}(z)
 &= \mel{\Omega}{W(z)}{\Omega}
 = e^{-\frac{1}{2}\abs{\z}^2}, z \in \CH. 
 \numberthis \label{sec:integral-kernel}
\end{align*}
By inversion formula (\ref{eq:5.kb2}), \[\ketbra{\Omega}{\Omega} = \frac{1}{\pi^n}\int\limits_{\CH}^{}e^{-\frac{1}{2}\abs{z}^2}W(-z)\dd z.\]
Conjugation by $\GL$ gives
\begin{equation*}
  \Gamma(L)\ketbra{\Omega}{\Omega}\Gamma(L)^{-1}= \frac{1}{\pi^n}\int\limits_{\CH}^{} e^{-\frac{1}{2}\abs{z}^2}W(-Lz)\dd z.
\end{equation*}
Considering the matrix element $\bra{\Omega}\Gamma(L)\ketbra{\Omega}{\Omega}\Gamma(L)^{-1}\ket{\Omega}$ from the equation above, 
 using the unitarity of $\GL$, Lemma \ref{lem:integral-kernel-2}, and gaussian integral formula (Proposition \ref{thm:gaussian-int}), we get
\begin{align*}
\abs{\mel{\Omega}{\Gamma(L)}{\Omega}}^2&= \frac{1}{\pi^n}\int\limits_{\mathbb{C}^n}^{} \exp{-\frac{1}{2}(\abs{\z}^2+\abs{L\z}^2)}\dd \z\\
& = \frac{1}{\pi^n}\int\limits_{\mathbb{R}^{2n}}^{}\exp{-\frac{1}{2} \bmqty{\x^T & \y^T}(I+L_0^TL_0)\bmqty{\x \\\y}}\dd \x \dd \y\\
& = \left(\det \frac{1}{2}(I + L_0^TL_0)\right)^{-1/2}.
\end{align*}
\end{proof}
\begin{thm}\label{sec:integral-kernel-3}
 For  $L\in Sp(\CH)$, there exists a unique unitary operator $\GOL$ in $\GH$ satisfying the following: 
\begin{enumerate}
    \item \label{item:1.kb3}$\mel{\Omega}{\GOL}{\Omega} = \alpha(L)^{-1/4}$,  where \begin{equation}\label{eq:alphaL}
      \alpha(L) = \det \frac{1}{2}(I + L_0^TL_0).
  \end{equation} 
  \item \label{item:2.kb3} $\GOL W(u) \GOL^{-1} = W(Lu), \forall u \in \CH.$
  \item \label{item:3.kb3} $\Gamma_0(L)\ket{e(v)} = \alpha(L)^{1/4}\frac{1}{\pi^n} \int_{\mathbb{C}^n}^{} \exp{-\frac{1}{2}(\abs{\z}^2+\abs{L\z}^2)+\braket{\z}{\bv}}\ket{e(L\z)} \dd \z, \forall v \in \CH.$
\end{enumerate}
\end{thm}
\begin{proof}
Let $\GL$ be any unitary operator satisfying $\GL W(u) \GL^{-1} = W(Lu)$. Define \[\GOL = \frac{\abs{\mel{\Omega}{\GL}{\Omega}}}{\mel{\Omega}{\GL}{\Omega}}\GL.\]
By Proposition \ref{prop:alphaL}, 
$\GOL$ is a well defined unitary operator differing from $\GL$ by a scalar multiple of modulus unity and satisfying properties \ref{item:1.kb3}) and \ref{item:2.kb3}) of the theorem. To prove \ref{item:3.kb3}) we look at the rank one operator $\ketbra{e(v)}{\Omega}$ as the inverse of its quantum characteristic function, \[
\ketbra{e(v)}{\Omega} = \frac{1}{\pi^n}\int\limits_{\mathbb{C}^n}^{} \exp{-\frac{1}{2}\abs{\z}^2-\braket{\z}{\bv}}W(-\z)\dd z.
\]
Conjugation by $\GOL$ on both sides yields \[\GOL\ketbra{e(v)}{\Omega}\GOL^{-1}=\frac{1}{\pi^n}\int\limits_{\mathbb{C}^n}^{} \exp{-\frac{1}{2}\abs{\z}^2-\braket{\z}{\bv}}W(-L\z) \dd z.\] 
 Now $\GOL\ketbra{e(v)}{\Omega}\GOL^{-1}\ket{\Omega}$ is computed from the equation above by a change of variable $\z\mapsto -\z$ in the integral on the right hand side. This completes the proof.
\end{proof}
\begin{cor}\label{cor:generation-function-g0}
 The generating function (equation (\ref{eq:gen-funct})) of $\GOL$ is given by \begin{equation*}
     G_{\GOL}(u, v) = \alpha(L)^{-1/4}\exp{u^T Au + u^T\Lambda v+v^TBv},
 \end{equation*}
 where $\alpha(L)$ is as in (\ref{eq:alphaL}),  $A,\Lambda, B$ are respectively the $2$-particle annihilation, exchange and creation amplitude matrices of the operator $\GOL$ so that
 \begin{align*}
A= A_{\GOL} &= \frac{1}{2}\bmqty{I& iI}(I+(L_0^{-1})^TL_0^{-1})^{-1} \bmqty{I\\iI},\\  
 \Lambda = \Lambda_{\GOL} & = \bmqty{I& iI}(L_0^{-1}+L_0^T)^{-1} \bmqty{I\\-iI} \numberthis \label{eq:parameters-gol},\\
B=B_{\GOL}& = \frac{1}{2} \bmqty{I& -iI}(I+L_0^TL_0)^{-1} \bmqty{I\\-iI}.
 \end{align*}
\end{cor}
\begin{proof}
  Computing $\bra{e(\bar{u})}\Gamma_0(L)\ket{e(v)}$ from the identity \ref{item:3.kb3}) in Theorem \ref{sec:integral-kernel-3} we get 
\begin{align*}
 G_{\GOL}(u, v)& = \alpha(L)^{1/4}\frac{1}{\pi^n} \int\limits_{\mathbb{C}^n}^{} \exp{-\frac{1}{2}(\abs{\z}^2+\abs{L\z}^2)+\braket{\z}{\bv}+\braket{\bar{\bu}}{L\z}} \dd \z\numberthis \label{eq:gen-int-G0}\\
& = \frac{\alpha(L)^{1/4}}{\pi^n}  \int\limits_{\mathbb{R}^{2n}}^{} \exp{-\bmqty{\x \\\y}^T\left(\frac{I+L_0^TL_0}{2}\right) \bmqty{\x \\ \y } + \m^{T}\bmqty{\x \\ \y} }\dd \x \dd \y, 
\end{align*}
where $\m^T = \bv^T\bmqty{I& -iI}+\bu^{T}\bmqty{I&iI} L_0$. By the gaussian integral formula (Proposition \ref{thm:gaussian-int}) now it follows that
\begin{equation*}
  G_{\GOL}(u, v) = \alpha(L)^{-1/4}\exp{\frac{1}{2}\m^{T}(I+L_0^TL_0)^{-1}\m}.
\end{equation*}
Expanding the exponent in the right side after substituting the expression for $\m$ in terms of $\bu,\bv$ we get the required result.
\end{proof}

\begin{thm}
[Klauder-Bargmann representation for $\GOL$] \label{thm:klaud-bargm-repr-GL} 
Let $\ket{\psi(z)}$ denote the coherent state at $z\in \CH$. Then the  unitary operator $\GOL$ has the following weak operator integral representation 
\begin{equation}
\label{eq:27}
\Gamma_0(L) = {\alpha(L)}^{1/4}\frac{1}{\pi^n}  \int\limits_{\CH}^{} \ketbra{\psi(Lz)}{\psi(z)}\dd z.
\end{equation}
In particular, the map $L\mapsto\GOL$ from $Sp(\CH)$ into $\mathcal{U}(\CH)$ is strongly continuous.  
\end{thm}
\begin{proof}First we show that the right side of (\ref{eq:27}) defines a bounded operator on $\GH$. It is a consequence of Klauder-Bargmann isometry  (in Section \ref{sec:2}) that for any $L \in Sp(\CH)$, the map $\phi\underset{K_L}{\mapsto} \pi^{-n/2}\braket{\psi(Lz)}{\phi}, z\in \CH$ is an isometry of $\GH$ into $L^2(\CH)$, where $\CH$ is equipped with the $2n$-dimensional Lebesgue measure.. Therefore  for any $\phi, \phi' \in \GH$ the map $z\mapsto \braket{\phi}{\psi(Lz)}\braket{\psi(z)}{\phi'}$ is measurable and by Cauchy-Schwarz inequality 
\begin{equation*}
\frac{1}{\pi^n}\int\limits_{\CH}^{} \abs{\braket{\phi}{\psi(Lz)}\braket{\psi(z)}{\phi'}}\dd z  =\mathlarger{\langle}\abs{K_L(\phi)} \mathlarger{|}\abs{K_I(\phi')}\mathlarger{\rangle}_{L^2} \leq \norm{\phi}\norm{\phi'}.
\end{equation*} Thus the right side of (\ref{eq:27}) defines a bounded operator on $\GH$.
  For $\CH = \BC^{n}$ we have
  \begin{equation*}
    \braket{e(\bar{\bu})}{\psi(L\z)}\braket{\psi(\z)}{e(\bv)} = \exp{-\frac{1}{2}(|\z|^2+|L\z|^2)+\braket{\bar{u}}{L\z}+\braket{\z}{\bv}}.
  \end{equation*}
Now equation (\ref{eq:gen-int-G0}) implies that
\begin{equation*}
   G_{\GOL}(u, v) = \alpha(L)^{1/4}\frac{1}{\pi^n}\int\limits_{\BC^n}^{}  \braket{e(\bar{\bu})}{\psi(L\z)}\braket{\psi(\z)}{e(\bv)}\dd \z, \forall \bu, \bv \in \CH
\end{equation*}
which yields (\ref{eq:27}).
\end{proof}

\begin{rmks}
For any $L\in Sp(\CH)$, a unitary operator $\Gamma$ in $\GH$ is said to intertwine $W(u)$ and $W(Lu)$ for all $u \in \CH$ if \begin{align*}
        \Gamma W(u) \Gamma^{-1} = W(Lu).
    \end{align*}
    Let $\mathcal{G}_L$ denote the set of all such intertwiners. We know that there exists a unique element $\GOL$ which satisfies the condition $\mel{\Omega}{\GOL}{\Omega}> 0$.
    \begin{enumerate}
        \item If $U$ is a unitary operator in $\CH$, i.e., $U\in \mathcal{U}(\CH)$ and $\Gamma(U)$ is the associated second quantization operator satisfying $\Gamma(U)\ket{e(u)} = \ket{e(Uu)}, \forall u \in \CH$, then $\mel{\Omega}{\Gamma(U)}{\Omega} = 1$ and hence $\Gamma_0(U) =\Gamma(U)$.
        \item For any $U, V$ in $\mathcal{U}(\CH)$ and $L \in Sp(\CH)$, $\Gamma_0(U)\Gamma_0(L) \Gamma_0(V)$ and $\Gamma_0(ULV) \in \mathcal{G}_{ULV}$ and $\mel{\Omega}{\Gamma_0(U)\Gamma_0(L)\Gamma_0(V)}{\Omega} = \mel{\Omega}{\Gamma_0(ULV)}{\Omega}>0$. Hence \begin{equation*}
            \Gamma_0(U)\Gamma_0(L)\Gamma_0(V) = \Gamma_0(ULV).
        \end{equation*}
        \item For any $L \in Sp(\CH)$, $\mel{\Omega}{\GOL^{-1}}{\Omega} = \overline{\mel{\Omega}{\GOL}{\Omega}}>0$, $\Gamma_0(L)^{-1}$ and $\Gamma_0(L^{-1})$ lie in $\mathcal{G}_{L^{-1}}$ and hence $\Gamma_0(L^{-1}) = \Gamma_0(L)^{-1}$. 
    \end{enumerate}
\end{rmks}
\begin{defn}\label{defn:gaussian-state}
An $n$-mode state  $\rho \in \B{\GH}$  is called a \emph{gaussian state} if there exists an $\m \in \BC^n$ and a $2n \times 2n$ real, symmetric matrix $S$ such that
\begin{equation}\label{eq:gaussian-qft-1}
    \hat{\rho}(z) = \exp{-2i\im\braket{\z}{\m}- \bmqty{\x^T&\y^T}{S}\bmqty{\x\\\y}}, 
\end{equation} where $z\mapsto \z = \x+i\y$ is an identification of $\CH$ with $\BC^n$.
In this case, we write $\rho = \rho_{\m, S}$, $\m$ is the \emph{mean annihilation vector}, simply called the \emph{mean} of the state and $S$ is the \emph{position-momentum covariance matrix}. 
\end{defn}
\begin{rmks} 
\begin{enumerate}
\item\label{item:1}  A  $2n \times 2n$ real, symmetric matrix  $S$ is the position-momentum covariance of a gaussian state if and only if the following matrix inequality holds \cite{Par10}
\begin{equation}
\label{eq:11}
S+\frac{i}{2}J\geq 0.
\end{equation} 
\item\label{item:2}  We say that a unitary operator $U$ on $\GH$ is a gaussian symmetry if, for any gaussian state $\rho$ in $\GH$, the state $U \rho U^{\dagger}$ is also gaussian. It is known that any such gaussian symmetry $U$ is equal to $\lambda W(u)\GOL$, where $\lambda$ is a scalar of modulus unity, $u \in \CH, L \in Sp(\CH)$ (See \cite{Par13,BhJoSr18}). Now Corollary \ref{cor:kb-g-symmetries} below says that every gaussian symmetry has a Klauder-Bargmann integral representation (\ref{eq:kb.3.2}).
\end{enumerate} 
\end{rmks}
\begin{cor}
  [Klauder-Bargmann representation for symmetries of gaussian \newline states] \label{cor:kb-g-symmetries}
For $u\in \CH, L\in Sp(\CH)$, 
\begin{equation}
\label{eq:kb.3.2}
W(u)\Gamma_0(L)=  \frac{\alpha(L)^{1/4}}{\pi^n}  \int\limits_{\CH}^{}\exp{-i\im \braket{u}{Lz}} \ketbra{\psi(u+Lz)}{\psi(z)}\dd z.
\end{equation}
\end{cor}
\begin{proof}
    This is immediate from (\ref{eq:27}).  
\end{proof}

\begin{thm}\label{thm:semidirect}
 Denote by $\CH \text{\textcircled{s}}Sp(\CH),$ the Lie group which is the semidirect product of the additive group $\CH$ and the Lie group $Sp(\CH)$ acting on $\CH$ so that the multiplication in $\CH\text{\textcircled{s}}Sp(\CH)$ is defined by 
 \begin{equation*}
     (u,L)(v,M) = (u+Lv, LM) \mbox{ for all }u,v \in \CH, L, M \in Sp(\CH).
 \end{equation*}
 Then the map $(u, L) \mapsto W(u)\GOL$ is a strongly continuous, projective unitary representation of $\CH\text{\textcircled{s}}Sp(\CH)$ in $\GH$.
\end{thm}
\begin{proof}
  Only the strong continuity of the map remains to be proved. To this end we consider the Lebesgue measure preserving group action $(u,L): z \rightarrow u + Lz, (u,L) \in \CH\text{\textcircled{s}}Sp(\CH)$. This yields a strongly continuous unitary representation 
  \begin{equation*}
      (U_{(u,L)} f)(z) = f((u,L)^{-1}z), f \in L^{2}(\CH)
  \end{equation*}
  with Lebesgue measure on $\CH$ by identification of $\CH$ with $\mathbb{C}^{n}$. This implies that the map 
  \begin{equation*}
      \phi \mapsto \braket{\phi}{\psi(u+Lz)}, z \in \CH, \phi \in \GH
  \end{equation*}
  is continuous as a map from $\GH$ into $L^{2}(\CH)$, thanks to the Klauder-Bargmann isometry (item \ref{bargmann-transform} in Section \ref{sec:2}). Now an application of Klauder-Bargmann representation implies that 
  \begin{equation*}
      (u,L) \mapsto \mel{\phi}{W(u)\GOL}{\phi'}
  \end{equation*}
  is continuous in $(u,L)$ for any $\phi, \phi'$ in $\GH$. In other words, $(u,L) \mapsto W(u)\GOL$ is weakly continuous. Unitarity implies strong continuity. 
\end{proof}
\begin{rmk}
  Writing 
  \begin{equation*}
      W(u,L) = W(u)\GOL
  \end{equation*}
  we conclude from Theorem \ref{thm:semidirect} that the map $(u,L) \mapsto W(u,L)$ is a strongly continuous projective irreducible unitary representation of the semidirect product group $\CH\text{\textcircled{s}} Sp(\CH)$ in $\GH$ and 
  \begin{equation*}
      W(u,L)W(v,M) = e^{-i\im\braket{u}{Lv}}\sigma_{0}(L,M)W(u+Lv, LM) \ \  \forall u,v \in \CH, L,M \in Sp(\CH)
  \end{equation*}
  where $\sigma_{0}(L,M)$ is a continuous function of $(L,M)$ taking values in the unit circle.
\end{rmk}
\section{The Semigroup \ensuremath{\CE_2(\CH)}}\label{sec:semigroup}
Recall the definition of the generating function of a bounded operator in $\GH$ from Section \ref{sec:2}.
\begin{defn}
An operator $Z$ on $\GH$ is said to be in the class $\CE_2(\CH)$ if there exists $c \in \BC$, $\bm{\alpha}, \bm{\beta} \in \BC^n$, $A, B, \Lambda \in M_n(\BC)$, with $A$ and $B$  symmetric, such that the generating function of $Z$ is of the form \begin{equation}\label{eq:e2-1}
    G_Z(u, v)  = c \exp{\cv{u}^T\bm{\alpha}+ \bm{\beta}^T\cv{v}+ \cv{u}^TA\cv{u} +\cv{u}^T\Lambda\cv{v}+\cv{v}^TB\cv{v}}, \forall \cv{u},\cv{v} \in \BC^n.
\end{equation}
The ordered $6$-tuple $(c, \bm{\alpha}, \bm{\beta}, A, \Lambda, B)$ completely characterizes $Z\in \mathcal{E}_2(\CH)$ and we call them the $\mathcal{E}_2(\CH)$-parameters of $Z$.
\end{defn}
\begin{egs}\label{sec:semigr-3} \begin{enumerate}
    \item Let $K$ be any contraction operator on $\CH$, then the second quantization contraction $\Gamma(K)$ satisfies $\Gamma(K)\ket{e(v)} = \ket{e(Kv)}, v\in \CH$. So, \begin{equation}
        G_{\Gamma(K)}(u, v) = \exp{\bu^TK\bv}.
    \end{equation} 
    Hence $\Gamma(K) \in \mathcal{E}_2(\CH)$ with parameters $(1, 0,0,0,K,0)$.
\item Let $z \in \CH$, then the associated Weyl displacement operator satisfies \[W(z)\ket{e(v)} =\exp{-\frac{1}{2}\abs{z}^2-\braket{z}{v}}\ket{e(z+v)}.\] So
  
\begin{equation}
\label{eq:2}
 G_{W(z)}(u, v) = \exp{-\frac{1}{2}\abs{z}^2-\braket{z}{v}+ \braket{\bar{u}}{z+v}}.
\end{equation}
Hence $W(z)\in \mathcal{E}_2(\CH)$ with parameters $(e^{-\frac{1}{2} \abs{z}^2},  \z, -\bar{\z}, 0, I, 0).$
   \item Let $L\in Sp(\CH)$. By Corollary \ref{cor:generation-function-g0}, $\GOL \in \mathcal{E}_2(\CH)$, and the parameters of the same are $(\alpha(L)^{-1/4}, 0,0, A_{\GOL}, \Lambda_{\GOL}, B_{\GOL}),$  where $\alpha(L)$ is as in (\ref{eq:alphaL}) and $A_{\GOL}$, $\Lambda_{\GOL}$, $ B_{\GOL}$ are as in (\ref{eq:parameters-gol}).
   \item \label{item:29} Let $Z_j\in \CE_2(\CH_j)$,  with parameters $(c_j,\bm{\alpha}_j, \bm{\beta}_j, A_j, \Lambda_j, B_j)$, $j=1,2$ then $Z_1\otimes Z_2 \in \CE_2(\CH_1\oplus\CH_2)$ with parameters $(c_1c_2,\bm{\alpha}_1\oplus\bm{\alpha}_2,\bm{\beta}_1\oplus\bm{\beta}_2,A_1\oplus A_2,\Lambda_1\oplus \Lambda_2, B_1\oplus B_2)$ using the identification of $\Gamma(\CH_1\oplus\CH_2)$ with $\Gamma(\CH_1)\otimes\Gamma(\CH_2)$ described in Section \ref{sec:2}.
  \end{enumerate}
\end{egs}
 Suppose $Z \in \mathcal{E}_{2}(\CH)$ with parameters $(c, \bm{\alpha}, \bm{\beta}, A, \Lambda, B)$, by  (\ref{eq:e2-1}) we have,
\begin{align}\label{eq:g3-1}
G_Z(u,v) &= \nonumber c \{1 +(\bu^{T}\bm{\alpha} + \bm{\beta}^{T}\bv + \bu^{T}A\bu + \bu^{T}\Lambda\bv + \bv^{T}B\bv ) \\
& \phantom{....} +  \frac{1}{2!}(\bu^{T}\bm{\alpha} + \bm{\beta}^{T}\bv + \bu^{T}A\bu + \bu^{T}\Lambda\bv + \bv^{T}B\bv )^{2} + \cdots \}.  
\end{align}
From the discussion in Section \ref{sec:2} and comparing with the definition of 1-particle creation, annihilation vectors of equation (\ref{eq:5}) and the 2-particle creation, annihilation and exchange matrices of equation (\ref{eq:16}) we see that
\begin{equation}
\label{eq:61}
\begin{split}
c = \mel{\Omega}{Z}{\Omega},\phantom{..............}c\bm{\alpha}= \bm{\lambda}_Z,\phantom{..............}c\bm{\beta}= \bm{\mu}_Z,\phantom{.......}\\ c(A+\frac{\bm{\alpha}\bm{\alpha}^T}{2}) =A_Z,\phantom{...} c(B+\frac{\bm{\beta}\bm{\beta}^T}{2}) = B_Z,\phantom{...} c(\Lambda+\bm{\alpha}\bm{\beta}^T)= \Lambda_Z.\end{split}
\end{equation}
\begin{rmk}
  The most notable feature of an operator $Z\in \CE_2(\CH)$ is the property that all the matrix entries of $Z$ in the particle basis are completely determined by the entries $\{\mel{\bk}{Z}{\bbell}|\abs{\bk}+\abs{\bbell} \leq 2\}$  which is a finite set of cardinality $2n^2+3n+1$.
 If $Z$ is selfadjoint then it is determined by $(n+1)^2$ entries out of which $n$ are real, $(n+1)^2-n$ may be complex entries. Later we shall prove that a state $\rho$ in $\GH$ is gaussian if and only if it is in $\CE_2(\CH)$. Thus tomography of a gaussian state in $\GH$ requires the estimation of at most $2(n+1)^2-n$ events which are one dimensional projections in the subspace spanned by $0, 1$ and $2$-particle vectors.
\end{rmk}
\begin{prop}\label{prop:gaussian-e2h}
 If  $\rho$ is an $n$-mode  gaussian state then $\rho\in \mathcal{E}_2(\CH)$. 
\end{prop}

\begin{proof}
Take $\CH=\BC^n$. Let $\rho = \rho_{\m, S}$ be a gaussian state on $\GCn$.
By (\ref{eq:gaussian-qft-1}),  (\ref{eq:2}) and (\ref{eq:19}),
  \begin{align}\label{eq:6}
    \hat{\rho}(\z) &= \exp{\bmqty{\x^T& \y^T}\bmqty{\phantom{i}(\bar{\m}-\m)\\i(\bar{\m}+\m)} -\bmqty{\x^T&\y^T}S\bmqty{\x \\ \y}}\\  \label{eq:8}
G_{W(\z)}(u, v)& = e^{\bu^{T}\bv}\exp{-\frac{1}{2}\bmqty{\x^T&\y^T}I\bmqty{\x\\\y}+\bmqty{\x^T&\y^T}\left(\bmqty{I\\iI}\bu - \bmqty{I\\-iI}\bv\right)}.
  \end{align}
Now by the Wigner isomorphism Theorem \ref{thm:wigner-iso}, and equations (\ref{eq:6}) and (\ref{eq:8}),
\begin{align*}
G_{\rho}(\bu, \bv)& = \tr \rho \ketbra{e(\bv)}{e(\bar{\bu})} =  \braket{\rho}{\ketbra{e(\bv)}{e(\bar{\bu})}}_{\Bt{\GH}}\\
& = \frac{1}{\pi^n}\int\limits_{\BC^n}^{} \overline{\hat{\rho}(\z)}\mel{e(\bar{\bu})}{W(\z)}{e(\bv)}\dd \z\\
& = \frac{1}{\pi^n}\int\limits_{\BC^n}^{} \overline{\hat{\rho}(\z)}G_{W(\z)}(u, v)\dd \z,\\
& = \frac{e^{\bu^{T}\bv}}{\pi^n}\int\limits_{\BR^{2n}}^{}\exp{-\bmqty{\x^T&\y^T}(\frac{1}{2}I+S)\bmqty{\x\\ \y}+\q^T\bmqty{\x\\ \y}},\numberthis \label{eq:gaussian-e2}
\end{align*}
where $\q = \bmqty{I\\iI}\bu-\bmqty{I\\-iI}\bv+\bmqty{\phantom{i}(\bar{\m}-\m)\\i(\bar{\m}+\m)}$. The right hand side of (\ref{eq:gaussian-e2}) being a gaussian integral in $\BR^{2n}$ implies that $\rho_{\m,S}\in \mathcal{E}_2(\CH)$. 
\end{proof}

\begin{prop}
  \label{prop:e2-sa-positive}
Suppose $Z\in  \mathcal{E}_2(\CH)$ with parameters  $(c, \bm{\alpha}, \bm{\beta}, A, \Lambda, B)$.  If $Z = Z^{\dagger}$, then $c\in \BR, \bm{\alpha}=\bar{\bm{\beta}}, A = \bar{B}$, and $\Lambda=\Lambda^{\dagger}$. Furthermore, if $Z\geq 0$ and $Z\neq 0$, then $c>0$ and  $\Lambda\geq 0$.
\end{prop}
\begin{proof}
 Take $\CH =\BC^n$ without loss of generality. By definition,
\begin{equation}
\label{eq:9}
G_{Z^{\dagger}}(\bu, \bv) = \overline{G_Z(\bar{\bv},\bar{\bu})}.
\end{equation}
Therefore self-adjointness of  $Z$ implies, 
\begin{equation*}
    G_{Z}(\bu, \bv) = \overline{G_Z(\bar{\bv},\bar{\bu})}.
\end{equation*}
Thus
\begin{equation}\label{eq:prop4.5}\begin{split}
    c \exp{\cv{u}^T\bm{\alpha}+ \bm{\beta}^T\cv{v}+ \cv{u}^TA\cv{u} +\cv{u}^T\Lambda\cv{v}+\cv{v}^TB\cv{v}} \phantom{.........................................}\\= \bar{c} \exp{\cv{v}^T\bar{\bm{\alpha}}+ \bar{\bm{\beta}}^T\cv{u}+ \cv{v}^T\bar{A}\cv{v} +\cv{v}^T\bar{\Lambda}\cv{u}+\cv{u}^T\bar{B}\cv{u}}, \forall \bu,\bv\in \BC^n.
        \end{split}
\end{equation}
This is possible only if $c = \bar{c}, \bm{\alpha} = \bar{\bm{\beta}}, A = \bar{B}$ and $\Lambda = \Lambda^{\dagger}$.
If $Z\geq 0$ and $Z\neq 0$, then $c\neq0$ and $c = G_Z(0,0)>0$. Furthermore the kernel 
$(\bu, \bv)\mapsto G_Z(\bar{\bu}, \bv)$ is positive definite and is of the form \[G_Z(\bar{\bu}, \bv) = c\overline{f(\bu)}f(\bv)e^{\mel{\bu}{\Lambda}{\bv}},\] where $f(\bv) = \exp{\bv^T\bar{A}\bv+\bar{\bm{\beta}}^T\bv}$. Since  $c\overline{f(\bu)}f(\bv)$ is already a positive definite kernel, $e^{\mel{\bu}{\Lambda}{\bv}}$ has to be positive definite.  Hence  $\mel{\bu}{\Lambda}{\bv}$ is positive semidefinite, in other words $\Lambda \geq 0$.
\end{proof}

\begin{thm}
  The class $\mathcal{E}_2(\CH)$ is a $\dagger$-closed, weakly closed semigroup in $\B{\GH}$.
\end{thm}
\begin{proof} 
The class $\mathcal{E}_2(\CH)$ is $\dagger$-closed because of (\ref{eq:9}).  To prove that $\mathcal{E}_2(\CH)$ is closed in the weak operator topology, let $Z_n$ be a sequence in $\mathcal{E}_2(\CH)$ which converges weakly to an operator $Z$. Then the generating functions $G_{Z_n}(u,v)\rightarrow G_Z(u,v)$, $\forall$ $u, v \in \CH$. By (\ref{eq:4}), $G_{Z_n}(u,v)$ and $G_Z(u,v)$ are power series in $\bu,\bv$ and hence the corresponding coefficients also converge. Now equations in 
(\ref{eq:61}) prove  that $Z\in \CE_2(\CH)$. Now we prove that $\CE_2(\CH)$ is a semigroup.
Let $Z_j \in \mathcal{E}_2(\CH)$, with parameters $(c_j, \bm{\alpha}_j, \bm{\beta}_j, A_j, \Lambda_j, B_j), j= 1, 2.$ Assume without loss of generality that $\CH=\BC^n$.
By the Klauder-Bargmann isometry and formula,
\begin{align*}
    G_{Z_1Z_2}(\bu,\bv) & = \mel{e(\bar{\bu})}{Z_1Z_2}{e(\bv)}\\
    &= \frac{1}{\pi^n}\int\limits_{\CH}^{}\mel{e(\bar{\bu})}{Z_1}{\psi(\z)}\mel{\psi(\z)}{Z_2}{e(\bv)}\dd \z\\
    &= \frac{1}{\pi^n}\int\limits_{\CH}^{}G_{Z_1}(\bu,\z)G_{Z_2}(\bar{\z},\bv) e^{-\abs{\z}^2}\dd \z, \numberthis \label{eq:e2-semigroup}
\end{align*} where the integrand is an $L_1$ function of $\z$ of the form \[f(\bu,\bv)g(\bu,\bv,\z),\] 
with
\begin{align*}
f(\bu,\bv) &= \frac{c_1c_2}{\pi^n}\exp{\bu^T\bm{\alpha}_1+ \bm{\beta}_2^T\bv + \bu^TA_1\bu + \bv^TB_2\bv } 
\textnormal{ and}\\
g(\bu,\bv,\z) & = \exp{-\abs{\z}^2 +\z^TB_1\z +\bar{\z}^TA_2\bar{\z}+ (\bm{\beta}_1^T+\bu^T\Lambda_1)\z + (\bm{\alpha}_2^T +\bv^T\Lambda_2^T)\bar{\z}}
\end{align*} which is an $L^1$ function of $\z$. If we write $\z = \x+i\y$ then $g(\bu,\bv,\z)$ assumes the form
\begin{equation*}
  g(\bu,\bv,\x +i\y) = \exp{-\bmqty{\x^T&\y^T}R\bmqty{\x\\\y} + \left(\bmqty{\bu\\\bv}^TS+\bmqty{\bm{\mu}\\\bm{\nu}}^T \right)\bmqty{\x\\\y}}
\end{equation*}
where $R$ is a symmetric $2n \times 2n$ matrix,  $S$ is a $2n \times 2n$ complex matrix and $\bm{\mu}, \bm{\nu}$ are vectors in $\BC^n$. Because $g(\bu,\bv,\z)$ is integrable, the real part of $R$ is strictly positive. Now an application of gaussian integral formula shows that $G_{Z_1Z_2}(\bu, \bv)$ as a function of $(\bu,\bv)$ has the structure of the generating function of an operator in $\mathcal{E}_2(\CH)$. This completes the proof.
\end{proof}

\begin{notn}
 Proposition \ref{prop:e2-sa-positive} allows a reduction in the number of parameters required to describe a positive operator  $\rho$ belonging to $\CE_2(\CH)$. We now parametrize such a $\rho$ by the quadruple $(c, \bm{\alpha}, A, \Lambda)$, where $c = \mel{\Omega}{\rho}{\Omega}>0$, $\bm{\alpha} \in \BC^n$, $A$ is a complex symmetric matrix and $\Lambda$ is a positive matrix, so that the generating function of  $\rho $ takes the form 
\begin{equation}
\label{eq:118}
G_{\rho}(u, v) = c \exp{\cv{u}^T\bm{\alpha}+ \bar{\bm{\alpha}}^T\cv{v}+ \cv{u}^TA\cv{u} +\cv{u}^T\Lambda\cv{v}+\cv{v}^T\bar{A}\cv{v}}.
\end{equation}
\end{notn}
 \begin{thm}\label{thm:gaussian-iff-e2h}
  A state  $\rho$ in $\GH$ is gaussian  if and only if $\rho \in \CE_2(\CH)$.
 \end{thm}
\begin{proof}
If $\rho$ is a gaussian state then by Proposition \ref{prop:gaussian-e2h} $\rho \in \CE_2(\CH)$. 
Conversely, let $\rho\in \mathcal{E}_2(\CH)$ be a state with parameters $(c, \bm{\alpha}, A, \Lambda)$. 
Now by Proposition \ref{prop:qft-in-terms-generator}, $\hat{\rho}(\bu)$ takes the form
\begin{equation}\label{eq:10}
\hat{\rho}(\bu) =e^{-\frac{1}{2}\abs{\bu}^2+\bu^T\bar{A}\bu+ \bm{\alpha}^T\bu} \frac{c}{\pi^n}\int\limits_{\BR^{2n}}^{}\exp{-\bmqty{\x^T&\y^T}M\bmqty{\x\\ \y}+\bbell^T\bmqty{\x\\ \y} }\dd \x \dd \y,
\end{equation}
where $M$ is an $2n\times 2n$ complex symmetric matrix  and $\bbell \in \BC^{2n}$ is of the form 
\[\bbell^T = \bu^TM_1 + \bar{\bu}^TM_2 + \mathbf{p}^T ,\] with $M_1$ and $M_2$ being $n\times 2n$ constant complex matrices    and $\mathbf{p} \in \BC^{2n}$ is a constant complex vector.
Since the function under the integral sign  is integrable, $M$ has strictly positive real part. If $\bu = \bm{\xi} +i\bm{\eta}$, an application of the gaussian integral formula in $\BR^{2n}$ shows that 
\begin{equation}\label{eq:12}
  \hat{\rho}(\bm{\xi}+i\bm{\eta})  = c'\exp{-Q(\bm{\xi}, \bm{\eta})+ \q_1^T\bm{\xi}+\q_2^T\bm{\eta} }, \forall \bm{\xi},\bm{\eta} \in \BR^n,
\end{equation}  
where $c'$ is a constant scalar, $Q$ is a quadratic form in $2n$ real variables with complex coefficients and $\q_1, \q_2$ are some elements in $\BC^n$. 
Furthermore, 
 by Property \ref{item:9.kb.2} of quantum characteristic function, the map $t\mapsto\hat{\rho}(t(\bm{\xi}+i\bm{\eta})), t\in \BR$ is the characteristic function of a probability distribution $\mu_{\bm{\xi},\bm{\eta}}$ on the real line for any fixed $\bm{\xi},\bm{\eta}$. Hence $\mu_{\bm{\xi},\bm{\eta}}$ is a normal distribution on $\BR$, $Q(\bm{\xi},\bm{\eta})\geq 0, \forall\bm{\xi},\bm{\eta}$ and $\q_{j}=i\bm{\gamma}_j, j=1, 2$ for some real vectors $\bm{\gamma}_1, \bm{\gamma}_2 \in \BR^n$. Thus $\hat{\rho}$ is the quantum characteristic function of a gaussian state $\rho$ in $\GH$.
\end{proof}
\begin{cor}\label{cor:gaussian-powers}
Let $\rho$ be a gaussian state,  $Z\in \CE_2(\CH)$. Then for any $t > 0$, the state $(\tr\rho^tZ^{\dagger} Z)^{-1}Z\rho^tZ^{\dagger}$
is a gaussian state. 
\end{cor}
\begin{proof}
By the structure theorem for gaussian states (Theorem 4, \cite{Par12}), $(\tr \rho^t)^{-1}\rho^t$
is again a gaussian state, and hence $\rho^t\in \CE_2(\CH)$.  
Since $\CE_2(\CH)$ is a semigroup, the Corollary follows immediately.
\end{proof}

\begin{rmk} For any $Z\in \mathcal{E}_2(\CH)$ the map $\rho\mapsto \frac{Z\rho^{1+t}Z^{\dagger}}{\tr\rho^{1+t}Z^{\dagger} Z}$ on the set of all states yields a nonlinear gaussian state preserving information channel.
\end{rmk}
\begin{prop}
  \begin{enumerate}
\item\label{item:3} Any unitary operator in  $\mathcal{E}_2(\CH)$ is a gaussian symmetry.

\item \label{item:4} Any finite rank projection operator in $\mathcal{E}_2(\CH)$ is conjugate to the vacuum projection $\ketbra{\Omega}$ by a gaussian symmetry.
\end{enumerate}
\end{prop}
\begin{proof}
  \ref{item:3}. Let $U \in \mathcal{E}_2(\CH)$ be  unitary. If $\rho$ is any gaussian state, then  $\rho, U\rho U^\dagger \in \mathcal{E}_2(\CH)$. Thus $U\rho U^\dagger$ is a gaussian state. Hence $U$  is a gaussian symmetry. 

\ref{item:4}.  Let $P$ be a finite rank projection operator in  $\mathcal{E}_2(\CH)$. Then $P$  is a constant multiple of a gaussian state. By the structure theorem for gaussian states (Theorem 4, \cite{Par12}), it is a one dimensional projection onto the span of $W(z)\GOL\ket{\Omega}$ for some  $z\in \CH, L\in Sp(\CH)$.
\end{proof}
\begin{lem}
  \label{lem:conj-sec-quant}
Let $Z\geq0$ and $Z\in \CE_2(\CH)$ with parameters $(c, \bm{\alpha}, A, \Lambda)$, $K$ be any contraction operator in $\CH$ and $z\in \CH$. Then the following hold:
\begin{enumerate}
\item\label{item:39} The $\CE_2(\CH)$-parameters of 
$\Gamma(K)Z\Gamma(K)^{\dagger}$ are 
$(c, K\bm{\alpha}, KAK^T,K\Lambda K^{\dagger})$. 
\item \label{item:40} The $\CE_2(\CH)$-parameters of 
$W(-z)ZW(z)$ are 
$(c',\bm{\alpha}-(I-\Lambda-2AC)\z, A, \Lambda)$, where $c'=\mel{\psi(z)}{Z}{\psi(z)}$, $C$ is the complex conjugation map $\x+i\y \mapsto \x-i\y$ on $\BC^n$.
\end{enumerate}
\end{lem}
\begin{proof}
 Both the results follow from a direct computation of the generating function. 
\end{proof}

\begin{rmks}
  Let $Z$ be as in Lemma \ref{lem:conj-sec-quant}.
\begin{enumerate}
\item\label{item:41} If  $U$ is a unitary matrix which diagonalizes $\Lambda$ then the $\CE_2(\CH)$-parameters of the operator $\Gamma(U)Z\Gamma(U)^{\dagger}$ are $(c, U\bm{\alpha}, UAU^T, D)$, where $D=U\Lambda U^\dagger$ is a diagonal matrix consisting of the eigenvalues of $\Lambda$.
\item \label{item:42} The parameter $\bm{\alpha}$ of $Z$ can be brought to $0$ by conjugating $Z$ with a Weyl operator if there exists $\m\in\BC^n$ such that $(I-\Lambda-2AC)\m =\bm{\alpha}$. We will see in Proposition \ref{sec:gaussian-states-1} that this can be done whenever $Z$ is a trace class operator. 
\end{enumerate}
\end{rmks}

\section{Gaussian States and the Uncertainty Relations in terms of \ensuremath{\CE_2(\CH)}-parameters}\label{sec:gaussian}

Our investigations in the previous section show that there are two different ways of parametrizing the set of all gaussian states in $\GH$, one obtained from the quantum characteristic function and the other from generating functions. In the first approach, a gaussian state $\rho$ is described by the pair $(\m, S)$, where $\m$ is the mean annihilation vector and $S$ is the position-momentum covariance matrix. Such a description involves $2n+n(2n+1)= 2n^2+3n$ real parameters. The parametrization $(c,\bm{\mu},A,\Lambda)$ arising from the generating function has  $1+2n+n(n+1)+n^2= 2n^2+3n+1$ real parameters. We will see in this section that the parameter $c$ is  a normalization factor which is a function of the other parameters $\bm{\mu}, A$ and $\Lambda$.  It is natural to explore the relationship between the two parametrizations, particularly, in the context of tomography of gaussian states as well as quantum information theory in infinite dimensions.  It is shown that $\bm{\mu}$ is completely determined by the mean annihilation vector $\m$ of the state and the pair $(A,\Lambda)$ by the covariance matrix $S$. 
Propositions \ref{sec:gaussian-states} and \ref{prop:5.6}  describe the exact relationship between the two parametrizations. These results show that a mean zero gaussian state is completely determined by a pair  $(A,\Lambda)$  of  complex matrices with $A$ being symmetric  and $\Lambda$  positive. Coming to the uncertainty relations, we know that a $n$-mode, mean zero gaussian state is determined by a covariance matrix, i.e., a $2n\times 2n$ real symmetric matrix $S$ satisfying the uncertainty relations expressed by the matrix inequality, $S+i/2J\geq 0$. 
Theorem \ref{sec:gauss-stat-uncert-6} shows that a pair $(A,\Lambda)$ determines a mean zero gaussian state if and only if the real linear operator $I-\Lambda-2AC$ is strictly positive. 
This condition on the parameters $(A,\Lambda)$ expressed as a $2n\times 2n$ real matrix inequality can be viewed as the $\CE_2(\CH)$-version of the uncertainty relations.

\begin{prop}\label{sec:gaussian-states}  Let $\rho = \rho_{\m, S}$ be a gaussian state with its quantum characteristic function given by (\ref{eq:gaussian-qft-1}). Then the $\CE_2(\CH)$-parameters,  $(c,\bm{\mu},A,\Lambda)$ of $\rho$ satisfy the following: 
\begin{align}\label{eq:15.1}
 \begin{split}
c &= \det (\frac{1}{2}I+S)^{-1/2}\exp{\bmqty{\Re \m^T&\Im\m^T}J(\frac{1}{2}I+S)^{-1}J\bmqty{\Re \m \\\Im \m}},\\
\bm{\mu} &= i\bmqty{I&iI}(\frac{1}{2}I+S)^{-1}J\bmqty{\Re \m \\\Im \m}, \\ 
 A& = \frac{1}{4}\bmqty{I&iI}(\frac{1}{2}I+S)^{-1}\bmqty{I\\iI},\phantom{...........} \Lambda  = I-\frac{1}{2}\bmqty{I&iI}(\frac{1}{2}I+S)^{-1}\bmqty{I\\-iI}.\\
\end{split}
\end{align}
\end{prop}
\begin{proof}
  The expressions in (\ref{eq:15.1}) are  obtained by applying the gaussian integral formula to the last integral in equation (\ref{eq:gaussian-e2}).
\end{proof}
\begin{notn}
  Let $C$ denote the complex conjugation map on $\BC^n$, i.e. $C(\z) = \bar{\z}$. Let $C_{0}$ denote the $2n\times 2n$ real matrix corresponding to  the real linear map $C$ as in (\ref{eq:28}), then 
\begin{equation}
\label{eq:82}
C_0 = \bmqty{I&0\\0&-I}.
\end{equation}
Given a symmetric matrix $A$ and a positive matrix $\Lambda$ in $M_n(\BC)$, define the $2n\times 2n$ matrix \begin{align}
\label{eq:81}
M(A, \Lambda): = I -\bmqty{\Re \Lambda & -\Im \Lambda \\ \Im \Lambda & \phantom{-}\Re\Lambda} -2 \bmqty{\Re A & \phantom{-}\Im A \\ \Im A & -\Re A } = I-\Lambda_0-2A_0C_0,
\end{align}
where $I$ denotes the  $2n\times 2n$ identity matrix. Then $M(A, \Lambda)$ is a real symmetric matrix.   If $M(A, \Lambda)\geq 0$, then define \begin{equation}
    \label{eq:100}
    c(A,\Lambda):= \sqrt{\det{M(A, \Lambda)}}.
\end{equation}
\end{notn}
\begin{lem}
  \label{sec:gauss-stat-uncert-4}
Let   $A$ be a symmetric matrix and  $\Lambda$ be a positive matrix in $M_n(\BC)$, then
\begin{equation}
\label{eq:123}
J^TM(-A,\Lambda)J=M(A,\Lambda).
\end{equation}
 where the matrix $M(A, \Lambda)$ is  defined by (\ref{eq:81}). In particular, $M(A,\Lambda)$ is invertible if and only if $M(-A,\Lambda)$ is invertible.
\end{lem}
\begin{proof}
   A direct computation shows that $A_0C_0J=-JA_0C_0$, since $\Lambda_0J=J\Lambda_0$, we have (\ref{eq:123}). 
\end{proof}
\begin{prop} \label{prop:5.6}Let $(c,\bm{\mu},A,\Lambda)$ be  the $\CE_2(\CH)$-parameters of a gaussian state $\rho$. Then the covariance matrix $S$ and the mean annihilation vector $\m$ of $\rho$ satisfy the following: 
 \begin{enumerate}
  \item \label{item:8} The matrices $S, A$ and $\Lambda$ are related as follows,
\begin{equation}
\label{eq:50}
(\frac{1}{2}I+S)^{-1} = M(-A, \Lambda).
\end{equation}   

\item \label{item:9} The vectors $\bm{\mu}$ and $\m$ are related by the equation 
\begin{equation}
\label{eq:124}
(I-\Lambda-2AC) \m = \bm{\mu},
\end{equation}  where $C$ is the complex conjugation map described above. In other words, 
\begin{equation}
\label{eq:48}
\bmqty{\Re \m \\\Im \m } = M(A,\Lambda)^{-1}\bmqty{\Re \bm{\mu} \\ \Im \bm{\mu}}.
\end{equation}
In particular,  $\bm{\mu} = 0$   if and only if $\m=0$.  
  \end{enumerate}
\end{prop}

\begin{proof}
 \ref{item:8}. Write $(\frac{1}{2}I+S)^{-1} = \bmqty{P&Q\\Q^T&R}$ as a block matrix where $P,Q,R$ are of order $n\times n$. We now solve for $P, Q$ and $R$ from (\ref{eq:15.1}). From the expression for $A$ we get 
    $4A = (P-R) + i(Q+Q^T)$. 
Hence \begin{equation}\label{eq:13}
    P-R = 2(A+\bar{A}).
\end{equation}
From the expression for $\Lambda$ we get $2(\Lambda-I) = -(P+R) +i (Q-Q^T)$. Hence \begin{equation}\label{eq:14}
    P+R = 2I-(\Lambda+\Lambda^T).
\end{equation}
We get  $P$ and $R$ from equations (\ref{eq:13}) and (\ref{eq:14}). The matrix   $Q$ is obtained by substitution. Finally we have \begin{align*}
\begin{split}
  P& = I+ (A+\bar{A}) - \frac{\Lambda+\Lambda^T}{2}\\
  R &=  I - (A+\bar{A}) - \frac{\Lambda+\Lambda^T}{2}\\
Q &= -i\left[(A-\bar{A}) + \frac{\Lambda-\Lambda^T}{2} \right].
\end{split}
\end{align*}

\ref{item:9}. 
By comparing the real and imaginary parts on  both sides of the expression for $\bm{\mu}$ in (\ref{eq:15.1}), we get 
\begin{align*}
J\bmqty{\Re \bm{\mu} \\ \Im \bm{\mu} } = (\frac{1}{2}I+S)^{-1} J\bmqty{\Re \m \\\Im \m}.
\end{align*}
By part \ref{item:8}) of the proposition we have $(\frac{1}{2}I+S)^{-1} = M(-A,\Lambda)$. Equation (\ref{eq:123}) 
completes the proof. 
\end{proof}
\begin{cor}\label{sec:gauss-stat-uncert-1}
  Let  $(c, \bm{\mu}, A, \Lambda)$  be the  $\CE_2(\CH)$-parameters of a gaussian state $\rho$. Then 
\begin{equation}
\label{eq:89}
M(e^{i\theta}A,\Lambda)>0, \forall \theta \in \BR.
\end{equation}


\end{cor}
\begin{proof}
By Lemma \ref{lem:conj-sec-quant},  $(c,e^{i\theta/2}\bm{\mu}, e^{i\theta}A, \Lambda)$ are the  the $\CE_2(\CH)$-parameters of the state $\Gamma(e^{i\theta/2}\cdot I)\rho\Gamma(e^{-i\theta/2}\cdot I)$ for any $\theta\in \BR$. Proposition \ref{prop:5.6} provides the necessary conclusion. 
\end{proof}

If $\rho$ is a nonzero, positive and trace class operator in $\CE_2(\CH)$ then $(\tr \rho)^{-1}\rho$ is a gaussian state (Theorem \ref{thm:gaussian-iff-e2h}). In this case,  $M(A,\Lambda)$ is strictly positive (Corollary \ref{sec:gauss-stat-uncert-1}). The following proposition proves a converse and also provides the value of $\tr \rho$ in terms of the $\CE_2(\CH)$-parameters.
 \begin{prop}\label{sec:gaussian-states-1}
  Let $\rho$ be a nonzero positive element of $\CE_2(\CH)$ with parameters given by $(c,\bm{\mu},A,\Lambda)$, where $\bm{\mu}= \bm{\mu}_1+i\bm{\mu}_2$, $\bm{\mu}_1,\bm{\mu}_2\in \BR^n$. 
 Then $\rho$ is trace class if and only if $M(A,\Lambda)$ defined by (\ref{eq:81}) is strictly positive.  In this case, 
\begin{equation}
\label{eq:120}
\tr \rho = \frac{c}{c(A,\Lambda)}\exp{[\bm{\mu}_1^T,\bm{\mu}_2^T] M(A, \Lambda)^{-1}\bmqty{\bm{\mu}_1\\\bm{\mu}_2}}.
\end{equation} 
\end{prop}
\begin{proof}  
  Recall the identification of $z\in \CH$ with $\z = \x+i\y, \x,\y \in \BR^{n}$ fixed in Section \ref{sec:2}. By Lemma \ref{lem:semigr-ensur}, the positive operator $\rho$ is trace class if and only if 
\begin{align*}
 \frac{1}{\pi^n}\int\limits_{\CH}^{}\mel{\psi(z)}{\rho}{\psi(z)}\dd z <\infty,
\end{align*} and in this case, $\tr \rho $ is the value of the integral above.
By Lemma \ref{lem:integral-kernel-2}, 
\begin{align*}
\int\limits_{\CH}^{}\mel{\psi(z)}{\rho}{\psi(z)} \frac{\dd z}{\pi^n}& = \frac{1}{\pi^n}\int\limits_{\CH}^{}\exp{-\abs{z}^2}G_{\rho}(\bar{z},z) \dd z\\
& =  \frac{c}{\pi^n}\int\limits_{\BC^n}^{}\exp{ -\abs{\z}^2 + \bar{\z}^T\bm{\mu} +\bar{\bm{\mu}}^T\z + \bar{\z}^TA\bar{\z}+ \bar{\z}^T\Lambda \z +\z^T\bar{A}\z}\dd \z\\
& = \frac{c}{\pi^n}\int\limits_{\BR^{2n}}^{}\exp{-[\x^T,\y^T]M(A,\Lambda)\bmqty{\x\\\y} + 2[\bm{\mu}_1^T,\bm{\mu}_2^T]\bmqty{\x\\\y}}\dd \x \dd\y \numberthis \label{eq:84}.
\end{align*}
Hence  $\tr \rho$ is finite 
if and only if $M(A,\Lambda)>0$. Equation  (\ref{eq:120}) is obtained by applying the gaussian integral formula (\ref{eq:gif}) to (\ref{eq:84}).
\end{proof}

\begin{rmks}\label{rmk:gaussian-states} \begin{enumerate}
\item \label{item:38}  Proposition \ref{sec:gaussian-states-1} shows that the $\CE_2(\CH)$-parameter $c$ of a gaussian state is purely a function of the other three parameters $\bm{\mu}, A, \Lambda$ and 
\begin{equation}
\label{eq:30}
c = c(A,\Lambda)\exp{-\bmqty{\bm{\mu}_1^T &  \bm{\mu}_2^T} M(A, \Lambda)^{-1}\bmqty{\bm{\mu}_1\\\bm{\mu}_2}}.
\end{equation} In particular, if $\bm{\mu}=0$ then $c = c(A,\Lambda)$.
 \item  \label{item:23}
 If $\rho = \rho_{\m,S}$ is a gaussian state, then $W(-\m)\rho W(-\m)^{\dagger} = \rho_{0, S}$. If $(c, \bm{\mu}, A, \Lambda)$ are the  $\CE_2(\CH)$-parameters of $\rho_{\m, S}$ then that of the transformed state $W(-\m)\rho W(-\m)^{\dagger}$ 
are $(c(A,\Lambda), 0, A, \Lambda)$. Also, we know by (\ref{eq:124}) that $\m = (1-\Lambda-2AC)^{-1}\bm{\mu}$. 

\item\label{item:13} If $M(A, \Lambda)>0$ then by (\ref{eq:123}) $M(-A,\Lambda)>0$ and we  have $M(A,\Lambda)+ M(-A,\Lambda) = 2M(0,\Lambda)>0$. Hence $I-\Lambda_0>0,$ and by Example \ref{eg:complex-linear}, this is equivalent to the strict positivity  of $I-\Lambda$. Thus the positive matrix  $\Lambda$ is a strict contraction in this case. 

\item \label{item:22} By (\ref{eq:50}), if a pair $(A, \Lambda)$ determines a gaussian state $\rho$ then $M(-A,\Lambda) = (1/2I+S)^{-1}$, where $S$ is the covariace matrix of the state $\rho$. Since $S\pm i/2J\geq 0$, we have $M(-A,\Lambda)^{-1}-1/2(I\pm iJ)\geq 0$ on $\BC^{2n}$. It may also be noted that the projections $1/2(I+ iJ)$ and $1/2(I-iJ)$ are orthogonal to each other in $\BC^{2n}$. 
\end{enumerate}
\end{rmks}

Now we prove a restatement of the  uncertainty relation  $S+i/2J\geq 0$ satisfied by the covariance matrix of a gaussian state, in terms of the $\CE_2(\CH)$-parameters.
\begin{thm}\label{sec:gaussian-states-2}
  A pair  $(A,\Lambda)$  of complex matrices with $A$ being symmetric  and $\Lambda$  positive, determines the $\CE_2(\CH)$-parameters $(c(A,\Lambda), 0, A,\Lambda)$ of a gaussian state if and only if 
\begin{equation}
\label{eq:88}
M(-A,\Lambda)^{-1}-\frac{1}{2}(I-iJ)\geq 0,
\end{equation} where $M(-A, \Lambda)$ is defined by (\ref{eq:81}). 
\end{thm}
\begin{proof}
  Assume first that (\ref{eq:88}) holds. Then there exists a gaussian state $\rho$ with covariance matrix $S =M(-A,\Lambda)^{-1}-\frac{1}{2}I$. Now we get the desired result by Proposition \ref{sec:gaussian-states}. 
Converse part follows from item \ref{item:22} in Remarks \ref{rmk:gaussian-states}.
\end{proof}

 \begin{lem}\label{sec:descr-mean-zero-lem}
  Let $\rho = \ketbra{\psi}{\psi}$ be a mean zero pure gaussian state. 
Then there exists $L\in Sp(\CH)$ such that the $\CE_2(\CH)$-parameters of $\rho$ are given by  $(\alpha(L)^{-1/2}, 0, A_{\GOL},0)$, where $\alpha(L)$ and $A_{\GOL}$ are as in Corollary \ref{cor:generation-function-g0}.
\end{lem}
\begin{proof}
  By the structure theorem for gaussian states (Theorem 4, \cite{Par12}), there exists an $L\in Sp(\CH)$ such that $\ket{\psi} = \GOL\ket{\psi}$. Now 
\begin{align*}
G_{\rho}(u, v)& = \mel{e(\bar{u})}{\GOL}{\Omega}\mel{\Omega}{\GOL^{\dagger}}{e(v)}\\
& = G_{\GOL}(u,0)G_{\GOL^{\dagger}}(0, v)\\
& =  G_{\GOL}(u,0)\overline{G_{\GOL}(\bar{v},0)}\\
& =  \alpha(L)^{-1/2}\exp{u^T A_{\GOL}u +v^T\bar{A}_{\GOL}v},
\end{align*}
where the last line follows from Corollary \ref{cor:generation-function-g0}.
\end{proof}

\begin{thm}\label{thm:gps-1}
  Let $\rho$ 
  be a gaussian state with covariance matrix $S$ and  $\CE_2(\CH)$-parameters $(c,\bm{\mu}, A, \Lambda)$. Then $\rho$ is a pure state if and only if one of the following holds: 
\begin{enumerate}
\item\label{item:7}The  matrix $\Lambda  = 0$. 
\item \label{item:10} The covariance matrix $S$ satisfies the relation
\begin{equation}
\label{eq:52}
(\frac{1}{2}I+S)^{-1} = \bmqty{P&Q\\ Q & 2I-P}
\end{equation}
for some real symmetric matrices $P, Q$ of order $n$.
\end{enumerate}
\end{thm}
\begin{proof}
  Write $(\frac{1}{2}+S)^{-1} = \bmqty{P&Q\\Q^T&R}$ as a $2\times 2$ block matrix. Then the expression for $\Lambda$ in (\ref{eq:15.1}) shows that $\Lambda$ vanishes if and only if condition \ref{item:10}) of the theorem holds. The necessity of  condition \ref{item:7} follows from Lemma \ref{sec:descr-mean-zero-lem} and item \ref{item:23}) in Remarks \ref{rmk:gaussian-states}. To prove sufficiency, take $\CH = \BC^n$ and note that the condition $\Lambda = 0$ implies
\begin{equation}
\label{eq:53}
G_{\rho}(\bar{\bu}, \bv) = \overline{F(u)}F(v),
\end{equation} where $F(\x) = \sqrt{c}e^{\x^T\bar{A}\x}, \x \in \BC^n$. Expanding the left side of (\ref{eq:53}) in the particle basis and comparing coefficients we get the matrix elements in the particle basis as $\mel{\bk}{\rho}{\bbell} = \overline{\beta(\bk)}\beta(\bbell)$ for some function $\beta$ with $\sum_{\bk\in \BZ_+}^{}\abs{\beta(\bk)}^2 = \tr \rho = 1$. Hence $\rho$ is a rank one operator and thus a pure state.
\end{proof}
It is interesting to note a corollary of the theorem above despite the fact that it does not play a role in the later part of this article.
\begin{cor}  Let $S$ be the covariance matrix of a pure gaussian state, i.e., $S = \frac{1}{2} L^TL$ for some symplectic matrix $L\in Sp(2n, \BR)$. Let $\mathcal{P} = {1}/{2}(I+iJ), \mathcal{P}^{\perp} = {1}/{2}(I-iJ)$. Then $\mathcal{P}$ and $\mathcal{P}^{\perp}$ are mutually orthogonal  projections with the property that in the direct sum decomposition 
 $\BC^{2n} = \ran{\mathcal{P}}\oplus \ran{\mathcal{P}^{\perp}},$ the positive operator $({1}/{2}+S)^{-1}$ admits the block representation \begin{equation}\label{eq:covariance-c2n}
    \bmqty{I_\mathcal{\mathcal{P}} & Q\\ Q^{\dagger} & I_{\mathcal{P}^\perp}},
\end{equation} 
where $I_\mathcal{\mathcal{P}}$ and $I_{\mathcal{P}^\perp}$ are identity operators on $\ran{\mathcal{P}}$ and $\ran{\mathcal{P}^{\perp}}$ respectively and $Q:\ran{\mathcal{P}^{\perp}}\rightarrow\ran{\mathcal{P}}$ is the operator $\mathcal{P}({1}/{2}+S)^{-1}\left.\right \vert_{\ran \mathcal{P}^{\perp}}$. 
\end{cor}
  \begin{proof}
Since $S$ is the covariance operator of a pure gaussian state, the  matrix $\Lambda = 0$ in its $\mathcal{E}_2(\CH)$ representation. Hence from the expression for $\Lambda$ in (\ref{eq:15.1}),
 \begin{align*}
     \phantom{\Rightarrow} & \frac{1}{2}\bmqty{I&iI}(\frac{1}{2}+S)^{-1}\bmqty{I\\-iI} = I\numberthis \label{eq:thm.6.2}\\
     \Rightarrow & \frac{1}{2}\bmqty{I\\-iI}\bmqty{I&iI}(\frac{1}{2}+S)^{-1}\bmqty{I\\-iI}\bmqty{I&iI} = \bmqty{I\\-iI}\bmqty{I&iI}.
  \end{align*}    
  Since $\bmqty{I\\-iI}\bmqty{I&iI} = I+iJ,$ we get \begin{equation}\label{eq:cor-P-2}
      \mathcal{P}(\frac{1}{2}+S)^{-1}\mathcal{P} = \mathcal{P}.
  \end{equation}
 By doing a similar computation after taking transpose on both sides of (\ref{eq:thm.6.2}) we get 
 \begin{equation}\label{eq:cor-P-1}
      \frac{1}{2}(I-iJ)(\frac{1}{2}+S)^{-1}\frac{1}{2}(I-iJ) = \frac{1}{2}(I-iJ).
 \end{equation}
But $\mathcal{P}^\perp = I- \mathcal{P}= \frac{1}{2}(I-iJ)$. 
Hence (\ref{eq:cor-P-1}) is same as
\begin{equation*}\label{eq:cor-P-3}
    \mathcal{P}^\perp(\frac{1}{2}+S)^{-1}\mathcal{P}^\perp = \mathcal{P}^\perp.
\end{equation*}
This together with (\ref{eq:cor-P-2}) completes the proof. 
  \end{proof}
Now we turn to the characterization of  the $\CE_2(\CH)$-parameters of gaussian states. In Theorem \ref{sec:gaussian-states-2}, we saw uncertainty relations written in terms of the $\CE_2(\CH)$-parameters. That was a necessary and sufficient condition on the parameters $(A,\Lambda)$ to determine a mean zero gaussian state as an $\CE_2(\CH)$-element. Theorems \ref{sec:gauss-stat-uncert-5} and \ref{sec:gauss-stat-uncert-6} below provide much simpler and more elegant necessary and sufficient conditions on a pair $(A,\Lambda)$ as above to determine a mean zero gaussian state.
We need a lemma before that.
\begin{lem}\label{sec:gauss-stat-uncert-2}
    Let $A=A^T\in M_n(\BC)$ be such that $M(A,0)>0$. Then the matrix 
\begin{equation}
\label{eq:90}
L(A) := 2M(A,0)^{-1}-I
\end{equation}
is a $2n\times 2n$ strictly positive element of the group $Sp(2n, \BR)$ and $L(A)^{-1} = L(-A)$.
  \end{lem}
  \begin{proof}
By Lemma \ref{sec:gauss-stat-uncert-4}, 
$M(-A,0)$ is a strictly positive matrix. First we show that $L(A)$ is a strictly positive matrix. Since $-1/2I<A_0C_0$, 
we have 
$1/2(I-2A_0C_0)<I$. Hence $2M(A,0)^{-1}>I$. To prove that $L(A)$ is a symplectic matrix, recall from Lemma \ref{sec:gauss-stat-uncert-4} that $M(A,0)J = JM(-A,0)$ and hence $JM(A,0)^{-1} = M(-A,0)^{-1}J$. Now
\begin{equation}\label{eq:22}
L(A)^TJL(A) = L(A)L(-A)J.
\end{equation}
But
\begin{equation}\label{eq:92}
L(A)L(-A) = 4\left\{\left\{M(A,0)M(-A,0)\right\}^{-1}-\frac{1}{2}\left\{M(A,0)^{-1}+M(-A,0)^{-1}\right\}+\frac{1}{4}I \right\}.
\end{equation}
Since \[\left\{M(A,0)M(-A,0)\right\}^{-1} = \{I-(2A_0C_0)^2\}^{-1}=\frac{1}{2}\{M(A,0)^{-1}+M(-A,0)^{-1}\},\]  the right side of (\ref{eq:92}) is the identity matrix and (\ref{eq:22}) completes the proof. 
  \end{proof}
In Theorem \ref{thm:gps-1} we saw that the $\CE_2(\CH)$-parameters of a mean zero pure gaussian state is of the form $(c,0,A,0)$ where $A$ is a symmetric matrix. Our next theorem characterizes pure gaussian states using these parameters.
  \begin{thm}
    \label{sec:gauss-stat-uncert-5}
 Let $c>0$ and $A=A^T\in M_n(\BC)$. The following statements are equivalent: 
\begin{enumerate}
\item\label{item:34} The tuple $(c, 0, A, 0)$ are the $\CE_2(\CH)$-parameters of a pure gaussian state $\rho(A,0)$. 
\item \label{item:35} The matrix $M(A,0)>0$ and the constant $c=c(A,0)$, where $M(A, 0)$ is defined by (\ref{eq:81}) and $c(A,0)$ by (\ref{eq:100}).
\item \label{item:33} The matrix $2A$ is a strict contraction, i.e., $\norm{A}<\frac{1}{2}$ and the  constant  $c=c(A,0)$. 
\end{enumerate}
  \end{thm}
  \begin{proof}
    \ref{item:34} $\Leftrightarrow$ \ref{item:35}.
    The necessity of \ref{item:35}) is a special case of Corollary \ref{sec:gauss-stat-uncert-1} and item \ref{item:38}) in Remarks \ref{rmk:gaussian-states}. Conversely, assume that $M(A,0)>0$. 
The matrix $L(A)$ defined in Lemma \ref{sec:gauss-stat-uncert-2} is a strictly positive symplectic matrix. Therefore, $(1/2)L(A)= M(A, 0)^{-1}-I$ is the covariance matrix of a mean zero pure gaussian state (Proposition 3.10 in \cite{Par10}). Now, Proposition \ref{sec:gaussian-states} and Proposition \ref{prop:5.6} together complete the proof. 

\ref{item:35} $\Leftrightarrow$ \ref{item:33}.  We look upon the $n\times n$ complex matrix $A$ as an operator in three different ways: (i) $A$ as an operator in the complex Hilbert space $\BC^n$; (ii) as a real linear operator in $\BC^n$ considered as a $2n$-dimensional real Hilbert space with scalar product $\Re \bar{\z}\z'$, $\z,\z'\in \BC^n$; (iii) the real matrix $A_0$, where  $A_0=UAU^{-1}$, $U$ being the real linear orthogonal transformation $U(\x+i\y)= (\x,\y)$ from $\BC^n$ to $\BR^{2n}$. Now observe that all these three operators have the same norm and hence $\norm{A} = \norm{A_0}$. On the other hand $2A_0C_0$ is a real symmetric matrix and $M(A,0)$ is strictly positive. Thus $\norm{2A_0C_0}<1$. Since $C_0$ is orthogonal we have $\norm{2A} = \norm{2A_0}=\norm{2A_0C_0}<1$.
  \end{proof}
  Now we are ready to prove the main theorem of this section. In Lemma \ref{lem:conj-sec-quant} we saw that, for a positive operator $Z$ in $\CE_2(\CH)$ with parameters $(c, \bm{\mu}, A, \Lambda)$, the conjugation with a Weyl operator changes the parameters $c$ and $\bm{\mu}$ only. Item \ref{item:23} in Remarks \ref{rmk:gaussian-states}, asserts that for any gaussian state with $\CE_2(\CH)$-parameters $(c, \bm{\mu}, A, \Lambda)$, the vector $\bm{\mu}$ can be brought to $0$ using an appropriate Weyl operator. Thus a mean zero gaussian state is completely determined by the pair $(A,\Lambda)$.   Theorem \ref{sec:gauss-stat-uncert-6} completely characterizes these parameters.
  \begin{thm}\label{sec:gauss-stat-uncert-6}
    Let $c>0$ and $A, \Lambda \in M_n(\BC)$  with $A=A^T$ and $\Lambda\geq 0$. The following statements are equivalent: 
\begin{enumerate}
\item\label{item:36} The quadruple $(c, 0, A, \Lambda)$ are the $\CE_2(\CH)$-parameters of a  gaussian state $\rho(A,\Lambda)$.
\item \label{item:37} The matrix $M(A,\Lambda)>0$ and the scalar $c=c(A,\Lambda)$, where $M(A, \Lambda)$ is defined by (\ref{eq:81}) and $c(A,\Lambda)$ by (\ref{eq:100}).
\end{enumerate}
  \end{thm}
  \begin{proof}
     Necessity of item \ref{item:37}) follows from Corollary \ref{sec:gauss-stat-uncert-1} and item \ref{item:38}) in Remarks \ref{rmk:gaussian-states}. 
To prove the sufficiency, 
notice by item \ref{item:13})  in Remarks \ref{rmk:gaussian-states}
that $\Lambda$ is a strict contraction. 
Hence $M(A,0)>0$ in particular and thus by Theroem \ref{sec:gauss-stat-uncert-5}, there exists a pure gaussian state $\rho(A,0)$ with $\CE_2(\CH)$-parameters $(c(A,0),0,A,0)$. Hence by Theorem \ref{sec:gaussian-states-2}, $0\leq M(-A,0)^{-1}-\frac{1}{2}(I-iJ)$. Furthermore,  $0<M(A,\Lambda)\leq M(A,0)$ implies $ M(A,0)^{-1}\leq M(A,\Lambda)^{-1}$
. Therefore, \[0\leq M(-A,0)^{-1}-\frac{1}{2}(I-iJ)\leq M(-A,\Lambda)^{-1}-\frac{1}{2}(I-iJ).\] Theorem \ref{sec:gaussian-states-2} completes this part of the proof. 
  \end{proof}
\begin{cor}
 Let  $(c(A,\Lambda), 0, A,\Lambda)$ be the $\CE_2(\CH)$-parameters of  a gaussian state $\rho(A,\Lambda)$. Let $\Lambda'$ be a positive matrix such that $\Lambda'\leq \Lambda$. Then there exists a gaussian state $\rho(A,\Lambda')$ with $\CE_2(\CH)$ parameters $(c(A,\Lambda'), 0, A,\Lambda')$.
\end{cor}
\begin{proof}
  This follows from the relation  $M(A,\Lambda')\geq M(A,\Lambda)>0$.
\end{proof}

\section{Positive operators in \ensuremath{\CE_2(\CH)}} \label{sec:positive-operators}
In Section \ref{sec:semigroup}, we noticed that any positive operator in $\CE_2(\CH)$ is determined by a quadruple $(c, \bm{\mu}, A, \Lambda)$ in the sense of (\ref{eq:118}). Also, we know that a state is in $\CE_2(\CH)$ if and only if it is a gaussian state. This section is devoted to the study  of positive operators in $\CE_2(\CH)$.  


 A few notations are needed before we proceed. The conventions $0^0=1 = 0!$ are used in what follows. 
\begin{notn}
  Let $\Delta_n(\BZ_+) = \{R|R=[r_{ij}], r_{ij}\in \BZ_+ \forall i,j, r_{ij}=0, \forall i>j\}$ denote the set of all $n\times n$  upper triangular matrices with nonnegative integer entries. Given $R= [r_{ij}] \in \Delta_n(\BZ_+)$, let 
  \begin{equation}\label{eq:r-tilde}
      \tilde{r}_i := \sum\limits_{j=1}^{i}r_{ji}+\sum\limits_{j=i}^{n}r_{ij}, \phantom{..............}\tilde{\br}(R) := [\tilde{r}_1, \tilde{r}_2, \dots, \tilde{r}_n]^{T},
  \end{equation} 
  $\abs{R} := \sum_{i,j}^{}r_{ij}$, $R!:= \Pi_{i,j}r_{ij}!$. For $\bt\in \BZ_+^n$ define
  \begin{equation}
\label{eq:matrix-index}
     \Delta(\bt) := \left\{R\in \Delta_n(\BZ_+)\middle| \tilde{\br}(R) = \bt\right\}. 
\end{equation}
Since $\abs{\tilde{\br}(R)}$ is an even number, $\Delta(\bt) = \phi$, the empty set, whenever $\abs{\bt}$ is odd. For any $B = [b_{ij}]\in M_n(\BC)$, and $R = [r_{ij}]\in \Delta_n(\BZ_+)$, let $B^{\circ^R}$ $:=\Pi_{i,j}b_{ij}^{r_{ij}}$ (notice that $B^{\circ^R}$ takes account of the upper triangular entries of $B$ only). 
 With the convention that sum over an empty set is zero, we  define the function $\varphi_{B}:\BZ_+^n\rightarrow \BC$  by 
\begin{equation}
\label{eq:74}
\varphi_B(\bt) =
 \sqrt{\bt!} \sum\limits_{R\in \Delta(\bt)}^{}2^{\abs{R} - \tr R}\frac{B^{\circ^R}}{R!}. 
\end{equation}
Then $\varphi_B(\bt) = 0$ if $\abs{\bt}$ is an odd number. Furthermore, we write $\bt\leq\bs$ for two multi-indices $\bt,\bs\in \BZ_+^{n}$, $\bt =(t_1,t_2,\dots,t_n)$ and $\bs=(s_1,s_2,\dots,s_n)$ to mean $t_j\leq s_j, 1\leq j\leq n$.
\end{notn}
\begin{lem}\label{lem:technical-1}
  Let $\bm{\mu}\in \BC^n$, $B = [\beta_{ij}]\in M_n(\BC)$ be a symmetric matrix, and $\z\in \BC^n$. Then
\begin{align}
\label{eq:38}
\exp{\z^TB\z} 
&= \sum\limits_{\bs\in\BZ_{+}^n}^{}\frac{\varphi_B(\bs)}{\sqrt{\bs!}}\z^{\bs},\\ \label{eq:99}
\exp{\bm{\mu}^T\z+\z^TB\z} & = \sum\limits_{\underset{\bk\leq\bs}{\bk,\bs\in \BZ_{+}^n}}^{}\frac{\bm{\mu^k}}{\bk!}\frac{\varphi_B(\bs-\bk)}{\sqrt{(\bs-\bk)!}}\z^{\bs}.
\end{align}
\end{lem}
\begin{proof} First we prove (\ref{eq:38}). Let $\z = (x_1,x_2,\dots,x_n)$, since $B$ is a symmetric matrix, for $\ell\in \mathbb{N}$ we have \[(\z^TB\z)^{\ell} =(\sum_{i,j} \beta_{ij}x_{i}x_{j})^{\ell} = {(\sum_{i} \beta_{ii}x_{i}^{2}+2 \sum_{i<j}\beta_{ij}x_{i}x_{j})^{\ell}}.\]
   By the multinomial expansion in $\frac{n(n+1)}{2}$ summands on the right hand side we have
     \begin{align*}
        \frac{(\z^TB\z)^{\ell}}{\ell!}    & =   \sum\limits_{\sum\limits_{i\leq j} {r_{ij}=\ell}}\frac{{\Pi_i} x_{i}^{2r_{ii}}\beta_{ii}^{r_{ii}} {\Pi_{i<j}}(2\beta_{ij}x_{i}x_{j})^{r_{ij}}   }{\underset{i \leq j}{\Pi}r_{ij}!}\\
& =  \sum\limits_{\sum\limits_{i\leq j} {r_{ij}=\ell}} \frac{2^{\sum\limits_{i<j}^{}r_{ij}}}{\underset{i \leq j}{\Pi}r_{ij}!}\underset{i}{\Pi}\beta_{ii}^{r_{ii}}\underset{i<j}{\Pi}\beta_{ij}^{r_{ij}}\underset{j}{\Pi}x_j^{\tilde{r_j}},
     \end{align*}
 where $r_{ij} = 0$ whenever $\beta_{ij}=0$. Therefore, by using the notations described before the statement of the lemma, we have 
\begin{align*}
\frac{(\z^TB\z)^{\ell}}{\ell!} &= \sum\limits_{\underset{ r_{ij}=0 \textnormal{ whenever } \beta_{ij} = 0}{R=[r_{ij}]\in \Delta_n(\BZ_+),\abs{R}=\ell,}}^{}2^{\abs{R} - \tr R}\frac{B^{\circ^R}}{R!}\z^{\tilde{\br}(R)}\\
&= \sum\limits_{\underset{\abs{\bs}=2\ell}{\bs\in\BZ_+^n}}^{}\frac{\varphi_B(\bs)}{\sqrt{\bs!}}\z^{\bs}.
\end{align*}
This together with the fact that $\varphi_B(\bs)=0$ when $\abs{\bs}$ is odd, completes the proof of (\ref{eq:38}).
To see (\ref{eq:99}), notice first by using the multinomial theorem that $\exp{\bm{\mu}^T\z}= \sum_{\bk\in \BZ_+^n}^{}\frac{\bm{\mu}^{\bk}}{\bk!}\z^{\bk}.$
Now
\begin{align*}
\exp{\bm{\mu}^T\z+\z^TB\z}&=\exp{\bm{\mu}^T\z}\exp{\z^TB\z}\\
&= \sum\limits_{\bs,\bk}^{}\frac{\bm{\mu}^{\bk}}{\bk!}\frac{\varphi_B(\bs)}{\sqrt{\bs!}}\z^{\bs+\bk}\\
&=  \sum\limits_{\bk\leq\bs}^{}\frac{\bm{\mu}^{\bk}}{\bk!}\frac{\varphi_B(\bs-\bk)}{\sqrt{(\bs-\bk)!}}\z^{\bs}
\end{align*}
which completes the proof. 
\end{proof} 
Now we need to make sense of `second quantization' $\Gamma(B)$ for an arbitrary linear operator $B$ on $\CH$ and also, for  $\mu_j, b_{ij}\in\BC, 1\leq i,j\leq n$ and the annihilation operators (from Section \ref{sec:weyl-oper-wign}), $a_j, 1\leq j\leq n$,  $\exp{\sum_j\mu_ja_j+\sum_{i,j}b_{ij}a_ia_j}$. If $B$ is any linear operator on $\CH$, let $B^{\small{\text{\textcircled{s}}}^{k}}$ denote the restriction of $B^{\otimes^{k}}$ to $\CH^{\small{\text{\textcircled{s}}}^{k}}$ and 
for $\br= (r_1,r_2,\dots,r_n) \in \BZ_+^n$, we use the notations
\begin{align*}
B^{\small{\text{\textcircled{s}}}^{k}}:= B^{\otimes^{k}}_{|_{\CH^{\small{\text{\textcircled{s}}}^{k}}}}
&\textnormal{ and }\bm{a}^{\br}:=a_1^{r_1}a_2^{r_2}\cdots a_n^{r_n}.
\end{align*} On the finite particle domain $\mathcal{F}$,  we have 
\begin{align*}
B^{\small{\text{\textcircled{s}}}^{k}}\ket{z}^{\otimes^k}= \ket{Bz}^{\otimes^k}\textnormal{ and } & 
\bm{a}^{\br}\frac{\ket{\bk}}{\sqrt{\bk!}} = \frac{\ket{\bk-\br}}{\sqrt{(\bk-\br)!}},  \forall z\in \CH, \bk\in \BZ_+^n
\end{align*}
where $\ket{\bk-\br} := 0$ if $\bk-\br$ has a negative entry. 
 For any $\bm{\mu} = (\mu_1, \mu_2, \dots, \mu_n) \in \BC^n$ and $B=[\beta_{rs}] \in M_n(\BC)$, write $\bm{\mu}^T\bm{a}:=\sum_r^{}\mu_ra_r$ and $\bm{a}^{T}B\bm{a} := \sum_{r,s}^{}\beta_{rs}a_ra_s$  and define the linear operators $\Gamma(B)$ and $\exp{\bm{\mu}^T\bm{a}+\bm{a}^{T}B\bm{a}}$  on the finite particle domain $\mathcal{F}$ by 
 \begin{align}\label{eq:23}
   \begin{split}
\Gamma(B)\ket{z}^{\otimes^k} &= \ket{Bz}^{\otimes^k},\\
          \exp{\bm{\mu}^T\bm{a}+\bm{a}^{T}B\bm{a}}\frac{\ket{\bt}}{\sqrt{\bt!}} &= \sum\limits_{\bk\leq\bs\leq \bt}^{}\frac{\bm{\mu}^{\bk}}{\bk!}\frac{\varphi_B(\bs-\bk)}{\sqrt{(\bs-\bk)!}}\frac{\ket{\bt-\bs}}{\sqrt{(\bt-\bs)!}},
   \end{split}
 \end{align}
for all $z\in \CH$ and $\bt\in \BZ_+^n$. In particular, if $\bm{\mu}=0$, 
\begin{equation}
\label{eq:104}
 \exp{\bm{a}^{T}B\bm{a}}\ket{\bt} = \sum\limits_{\bs\leq\bt}^{}\sqrt{\binom{\bt}{\bs}}\varphi_B(\bt-\bs)\ket{\bs}.
\end{equation}
 \begin{prop}\label{sec:semigr-ensur-1}
   Let the linear operators $ \Gamma(B)$ and $\exp{\bm{\mu}^T\bm{a}+\bm{a}^{T}B\bm{a}}$  be defined on the finite particle domain $\mathcal{F}$ by (\ref{eq:23}). Then 
\begin{align}\label{eq:105}
  \begin{split}
    \lim_{N\rightarrow \infty}\Gamma(B)\left(\underset{k\leq N}{\oplus}\frac{\ket{z}^{\otimes^k}}{\sqrt{k!}}\right) &=\ket{e(Bz)}\\
 \lim_{N\rightarrow \infty}\exp{\bm{\mu}^T\bm{a}+\bm{a}^{T}B\bm{a}}\left(\underset{\abs{\bt}\leq N}{\oplus}\frac{\z^{\bt}\ket{\bt}}{\sqrt{\bt!}}\right)&= \exp{\bm{\mu}^T\z+\z^TB\z}\ket{e(z)},
  \end{split}
\end{align}
where $B$ in $\z^TB\z$ above denotes the matrix of $B$.
 \end{prop}
 \begin{proof}
   The first equation in (\ref{eq:105}) is immediate from (\ref{eq:23}). To prove the second, notice from (\ref{eq:23}) that
   \begin{align*}
     \exp{\bm{\mu}^T\bm{a}+\bm{a}^{T}B\bm{a}}\left(\underset{\abs{\bt}\leq N}{\oplus}\frac{\z^{\bt}\ket{\bt}}{\sqrt{\bt!}}\right)& = \sum\limits_{\abs{\bt}\leq N}^{}\z^{\bt} \sum\limits_{\bk\leq\bs\leq \bt}^{}\frac{\bm{\mu}^{\bk}}{\bk!}\frac{\varphi_B(\bs-\bk)}{\sqrt{(\bs-\bk)!}}\frac{\ket{\bt-\bs}}{\sqrt{(\bt-\bs)!}},
   \end{align*}
write $\bt-\bs = \m$ on the right side above so that $\bt=\bs+\m$ and

 \begin{align*}
     \exp{\bm{\mu}^T\bm{a}+\bm{a}^{T}B\bm{a}}\left(\underset{\abs{\bt}\leq N}{\oplus}\frac{\z^{\bt}\ket{\bt}}{\sqrt{\bt!}}\right)& = \sum\limits_{\abs{\m}\leq N}^{}\left(\sum\limits_{\abs{\bs}\leq N-\abs{\m}}
 \z^{\bs} \sum\limits_{\bk\leq\bs}^{}\frac{\bm{\mu}^{\bk}}{\bk!}\frac{\varphi_B(\bs-\bk)}{\sqrt{(\bs-\bk)!}}\right)\frac{\z^{\m}\ket{\m}}{\sqrt{\m!}}.
   \end{align*}
By Lemma \ref{lem:technical-1}, for each fixed $\m$ the coefficient of $\m!^{-1/2}{\z^{\m}\ket{\m}}$ in the equation above converges to $\exp{\bm{\mu}^T\z+\z^TB\z}$ as $N\rightarrow \infty$ and this completes the proof.
 \end{proof}
In the light of Proposition \ref{sec:semigr-ensur-1}  we extend the definition of the operators $\Gamma(B)$ and $\exp{\bm{\mu}^T\bm{a}+\bm{a}^{T}B\bm{a}}$ to the exponential domain $\CE$ by
\begin{align}
\label{eq:103}
  \begin{split}
    \Gamma(B)\ket{e(z)} &= \ket{e(Bz)},\\
 \exp{\bm{\mu}^T\bm{a}+\bm{a}^{T}B\bm{a}}\ket{e(z)}&= \exp{\bm{\mu}^T\z+\z^TB\z}\ket{e(z)}.
  \end{split}
\end{align}
Furthermore, for linear operators $B_1$ and $B_2$ on $\CH$ and $\bm{\mu}\in \BC$, the operator defined as $\Gamma(B_1)\exp{\bm{\mu}^T\bm{a}+\bm{a}^{T}B_2\bm{a}}$  on the exponential domain is a well defined linear operator with
\begin{equation}
\label{eq:106}
\Gamma(B_1)\exp{\bm{\mu}^T\bm{a}+\bm{a}^{T}B_2\bm{a}}\ket{e(z)}= \exp{\bm{\mu}^T\z+\z^TB_2\z}\ket{e(B_1z)}.
\end{equation}

\begin{rmk}
  Is the operator $\exp{\bm{\mu}^T\bm{a}+\bm{a}^{T}B\bm{a}}$ defined on the linear span of $\mathcal{F}\cup \CE$ as above a closable operator in general? We will see in what follows that the positive operators in the semigroup $\CE_2(\CH)$ provide examples of $B$'s where $\exp{\bm{\mu}^T\bm{a}+\bm{a}^{T}B\bm{a}}$ extends to a  bounded  operator on $\GH$.
\end{rmk}

\begin{thm}\label{thm:wick-ordering}
  Let $Z\geq 0 $ and $Z\in \CE_2(\CH)$ with parameters $(c,\bm{\mu}, A, \Lambda)$, 
Define 
 \begin{equation}\label{eq:1}
      Z_1 = \sqrt{c}\Gamma(\sqrt{\Lambda})\exp{\bar{\bm{\mu}}^T\bm{a}+ \bm{a}^T\bar{A}\bm{a}}
  \end{equation} on $\spn \mathcal{F}\cup \CE$ in the sense of (\ref{eq:23}) and (\ref{eq:106}). Then the following hold: 
\begin{enumerate}
\item\label{item:11}  The linear operator $Z_1$ closes to a bounded operator  on $\GH$, again denoted by $Z_1$ with $Z_1\in \CE_2(\CH)$ with parameters $(\sqrt{c}, 0, \bar{\mu}, 0,  \sqrt{\Lambda}, \bar{A})$.
\item\label{item:16} The operator $Z$ admits the factorization 
\begin{equation}
\label{eq:32}
Z = Z_1^{\dagger}Z_1.
\end{equation} 
\item \label{item:12} There exists a partial isometry $V$ 
such that 
\begin{equation}
\label{eq:33}
Z_1 = V \sqrt{Z}.
\end{equation}
In particular, $Z \in \Bo{\GH}$ (i.e., trace class) if and only if $Z_1 \in \Bt{\GH}$ (i.e., Hilbert-Schmidt).
\item \label{item:17} Let $\bm{\mu}=0$, then 
\begin{equation}
\label{eq:76}
Z_1\ket{\bt}  =  \sqrt{c} \sum\limits_{\{\bs:\bs\leq\bt\}}^{}\sqrt{\binom{\bt}{\bs}}\varphi_{\bar{A}}(\bt-\bs)\sqrt{\Lambda}^{\otimes^{\abs{\bs}}}\ket{\bs} 
\end{equation}
where $\varphi_{\bar{A}}$ is defined by (\ref{eq:74}). The matrix entry of  $Z$ in the $n$-mode particle basis, corresponding to 
$\bt = (t_1,t_2,\dots, t_n), \bt' = (t'_1,t'_2,\dots, t'_n)\in\BZ_+^n$ is given by
\begin{equation}
    \label{eq:44}
\mel{\bt}{Z}{\bt'} = c \sum\limits_{\left\{\bs,\bs':{\bs,\bs'\leq \bt\wedge\bt'},\abs{\bs}=\abs{\bs'}\right\}}^{}\sqrt{\binom{\bt}{\bs}}\varphi_A(\bt-\bs)\mel{\bs}{\Lambda^{\otimes^{\abs{\bs}}}}{\bs'}\sqrt{\binom{\bt'}{\bs'}}\overline{\varphi_{A}(\bt'-\bs')},
\end{equation} where the $j$-th coordinate of $\bt\wedge\bt'$ is $\operatorname{min}(t_j,t'_j), \forall j = 1,2,\dots,n$. In other words, the matrix representation of $Z$ in the particle basis is given by, 
\begin{equation}
\label{eq:107}
[Z]=c\left[E_A\right]\left[\Gamma(\Lambda)\right]\left[E_A\right]^{\dagger},
\end{equation}
where the $(\bt,\bs)$-th entry of the matrix  $[E_A]$ is given by 
\begin{equation}
\label{eq:108}
E_A(\bt,\bs)=
\begin{cases}
  \sqrt{\binom{\bt}{\bs}}\varphi_A(\bt-\bs), &\bs\leq\bt,\\
  \phantom{..........}0, & \textnormal{ otherwise.}
\end{cases}
\end{equation}
 \end{enumerate}
\end{thm}
\begin{proof} \ref{item:11}.
  By (\ref{eq:1}), 
\begin{equation}
\label{eq:31}
Z_1\ket{e(v)} = \sqrt{c}\exp{\bar{\bm{\alpha}}^T\bv+\bv^T\bar{A}\bv} \ket*{e(\sqrt{\Lambda}\bv)}
\end{equation} and thus 
\begin{equation}
\label{eq:26}
\braket{Z_1e(u)}{Z_1e(v)}  =  G_{Z}(\bar{u}, v) = \mel{e(u)}{Z}{e(v)} = \braket{\sqrt{Z}e(u)}{\sqrt{Z}e(v)}.
\end{equation}
Therefore, for any finite linear combination $\sum_{j=1}^k\beta_j\ket{e(v_j)}$ of exponential vectors, 
\begin{align*}
\|{Z_1\sum_{j=1}^k\beta_j\ket{e(v_j)}}\|^2&=\|{\sqrt{Z}\sum\limits_{j=1}^k\beta_j\ket{e(v_j)}}\|^2\\
&\leq \|\sqrt{Z}\|^2\|\sum_{j=1}^k\beta_j\ket{e(v_j)}\|^2.
\end{align*}
Hence the map $Z_1$ closes to a bounded operator on $\GH$. Now by (\ref{eq:31}), $Z_1\in \CE_2(\CH)$ with parameters $(\sqrt{c},0, \bar{\bm{\alpha}},0,\sqrt{\Lambda},\bar{A})$.

\ref{item:16}. By equation (\ref{eq:26}), $G_{Z_1^{\dagger}Z_1}(u,v) = G_{Z}(u,v)$ for all $u, v \in \CH$, hence $Z =  Z_1^{\dagger}Z_1$.

\ref{item:12}. Equation (\ref{eq:26}) asserts that the map $\sqrt{Z}\ket{e(v)}\mapsto Z_1\ket{e(v)}$ is scalar product preserving on $\ran \sqrt{Z}$. Hence it extends to a partial isometry $V$ with inital space $\overline{\ran}{\sqrt{Z}}$ satisfying $V \sqrt{Z}\ket{e(v)} = Z_1\ket{e(v)}$ for all $v \in \CH$. 

\ref{item:17}.
 Define the operator $Z_1$ as earlier by
\begin{equation}
\label{eq:35}
Z_1 = \sqrt{c}\Gamma(\sqrt{\Lambda})\exp{\bm{a}^T\bar{A}\bm{a}}
\end{equation} on  $\spn \mathcal{F}\cup\CE$. Then by part \ref{item:11} of the theorem,
$Z_1$ extends to a bounded operator.
By (\ref{eq:104}), we have (\ref{eq:76}). Since $\mel{\bt}{Z}{\bt'}=\braket{Z_1\bt}{Z_1\bt'}$, equation (\ref{eq:44}) follows from (\ref{eq:76}).
\end{proof}
\begin{rmk}
 It maybe noticed from (\ref{eq:108}) that the matrix $[E_A]$ appearing in (\ref{eq:107}) is a multiindex lower triangular matrix. Furthermore, since  $\Gamma(\Lambda)_{|_{\CH^{{\small{\text{\textcircled{s}}}^{k}}}}}=\Lambda^{\otimes^k}_{|_{\CH^{{\small{\text{\textcircled{s}}}^{k}}}}}$, $\Gamma(\Lambda)$ leaves ${\CH^{{\small{\text{\textcircled{s}}}^{k}}}}$ invariant for $k=0,1,2,\dots,$ and the matrix of $\Gamma(\Lambda)$ is a block diagonal matrix in the particle basis.
\end{rmk}

\begin{thm}\label{cor:criterion-on-a-lambda}
  Let $c\in \BC, \bm{\alpha}\in \BC^n, A, \Lambda\in M_n(\BC)$, where $A$ is complex symmetric and $\Lambda$ is positive matrix. Then there exists a positive operator $Z \in \CE_2(\CH)$ with parameters  $(c,\bm{\alpha}, A, \Lambda)$ if and only if the operator $Z_1$ defined by (\ref{eq:1}) on $\CE$ extends to a bounded operator on $\GH$.
\end{thm}
\begin{proof}
 If there exists a positive operator $Z \in \CE_2(\CH)$ with parameters  $(c,\bm{\alpha}, A, \Lambda)$ then Theorem \ref{thm:wick-ordering} provides the required result. Conversely, if  $Z_1$ defined by (\ref{eq:1}) on $\CE$ extends to a bounded operator on $\GH$ which is again denoted by $Z_1$, then $Z_1^\dagger Z_1$ is the required operator.
\end{proof}

\begin{rmk}
  The factorization of $Z$ in (\ref{eq:32}) of Theorem \ref{thm:wick-ordering} tempts us to express a gaussian state $\rho$ with $\CE_2(\CH)$-parameters  $(c,0, A, \Lambda)$, where $A=[\alpha_{rs}]$ as
  \begin{equation*}
    \rho = c\exp{\sum\limits_{r,s}^{}\alpha_{r,s}a_r^{\dagger}a_s^{\dagger}}\Gamma(\Lambda)\exp{\sum\limits_{r,s}^{}\bar{\alpha}_{r,s}a_ra_s}
  \end{equation*}
on the exponential domain $\CE$ but the operator $\exp{\sum_{r,s}^{}\alpha_{r,s}a_r^{\dagger}a_s^{\dagger}}$ makes sense only when $A$ is in a special region in the space of complex symmetric matrices of order $n$.
\end{rmk}

We had seen in Proposition \ref{sec:gaussian-states-1} that a positive operator $Z\in \CE_2(\CH)$ with parameters  $(c,\bm{\mu},B,\Lambda)$ is trace class if and only if the matrix  $M(B,\Lambda)$ defined by (\ref{eq:81}) is strictly positive and in this case $\Lambda$ must be a strict contraction (item \ref{item:13}) of Remarks \ref{rmk:gaussian-states}).  The following lemma analyses the situation for a general positive element in $\CE_2(\CH)$. 
\begin{lem}\label{sec:posit-oper}
  Let $Z\geq 0$ and $Z\in \CE_2(\CH)$ with parameters $(c,\bm{\mu},B,\Lambda)$. Let $M(B,\Lambda)$ be the $2n\times 2n$ real symmetric matrix defined by (\ref{eq:81}). Then, 
\begin{enumerate}
\item\label{item:43} $M(B,0)>0$  or equivalently $\norm{B}<\frac{1}{2}$; 
\item \label{item:44} $M(B,\Lambda)\geq 0$; 
\item \label{item:45} $\Lambda$ is a contraction;
\item \label{item:46} if $\Lambda=D$, where $D= \diag{(\lambda_1,\lambda_2,\dots,\lambda_{n_0}, \underset{n_1-\textnormal{times}}{1, 1,\dots, 1})}$, $\lambda_j<1$, $1\leq j \leq n_0$,  $n_1 = n-n_0$, then in the direct sum decomposition $\BC^n=\BC^{n_0}\oplus \BC^{n_1}$, the matrix $B$ has the block diagonal form
\begin{equation}
\label{eq:125}
B= \begin{blockarray}{cc}
			\begin{block}{[c|c]}
			B_{00} & 0 \\ \BAhline 0 & 0 \\
			\end{block}
			\end{blockarray}.
                      \end{equation}
                      
\end{enumerate}
\end{lem}
\begin{proof}
  \ref{item:43}. Write $Z= Z_1^{\dagger}Z_1$ as in Theorem \ref{thm:wick-ordering}, then $Z_1^0= \Gamma(0)Z_1$ is trace class and  $(Z_1^0)^{\dagger}Z_1^0$ is a positive trace class element in $\CE_2(\CH)$ with parameters $(c, \bm{\mu}, B, 0)$. Hence by Proposition \ref{sec:gaussian-states-1}, $M(B,0)>0$.

 \ref{item:44}. Let $0<\theta<1$, then $\Gamma(\theta\cdot I)$ is trace class and thus $\Gamma(\theta\cdot I)Z\Gamma(\theta\cdot I)^{\dagger}$ is a positive trace class element of $\CE_2(\CH)$ with parameters $(c, \bm{\mu}, \theta^2B, \theta^2\Lambda)$. Hence by Proposition \ref{sec:gaussian-states-1}, $M(\theta^2B, \theta^2\Lambda)>0$. The result follows because $M(\theta^2B, \theta^2\Lambda)\rightarrow M(B,\Lambda)$ as $\theta\rightarrow 1$. 

\ref{item:45}. Follows from \ref{item:44}) as in item \ref{item:13}) of Remarks \ref{rmk:gaussian-states}.

\ref{item:46}. Let $B = \bmqty{B_{00}& B_{01}\\ B_{10}& B_{11}}$ in the decomposition $\BC^n= \BC^{n_0}\oplus  \BC^{n_1}$ and $D_{\bm{\lambda}}$ be the $n_0\times n_0$ diagonal matrix given by  $D_{\bm{\lambda}}= \diag{(\lambda_1,\lambda_2,\dots,\lambda_{n_0})}$, then  
by  definition, 
\[M(B,\Lambda)=\begin{blockarray}{cccc}
			\BR^{n_0} & \BR^{n_1}&\BR^{n_0} & \BR^{n_1} \\
			\begin{block}{[cc|cc]}
I_{n_0}-D_{\bm{\lambda}}-2\re B_{00}&-2\re B_{01}&-2\im B_{00}&-2\im B_{01}\\-2\re B_{10}&-2\re B_{11}&-2\im B_{10}&-2\im B_{11}\\\BAhline -2\im B_{00}&-2\im B_{01}&I_{n_0}-D_{\bm{\lambda}}+2\re B_{00}&\phantom{-}2\re B_{01}\\-2\im B_{10}&-2\im B_{11}&\phantom{-}2\re B_{10}&\phantom{-}2\re B_{11}\\
			\end{block}
			\end{blockarray}.\]
 By \ref{item:44}) we know that $M(B,\Lambda)\geq 0$. This implies in particular that the real symmetric matrices $-2\re B_{11}$ and $2\re B_{11}$ both are positive matrices. So $2\re B_{11} = 0$ and 
by the positivity of $M(B,\Lambda)$,
we see that $ B_{01},B_{10}$ and $B_{11}$ are all zero matrices.
\end{proof}
\begin{thm}
  Let $Z\geq 0$ and $Z\in \CE_2(\CH)$ with parameters $(c, \bm{\mu}, A, \Lambda)$.  Then $Z$ lies in the weak-closure of the set of all positive scalar multiples of gaussian states.
\end{thm}
\begin{proof}
  Let $0<\theta<1$, define $\rho_{\theta} =\Gamma(\theta\cdot I)Z\Gamma(\theta\cdot I)^{\dagger}$ as in the proof of \ref{item:44}) in Lemma \ref{sec:posit-oper}, then $\rho_{\theta}$ is a  positive scalar multiple of a gaussian state for all $\theta\in (0,1)$. Observe that $\rho_{\theta}\rightarrow Z$ weakly as $\theta\rightarrow 1$.
\end{proof}
\begin{rmk} 
  Let  $Z \in \CE_2(\CH)$ be a positive operator with parameters $(c, 0, A,\Lambda)$. By Lemma \ref{sec:posit-oper}, $\Lambda$ is a positive contraction.
Then there exists a unitary $U$ such that  $\Gamma(U)Z\Gamma(U)^{\dagger}$ has $\CE_2(\CH)$-parameters given by $(c,0, B, D)$, where $D$ is a diagonal contraction and $B=UAU^T$.  If $D= \diag{(D_{\bm{\lambda}}, I_{n_1})}$ as in item \ref{item:46}) of Lemma \ref{sec:posit-oper}, then write $\CH=\CH_0\oplus\CH_1$. Now keeping the same notations as in item \ref{item:46}) of Lemma \ref{sec:posit-oper}, \[\Gamma(U)Z\Gamma(U)^{\dagger} = Z_0\otimes\Gamma(I_{n_1}),\] where  $Z_0\in \CE_2(\CH_0)$ is a positive operator with  $\CE_2(\CH_0)$-parameters $(c,0, B_{00},D_{\bm{\lambda}})$, $\Gamma(I_{n_1})$ is the identity operator on $\Gamma(\CH_1)$. The operator $Z_0$ is a positive scalar multiple of a gaussian state if and only if $M(B_{00}, D_{\bm{\lambda}})>0$. 
\end{rmk}

\section{A Density Matrix Formula for Gaussian States}\label{sec:arch-gauss-stat}

In Section \ref{sec:gaussian}, we analysed the $\CE_2(\CH)$-parameters of a gaussian state and found that a mean zero gaussian state is completely determined by a pair of $n\times n$ complex matrices $(A,\Lambda)$, $A$ being symmetric and $\Lambda$ being positive. A necessary and sufficient condition on the pair $(A,\Lambda)$ to determine a  gaussian state is the property $M(A,\Lambda)>0$, (Theorem \ref{sec:gauss-stat-uncert-6}) where $M(A, \Lambda)$ is defined by (\ref{eq:81}). Furthermore, it was also observed that the constant parameter $c$ in the $\CE_2(\CH)$ parametrization of a mean zero gaussian state has to be $c(A,\Lambda):= \sqrt{\det M(A,\Lambda)}$.  In Section \ref{sec:positive-operators}, we described an important factorization property of positive operators in $\CE_2(\CH)$. 
 In this section, we see the consequences of Theorem \ref{thm:wick-ordering} in the case of gaussian states. We draw the attention of the reader to equations (\ref{eq:56}) and (\ref{eq:39}) which describe respectively the particle basis expansion for a  mean zero pure gaussian state and the density matrix formula (DMF) for an arbitrary mean zero gaussian state in the same basis. 
These results lead to Theorem \ref{thm:gaussian-state-architecture}, which describes the architecture of a gaussian state. 

\begin{prop}\label{prop:rank-one-proj-e2}
If $c>0$, 
and $A=[\alpha_{ij}]$ is a complex symmetric matrix, then  $(c, 0, A, 0)$ are the  $\CE_2(\CH)$-parameters of a positive operator $Z\in \B{\GH}$ if and only if 
\begin{equation}\label{eq:80}
 \sum\limits_{\bt \in \BZ_+^n}^{}\abs{\varphi_A(\bt)}^2<\infty,
\end{equation}where $\varphi_A$ is given by (\ref{eq:74}).
In this case, $Z = \ketbra{\psi}$, where $\ket{\psi}=\sqrt{c} \sum_{{\bt\in\BZ_+^n}}\varphi_{A}(\bt)\ket{\bt}$ and 
\begin{align}
\label{eq:56}
\ket{\psi_A} &=\sqrt{c(A,0)} \sum\limits_{{\bt\in\BZ_+^n}}\varphi_{A}(\bt)\ket{\bt}
\end{align} is a pure gaussian state.
\end{prop}
\begin{proof}
 Assume first that $(c, 0, A, 0)$ are the $\CE_2(\CH)$-parameters of a positive operator $Z$.  Then by Theorem \ref{thm:wick-ordering}, $Z_1^0= \sqrt{c}\Gamma(0)\exp{\bm{a}^T\bar{A}\bm{a}}$  defined on $\spn \mathcal{F}\cup \CE$ in the sense of (\ref{eq:23}) and (\ref{eq:106}) extends to a bounded operator on $\GH$. Furthermore, by (\ref{eq:104}), \begin{equation*}
    Z_1^0 \ket{\bt} = \sqrt{c}\varphi_{\bar{A}}(\bt)\ket{\Omega}, \forall \bt \in \BZ_+^n.
  \end{equation*}
  In other words, $Z_1^0$ is the rank one operator $\ketbra{\Omega}{\psi}$, 
  where $\braket{\psi}{\bt} = \sqrt{c}\varphi_{\bar{A}}(\bt), \forall \bt \in \BZ_+^n$. So  (\ref{eq:80}) is satisfied and $Z=\ketbra{\psi}$. Conversely, if (\ref{eq:80}) is satisfied, define $\ket{\psi}=\sqrt{c} \sum_{{\bt\in\BZ_+^n}}\varphi_{A}(\bt)\ket{\bt}$. Then $Z=\ketbra{\psi}\in \CE_2(\CH)$ with the required properties. The state defined by (\ref{eq:56}) is a gaussian state because of item \ref{item:38}) in Remarks \ref{rmk:gaussian-states}. 
\end{proof}

\begin{cor} If $(c(A,0),0,A,0)$ are the $\CE_2(\CH)$-parameters of a mean zero pure gaussian state $\ketbra{\psi_A}$, then 
  \begin{equation}
      \label{eq:87}
\varphi_A(\bt) = \frac{\braket{\bt}{\psi_A}}{\sqrt{c(A,0)}}.
  \end{equation}
\end{cor}
\begin{proof}
Equation (\ref{eq:87}) is immediate from (\ref{eq:56}).
\end{proof}

\begin{prop}
  Let $c>0$ and $A, \Lambda \in M_n(\BC)$  with $A=A^T$ and $\Lambda\geq 0$. The tuple $(c, 0, A, \Lambda)$ are the $\CE_2(\CH)$-parameters of a  gaussian state $\rho(A,\Lambda)$ if and only if the operator $Z_1$ defined by 
\begin{equation} 
\label{eq:36}
Z_1 = \sqrt{c}\Gamma(\sqrt{\Lambda})\exp{\bm{a}^T\bar{A}\bm{a}}
\end{equation}
on $\spn \mathcal{F}\cup \CE$ in the sense of (\ref{eq:23}) and (\ref{eq:106}) extends to a Hilbert-Schmidt operator $Z_1^{A,\Lambda}$ on $\GH$ and $\tr (Z_1^{A,\Lambda})^{\dagger}(Z_1^{A,\Lambda}) =1$. In this case, $\rho(A,\Lambda)=(Z_1^{A,\Lambda})^{\dagger}(Z_1^{A,\Lambda})$.
\end{prop}
\begin{proof}
  This follows from Theorem \ref{thm:wick-ordering} and Theorem \ref{cor:criterion-on-a-lambda}.
\end{proof}
\begin{thm}\label{thm:DMF}
Let $\rho$ be a mean zero gaussian state. Then there exists a pair of $n\times n$ complex matrices  $(A, \Lambda)$, $A$ being symmetric and $\Lambda$  positive, such that the matrix representation of $\rho=\rho(A,\Lambda)$ in the particle basis is given by the density matrix formula (DMF), 
\begin{equation}\label{eq:39}
\rho^{\operatorname{mat}}(A,\Lambda)=c(A,\Lambda)\left[E_A\right]\left[\Gamma(\Lambda)\right]\left[E_A\right]^{\dagger},
\end{equation}
where $ c(A,\Lambda)$ is defined by (\ref{eq:100}), and the $(\bt,\bs)$-th entry of the matrix  $[E_A]$ is given by 
\begin{equation}\label{eq:40}
E_A(\bt,\bs)=
\begin{cases}
  \sqrt{\binom{\bt}{\bs}}\varphi_A(\bt-\bs), &\textnormal{ if }\bs\leq\bt,\\
  \phantom{..........}0, & \textnormal{ otherwise,}
\end{cases}
\end{equation} with
$\varphi_A$  defined by (\ref{eq:74}).
\end{thm}
 \begin{proof}

Equation (\ref{eq:39}) follows from item \ref{item:17}) in  Theorem \ref{thm:wick-ordering}.
  \end{proof}
 
\begin{defn}
  Given $\bs\in \BZ_+^n$, consider the shift isometry  $S_{\bs}$ at $\bs\in \BZ_+^n$ defined by $S_{\bs}\ket{\bt} = \ket{\bt+\bs}$ and the corresponding homomorphism $\tau_{\bs}$ on $\B{\GH}$ defined by $\tau_{\bs}(X) = S_{\bs}XS_{\bs}^{\dagger}$, $X\in \B{\GH}$. It may be noticed that $\{\tau_{\bs}|\bs\in \BZ_+^n\}$ is an $n$-parameter semigroup of homomorphisms on $\B{\GH}$. Given a vector $\bm{\lambda} = (\lambda_1,\lambda_2,\dots,\lambda_n)\in \BR^n_+$,  define the mixing kernel $M_{\bm{\lambda}}$ on $\BZ_+^n$ by 
\begin{equation}
\label{eq:77}
M_{\bm{\lambda}}(\bt,\bt') = \sum\limits_{\bs\leq \bt\wedge \bt'}^{}\sqrt{\binom{\bt}{\bs}\binom{\bt'}{\bs}}\bm{\lambda}^{\bs}\tau_{\bs},
\end{equation}where $\bs\leq \bt\wedge\bt'$ is meant  entrywise.
\end{defn}

\begin{thm}\label{thm:gaussian-state-architecture}
  Let $D_{\bm{\lambda}}=\diag{(\lambda_1,\lambda_2,\dots,\lambda_n)}$ be a $n\times n$ positive diagonal matrix  such that $(c(A,D_{\bm{\lambda}}), 0,A,D_{\bm{\lambda}})$ are the $\CE_2(\CH)$-parameters of a gaussian state $\rho(A,D_{\bm{\lambda}})$. Let $\rho(A,0) = \ketbra{\psi_A}$ with $\ket{\psi_A}$ as in (\ref{eq:56}). 
Then $\rho(A,0)$ is a pure gaussian state with $\CE_2(\CH)$-parameters  $(c(A,0), 0, A, 0)$ and the matrix entries of $\rho(A,D_{\bm{\lambda}})$ are given by
\begin{equation}
\label{eq:86}
\mel{\bt}{\rho(A,D_{\bm{\lambda}})}{\bt'} = \frac{c(A,D_{\bm{\lambda}})}{c(A,0)}\mel{\bt}{M_{\bm{\lambda}}(\bt,\bt')(\rho(A,0))}{\bt'},\forall \bt,\bt'\in \BZ_+^n.
\end{equation}
\end{thm}
\begin{proof}
 By (\ref{eq:44}), and Proposition \ref{sec:gaussian-states-1}  we have \begin{equation}
    \label{eq:44-repeat}
\mel{\bt}{\rho(A,D_{\bm{\lambda}})}{\bt'} = c(A, D_{\bm{\lambda}}) \sqrt{\bt!\bt'!}\sum\limits_{\bs\leq \bt\wedge\bt' }^{}\frac{\varphi_A(\bt-\bs)}{\sqrt{(\bt-\bs)!}}\frac{\varphi_{\bar{A}}(\bt'-\bs)}{\sqrt{(\bt'-\bs)!}}\frac{\bm{\lambda}^{\bs}}{\bs!},
\end{equation} Now by  (\ref{eq:87}), 
\begin{align*}
\mel{\bt}{\rho(A,D_{\bm{\lambda}})}{\bt'}& =\frac{c(A,D_{\bm{\lambda}})}{c(A,0)} \sqrt{\bt!\bt'!}\sum\limits_{\bs\leq \bt\wedge\bt' }^{}\frac{\braket{\bt-\bs}{\psi_A}}{\sqrt{(\bt-\bs)!}}\frac{\braket{\psi_A}{\bt'-\bs}}{\sqrt{(\bt'-\bs)!}}\frac{\bm{\lambda}^{\bs}}{\bs!}\\
& = \frac{c(A,D_{\bm{\lambda}})}{c(A,0)}\sum\limits_{\bs\leq \bt\wedge\bt'}^{}\sqrt{\binom{\bt}{\bs}\binom{\bt'}{\bs}}\bm{\lambda}^{\bs}\mel{\bt}{S_{\bs}\ketbra{\psi_A}S_{\bs}^{\dagger}}{\bt'}.
\end{align*} This is same as (\ref{eq:86}).
\end{proof}
\begin{rmks} 
\begin{enumerate}
\item\label{item:15} The total particle number of any basis vector of the form $\ket{\bt}$ is decreased by the strictly upper triangular matrix $E_A^{\dagger}-I$, preserved by the block diagonal matrix $\Gamma(\Lambda)$ and increased by the strictly lower triangular matrix $E_A-I$ in their respective actions. Thus our  DMF in (\ref{eq:39}) as a factorization preserves the spirit of Wick ordering 
in quantum stochastic calculus \cite{Par12}. 
\item \label{item:32} Theorem \ref{thm:gaussian-state-architecture} throws light on the architecture of a general gaussian state. Let $\rho$ be a gaussian state with $\CE_2(\CH)$-parameters $(c(\bm{\alpha},A',\Lambda'), \bm{\alpha}, A',\Lambda')$. Conjugating $\rho$ with an appropriate Weyl operator followed by the second quantization of a unitary matrix operator in $\CH$, as suggested by Remarks \ref{rmk:gaussian-states} and Lemma \ref{lem:conj-sec-quant} respectively, transforms $\rho$ to a gaussian state with $\CE_2(\CH)$ parameters $(c(A,D_{\bm{\lambda}}), 0,A,D_{\bm{\lambda}})$, where $D_{\bm{\lambda}}$ is a positive diagonal matrix. By Theorem \ref{thm:gaussian-state-architecture}, this transformed state is completely determined by the action of a positivity and trace class preserving kernel $M_{\bm{\lambda}}(\bt,\bt')$ on the pure gaussian state $\ketbra{\psi_A}$ which is constructed from the $2$-particle annihilation amplitude matrix $A$ (cf. Section \ref{sec:2} and equation (\ref{eq:61})).
\end{enumerate}
\end{rmks}
\section{Completely Entangled Pure Gaussian States and a Criterion for Testing Entanglement}\label{sec:gps-particle}
In this section, we discuss some interesting examples of gaussian states using the $\CE_2(\CH)$-parameters. An easy condition to check the entanglement of a pure gaussian state is obtained in Corollary \ref{sec:compl-entangl-pure-2}. Furthermore, a whole class of completely entangled pure gaussian states is obtained. This yields examples of such entangled states which are also
invariant under the action of the permutation group $S_n$ on the set of all the $n$ modes.
\begin{eg}\label{eg:1-mode-pure}[$1$-mode mean zero pure gaussian states] Let $n=1$, $\alpha\in \BC$. Then by Theorem \ref{sec:gauss-stat-uncert-5}, $(c,0,\alpha,0)$ are the $\CE_2(\CH)$-parameters of a pure gaussian state $\rho(\alpha,0)$ if and only if $\abs{\alpha}<1/2$, $c=c(\alpha,0)=(1-4\abs{\alpha}^2)^{1/2}$. The index set $\Delta(t)$ defined by (\ref{eq:matrix-index}) is given by
\begin{align*}
\Delta(t)= \begin{cases}
\phi, & \textnormal{ if } t \textnormal{ is odd}\\
\{t/2\}, & \textnormal{ if } t \textnormal{ is even},
\end{cases} 
\end{align*} where $\phi$ denote the empty set. Now the function $\varphi_\alpha(t)$ defined by (\ref{eq:74}) is given by
\begin{align*}
\varphi_\alpha(t)= \begin{cases}
\phantom{......}0, & \textnormal{ if } t \textnormal{ is odd}\\
\binom{t}{t/2}^{1/2}\alpha^{t/2}, & \textnormal{ if } t \textnormal{ is even}.
\end{cases} 
\end{align*} Therefore the gaussian state
$\rho(\alpha,0)= \ketbra{\psi_\alpha}$, where 
\begin{equation}
\label{eq:67}
\ket{\psi_{\alpha}} = (1-4\abs{\alpha}^2)^{1/4}\sum\limits_{t\in \BZ_+}^{}\binom{2t}{t}^{1/2}\alpha^t\ket{2t}.
\end{equation} Equation (\ref{eq:67}) has the following interpretation: if the observable $a^{\dagger}a$ (which measures the number of particles) is measured in the state $\ket{\psi_{\alpha}}$ then the possible outcomes are $0,2,4,\dots,2t,\dots$ and the probability for the outcome $2t$ is equal to $(1-4\abs{\alpha}^2)^{1/2}\binom{2t}{t}\abs{\alpha}^{2t}$. Simple algebra shows that this is equal to \[\operatorname{Pr}(\{2t\})=\sqrt{(1-4\abs{\alpha}^2)}\frac{\frac{1}{2}(\frac{1}{2}+1)\cdots(\frac{1}{2}+(t-1))}{t!}(4\abs{\alpha}^2)^t, t=0,1,2,\dots, \]
where the right hand side as a function of $t$ on $\BZ_+$ is the well known negative binomial distribution \cite{feller1} on $\BZ_+$ with  index $-1/2$ and parameter $p=4\abs{\alpha}^2$, $0<p<1$. This example with $\alpha\neq 0$, is known as \emph{single mode squeezed vacuum} (SMSV) state.
\end{eg}
\begin{eg}\label{eg:1-mode-mixed}[$1$-mode mean zero mixed gaussian states]  Let $n=1$, $\alpha\in \BC$, $\lambda>0$. Then by Theorem \ref{thm:DMF}, $(c,0,\alpha,\lambda)$ are the $\CE_2(\CH)$-parameters of a gaussian state $\rho(\alpha,\lambda)$ if and only if $\abs{\alpha}<(1-\lambda)/2$, $c=c(\alpha,\lambda) = \{(1-\lambda)^2-4\abs{\alpha}^2\}^{1/2}$. Then by the DMF (\ref{eq:39}) 
\begin{equation}
\label{eq:41}
\rho^{\operatorname{mat}}(\alpha,\lambda)(t,t') = \{(1-\lambda)^2-4\abs{\alpha}^2\}^{1/2}\sum\limits_{\underset{t-s, t'-s \textnormal{ even}}{s\leq t\wedge t'}}^{}\frac{\sqrt{t!t'!}}{s!(\frac{t-s}{2})!(\frac{t'-s}{2})!}\bar{\alpha}^{\frac{t-s}{2}}\alpha^{\frac{t'-s}{2}}\lambda^s.
\end{equation}
Thus, in the particle basis measurement, the probability for a $t$-particle count is equal to 
\begin{equation}
\label{eq:42}
 \rho^{\operatorname{mat}}(\alpha,\lambda)(t,t)=\{(1-\lambda)^2-4\abs{\alpha}^2\}^{1/2}\sum\limits_{\underset{t-s \textnormal{ even}}{s\leq t}}^{}\frac{t!}{s!(\frac{t-s}{2}!)^2}\abs{\alpha}^{t-s}\lambda^s.
\end{equation}
\end{eg}
\begin{eg}
  \label{eg:2-mode-pure}[2-mode mean zero pure gaussian states] Let $\rho({A,0}) = \ketbra{\psi_A}$ be a general $2$-mode mean zero pure gaussian state parametrized by a matrix  $A=\bmqty{\alpha &\beta\\ \beta &\gamma}\in M_2(\BC)$ as in Theorem \ref{sec:gauss-stat-uncert-5}. 
Using  (\ref{eq:56}) we will describe the particle basis expansion of $\ket{\psi_A}$. First step in this direction is the description of  the matrix index set $\Delta(\bt)$ in (\ref{eq:matrix-index}). We have
\begin{align}
\label{eq:47}
   \Delta(\bt) = \Bigl\{R=\bmqty{r_1&r_2\\0&r_3} \mid  r_1,r_2,r_3\in\BZ_+,  2r_1+r_2=t_1,2r_3+r_2=t_2 \Bigr\}. 
\end{align}
Recall from  the definition of $\Delta(\bt)$ 
that $\Delta(\bt)=\phi$ if $\abs{\bt}$ is odd. Let $\bt = (t_1,t_2)\in \BZ_+^2$ be such that $t_1+t_2$ is even. 
Then $t_1$ and $t_2$  are both even or both odd. 
There are two cases, namely, (i)  $t_{1}=2k,t_2=2\ell$, 
(ii) $t_{1}=2k+1,t_2=2\ell+1$, where $k,\ell\in \BZ_+$. 

\textbf{Case (i)}: 
The nonnegative integer equations $2r_1+r_2=2k$ and $2r_3+r_2=2\ell$ in (\ref{eq:47}) imply that $r_2$ must be  even  and $0\leq r_2\leq 2k\wedge 2\ell$. If $r_2=2r, 0\leq r\leq k\wedge \ell$, then $r_1=k-r $ and  $r_2=\ell-r$. 

\textbf{Case (ii)}: A similar argument as in Case (i) shows that, 
$r_2=2r+1$   and $0\leq r\leq k\wedge \ell$. 

Thus the set $\Delta(\bt) $ has $(k\wedge \ell) +1$ matrices in both cases and 

\begin{align}
\label{eq:70}
\Delta(\bt) =
  \begin{cases}
    \left\{\bmqty{k-r&2r\\0&\ell-r} \mid r\in \BZ_+,  0\leq r\leq k\wedge \ell  \right\}, & \textnormal{if } \bt= (2k,2\ell),\\
\phantom{....}\\
\left\{\bmqty{k-r&2r+1\\0&\ell-r} \mid r\in \BZ_+,  0\leq r\leq k\wedge \ell  \right\}, & \textnormal{if } \bt= (2k+1,2\ell+1)\\
\phantom{....}\\
\phantom{....}\phi, & \textnormal{otherwise.} 
  \end{cases}
\end{align}
So by  (\ref{eq:74}), 
\begin{align}
\label{eq:109}
\varphi_A(\bt) =
  \begin{cases}
    \sqrt{(2k)!(2\ell)!}\sum\limits_{r=0}^{k\wedge\ell}\frac{2^{2r}\alpha^{k-r}\beta^{2r}\gamma^{\ell-r}}{(k-r)!(2r)!(\ell-r)!} & \textnormal{ if } \bt = (2k,2\ell),\\
\phantom{...}\\
\sqrt{(2k+1)!(2\ell+1)!} \sum\limits_{r=0}^{k\wedge\ell}\frac{2^{2r+1}\alpha^{k-r}\beta^{2r+1}\gamma^{\ell-r}}{(k-r)!(2r+1)!(\ell-r)!} & \textnormal{ if } \bt =  (2k+1,2\ell+1).\\
\phantom{....}\\
\phantom{....}0, & \textnormal{otherwise.}
  \end{cases}
\end{align}
Now by  (\ref{eq:56}) 

\begin{align}
\label{eq:110}
\ket{\psi_A} &= \sqrt{c(A,0)}\sum\limits_{k,l\in \BZ_+}^{}\sqrt{(2k)!(2\ell)!}\sum\limits_{r=0}^{k\wedge \ell}\frac{\alpha^{k-r}(2\beta)^{2r}\gamma^{\ell-r}}{(k-r)!(2r)!(\ell-r)!}\nonumber\\
&\phantom{...................................}\times \left(\ket{2k,2\ell} +\frac{\sqrt{(2k+1)(2\ell+1)}2\beta}{2r+1}\ket{2k+1,2\ell+1}\right). 
\end{align}
Thus, in the particle basis measurement,  the probability of counting $2k$-particles in the first mode and $2\ell$-particles in the second mode is 
\begin{align}
\label{eq:113}
\operatorname{Pr}(2k,2\ell) = c(A,0)(2k)!(2\ell)!\abs{\sum\limits_{r=0}^{k\wedge \ell}\frac{\alpha^{k-r}(2\beta)^{2r}\gamma^{\ell-r}}{(k-r)!(2r)!(\ell-r)!}}^2
\end{align}
and that of counting $2k+1$-particles in the first mode and $2\ell+1$-particles in the second mode is 
\begin{align}
\label{eq:114}
\operatorname{Pr}(2k+1,2\ell+1) =c(A,0) (2k+1)!(2\ell+1)!\abs{\sum\limits_{r=0}^{k\wedge\ell}\frac{\alpha^{k-r}(2\beta)^{2r+1}\gamma^{\ell-r}}{(k-r)!(2r+1)!(\ell-r)!}}^2.
\end{align}
\end{eg} Next three examples are special cases of the previous example.
\begin{eg}\label{eg:2-mode-special-1}
Let  $A=\bmqty{\alpha &\beta\\ \beta &0}\in M_2(\BC)$ with $\norm{A}<1/2$, $\alpha,\beta \neq 0$. Demanding $\gamma = 0$ in (\ref{eq:109}) gives 
\begin{equation}
\label{eq:115}
 \varphi_A(\bt) =
\begin{cases}
 \sqrt{t!(t-2k)!} \frac{\alpha^k(2\beta)^{t-2k}}{k!(t-2k)!} & \textnormal{ if } \bt = (t, t-2k), 0\leq 2k\leq t\\
\phantom{........}\\
\phantom{...........} 0, &\textnormal{ otherwise.}
\end{cases}
\end{equation}
Thus  \emph{the number of particles in the second mode cannot exceed that in the first mode}. Furthermore, 
\begin{equation}
\label{eq:116}
\operatorname{Pr}(t,t-2k) = c(A,0) \binom{t}{k}\binom{t-k}{t-2k}\abs{\alpha}^{2k}(4\abs{\beta}^2)^{t-2k},  0\leq 2k\leq t.
\end{equation}
\end{eg}
\begin{eg}
  Let  $B=\bmqty{0 &\beta\\ \beta &\alpha}\in M_2(\BC)$ with $\norm{B}<1/2$, $\alpha,\beta \neq 0$. This case is similar to the previous example and it may be noticed that  \emph{the number of particles in the first mode cannot exceed that in the second mode} and $\varphi_B(t_1,t_2) = \varphi_A(t_2, t_1)$, where $A$ is same as that in Example \ref{eg:2-mode-special-1} and $\operatorname{Pr}(t-2k, t)$ is same as the value on the right side of (\ref{eq:116}).
 \end{eg}
 \begin{eg}[2-mode squeezed
vacuum (TMSV) state ]\label{sec:compl-entangl-pure-1} Let $A=\bmqty{0 &\beta\\ \beta &0}\in M_2(\BC)$ with $\norm{A}<1/2$, i.e., $\abs{\beta}<1/2$, assume further that $\beta\neq 0$ so that we get a TMSV state. In this case,  \begin{equation*}
          \varphi_{A}(\bt) = \begin{cases}
          2^{k}\beta^k, & \textnormal{if } t_{1} = t_{2} = k,\\
            0, & \textnormal{if } t_{1} \neq t_{2}.
          \end{cases}  
        \end{equation*}
Hence
   \begin{equation}\label{eq:94}
           \ket{\psi_{A}} = \sqrt{1-\abs{2\beta}^2} \sum\limits_{k=0}^{\infty} (2\beta)^k\ket{k,k}.
       \end{equation}
Thus \emph{the number of particles in both the modes must be the same}. Furthermore,
\begin{equation}
\label{eq:117}
\operatorname{Pr}(k,k)= (1-\abs{2\beta}^2)\abs{2\beta}^{2k},
\end{equation}
which is a \emph{geometric distribution} with parameter $p=\abs{2\beta}^2$. Equation (\ref{eq:94}) gives an example of a $2$-mode, mean zero, pure gaussian state which is \emph{invariant under the permutation of modes} and has the mixed one-mode marginal gaussian state 
\begin{equation}
\label{eq:58}
(1-4\abs{\beta}^2) \sum\limits_{k=0}^{\infty}(4\abs{\beta}^2)^k \ketbra{k}.
\end{equation}
This is a well known example of an entangled gaussian state which is called a photon number entangled state (PNES) by some authors. Refer to \cite{Duan-Giedke-Cirac-Zoller-2000, eisert-plenio-2003, sabapathy-solomon-simon-2011} for more details. Comprehensive accounts can be found in \cite{serafini2017quantum,ADMS95-PRA,EPR2009,Lvovsky2015}.
We shall meet a hierarchy of such arbitrary mode gaussian states in the following discussions.
 \end{eg}
\begin{eg}\label{3-mode-0-diagonal}
       Let $n=3,$ and $$A=\bmqty{0 & \alpha_{12} & \alpha_{13} \\
       \alpha_{12} & 0 & \alpha_{23}\\
       \alpha_{13} & \alpha_{23} & 0 }$$ with $\norm{A}<1/2$ 
so that there exists a pure gaussian state $\rho(A,0)$ with parameters $(c(A,0),0,A,0)$ by Theorem  \ref{sec:gauss-stat-uncert-5}. 
    If $\abs{\bt} = 2k$ the elements of $\Delta(\bt)$ which contribute  to the sum in the definition of $\varphi_A$ have the property $\sum_{i\leq j}r_{ij} = k$. Thus we have $t_{j} \leq k, \forall j$ and the subset
       $$\left\{ \bmqty{0 & k-t_{3} & k-t_{2} \\
       0 & 0 & k-t_{1}\\
       0 & 0 & 0
       }\right\} \subseteq \Delta(\bt)$$  is the index set of the sum in the definition of $\varphi_A$.
       For $\bt = (t_1,t_2,t_3),$ such that $t_1+t_2+t_3 =2k$,
\begin{align*}
\varphi_A(\bt) = \sqrt{t_{1}!t_{2}!t_{3}!}\frac{2^k\alpha_{12}^{k-t_{3}}\alpha_{23}^{k-t_{1}}\alpha_{31}^{k-t_{2}}}{(k-t_{1})!(k-t_{2})!(k-t_{3})!  }.
\end{align*}
 Hence \begin{equation}\label{G3eqnExample}
           \ket{\psi_A} = c \sum\limits_{k=0}^{\infty}2^{k}\sum\limits_{ \substack{t_{1}+t_2+t_{3}=2k\\ \underset{i}{\operatorname{max}}\ t_{i}\leq k}}  \frac{\sqrt{t_{1}!t_{2}!t_{3}!}}{(k-t_{1})!(k-t_{2})!(k-t_{3})!}\alpha_{12}^{k-t_{3}}\alpha_{23}^{k-t_{1}}\alpha_{31}^{k-t_{2}} \ket{t_{1},t_{2},t_{3}},       \end{equation}
       where $c=\sqrt{c(A,0)}$ is so that $\norm{\psi_A} =1$.
      \end{eg}

\begin{eg}\label{eq:n-mode-diagonal}[$n$-mode mean zero gaussian state with $A$ and $\Lambda$ as diagonal matrices]
Let $ D_{\bm{\alpha}} = \diag\{\alpha_1,\alpha_2, \dots, \alpha_n\}$ and  $ D_{\bm{\lambda}} = \diag\{\lambda_1,\lambda_2, \dots, \lambda_n\}$. Let $A= D_{\bm{\alpha}}$ and $\Lambda= D_{\bm{\lambda}}$, be the parameters of a gaussian state. Furthermore, let $Z_1^{\alpha_j,\lambda_j}$ be the bounded extension of the $1$-mode operator $c(\alpha_j,\lambda_j)\Gamma(\sqrt{\lambda_j})e^{\alpha_ja_ja_j}$ defined on $\spn \mathcal{F}\cup \CE$ in the sense of (\ref{eq:23}) and (\ref{eq:106}). If $\bt = (t_1,t_2,\dots,t_n)\in \BZ_+^n$, then the operator $Z_1^{A,\Lambda}$ in (\ref{eq:36}) satisfies 
\begin{align*}
Z_1^{A,\Lambda}\ket{\bt}=Z_1^{D_{\bm{\alpha}},D_{\bm{\lambda}}}\ket{\bt} = \otimes_{j=1}^nZ_1^{\alpha_j,\lambda_j}\ket{t_j}.
\end{align*} 
Hence 
\begin{align*}
\rho^{A,\Lambda}=\rho^{D_{\bm{\alpha}},D_{\bm{\lambda}}}= \otimes_{j=1}^n \rho^{\alpha_j,\lambda_j},
\end{align*}
where $\rho^{\alpha_j,\lambda_j}$ is the $1$-mode gaussian state with parameters $(c(\alpha_j,\lambda_j), 0, \alpha_j,\lambda_j)$, as in Example \ref{eg:1-mode-mixed}, $j=1, 2, \dots, n$.
\end{eg}

      Now we obtain a hierarchy of completely entangled pure gaussian states invariant under the action of the permutation group $S_n$  on the set of all the $n$ modes. 
      \begin{defn}
      Let $\CH = \CH_0\oplus \CH_1$. A state $\rho_{01}$ on $\GH=\Gamma(\CH_0)\otimes\Gamma(\CH_1)$ is said to be \emph{separable} if it can be written in the form 
\begin{align*}
\rho_{01}= \sum\limits_j^{}p_j\rho_0^j\otimes\rho_1^j
\end{align*} 
where $p_j\geq 0, \sum_jp_j=1$, $\rho_k^j$ is a state on $\Gamma(\CH_k), k=0,1, \forall j$. The state $\rho_{01}$ is said to be \emph{entangled} if it is not separable. It follows that if $\rho_{01}$ is a pure state, then it is separable if and only $\rho_{01}=\rho_0\otimes\rho_1$, where $\rho_k$ is a pure state on  $\Gamma(\CH_k), k=0,1$. Fix an orthonormal basis $\{e_1, e_2,\dots, e_n \}$ of $\CH$, a state $\rho$ on $\GH$ is said to be  \emph{completely entangled}, if $\rho$ remains entangled for any decomposition $\CH=\CH_0\oplus\CH_1$ where $\CH_0=\spn\{e_{i_1},\dots,e_{i_k}\}$, $\{i_1,\dots i_k\} \subset \{1,2,\dots,n\}$,  $1\leq k< n$, $\CH_1=\CH_0^{\perp}$. 
      \end{defn}
      \begin{rmk} A note on nomenclature.
      The notion  completely entangled is also known as \emph{not biseparable} as in \cite{dur-cirac-2000}, while they study it in multiqubit systems we study this for continuous variable case, particularly for pure gaussian states in this work.
        The related notion of completely entangled \emph{subspaces} are studied by several authors including \cite{Par2004, Bha2006, Walgate_2008}.
      \end{rmk}
      \begin{prop}\label{prop:partial-kb-gaussian}
Let $\CH = \CH_0\oplus \CH_1$.  Let $\rho(A,\Lambda)$ be the gaussian state with $\CE_2(\CH)$-parameters $(c(A,\Lambda), 0, A, \Lambda)$. If $A = \bmqty{A_{00}&A_{01}\\A_{10}&A_{11}}$ and $\Lambda = \bmqty{\Lambda_{00}&\Lambda_{01}\\ \Lambda_{10}&\Lambda_{11}}$ in the direct sum decomposition $\CH = \CH_0\oplus \CH_1$, then the marginal state $\rho_0 = \tr_1\rho(A,\Lambda)$ in the tensor product decomposition $\GH=\Gamma(\CH_0)\otimes\Gamma(\CH_1)$ has the $\CE_2(\CH_0)$-parameters $(c_0, 0, A_0, \Lambda_0)$ where 
\begin{align}
\label{eq:51}
  \begin{split}
    c_0 &= \frac{c(A,\Lambda)}{c(A_{11},\Lambda_{11})},\\
    A_0 &= A_{00}+ \frac{1}{4}C_{01}M(A_{11},\Lambda_{11})^{-1}C_{01}^T,\\
\Lambda_0&= \Lambda_{00}+ \frac{1}{4}C_{01}M(A_{11},\Lambda_{11})^{-1}C_{01}^{\dagger},
  \end{split}\\
\nonumber
    \textnormal{with}\\
C_{01}&=\bmqty{(\Lambda_{01}+2A_{01}) &i (\Lambda_{01}-2A_{01})}.\label{eq:marginal-gaussian-1} 
\end{align}
\end{prop}
\begin{proof}
  Notice first that $M(A_{11},\Lambda_{11})>0$ as it is a principal submatrix of the strictly positive matrix $M(A,\Lambda)$. The proof follows from a routine computation using (\ref{eq:marginal-gen-fn}) and the gaussian integral formula.
\end{proof}
 In the context of Proposition \ref{prop:partial-kb-gaussian}, 
 it  is evident from equations (\ref{eq:51}) and  (\ref{eq:marginal-gaussian-1}) that $\Lambda_0=\Lambda_{00}$ if and only if $C_{01} = 0$, and in such a case $A_0=A_{00}$. But by definition, $C_{01}= 0$ if and only if $A_{01}=0$ and $\Lambda_{01}=0$, i.e.,  both $A$ and $\Lambda$ are block diagonal in the decomposition  $\CH = \CH_0\oplus \CH_1$. Since $M(A_{jj},\Lambda_{jj})>0$, $\rho(A_{jj},\Lambda_{jj})$ is a gaussian state on $\Gamma(\CH_j  )$, $j=0,1$.  If $A$ and $\Lambda$ were block diagonal in the first place, then by item \ref{item:29}) in Examples \ref{sec:semigr-3}, $\rho(A,\Lambda) = \rho(A_{00},\Lambda_{00})\otimes \rho(A_{11},\Lambda_{11})$ is a separable state. 
 Furthermore, the positive value $\tr C_{01}M(A_{11},\Lambda_{11})^{-1}C_{01}^{\dagger}$ may be considered as a measure of the `influence' of the second party on the first. 
The discussion above has turned out to be a useful test for the entanglement of a bipartite pure gaussian state.
 \begin{thm}
 \label{sec:compl-entangl-pure-2}
Let $\CH = \CH_0\oplus \CH_1$.  Let $\rho$ be a pure gaussian state with $\CE_2(\CH)$-parameters $(c(A,0), 0, A, 0)$. Then $\rho$ is separable if and only if $A$ is a block diagonal matrix in the direct sum decomposition $\CH_0\oplus\CH_1$. In terms of the covariance matrix $S$, by (\ref{eq:50}), $\rho$ is separable if and only if $S$ is a block diagonal matrix in the direct sum decomposition $\CH_0\oplus\CH_1$.
 \end{thm}
 \begin{rmk}
   The covariance matrix part in Theorem \ref{sec:compl-entangl-pure-2} is a special case of the main theorem in \cite{werner-wolf-2001}.
 \end{rmk}
Moreover we have the following.
\begin{prop}\label{prop-autonne}
  Let $\rho(A,0)$ be a $n$-mode pure gaussian state. If  the singular values of $A$ are  $\{\alpha_1, \alpha_2, \cdots, \alpha_n\},$  then there exists a unitary $U$ on $\BC^n$ such that 
\begin{equation}
\label{eq:43}
\rho(A,0) = \Gamma(U)\otimes_{j=1}^n\rho(\alpha_j,0)\Gamma(U)^{\dagger},
\end{equation} where $\rho(\alpha_j,0) = \ketbra{\psi_{\alpha_j}}$ is a $1$-mode gaussian state with parameters $(\alpha_j, 0)$ as in Example \ref{eg:1-mode-pure}, $j=1,2,\dots, n$.
\end{prop}
\begin{proof}
  By Autonne's theorem in linear algebra (Corollary 2.6.6 in \cite{horn_johnson_2012}), the complex symmetric matrix $A$ has a singular value decomposition of the form  $A=UD_{\bm{\alpha}}U^T$. Lemma \ref{lem:conj-sec-quant} and Example \ref{eq:n-mode-diagonal} together  completes the proof.  
\end{proof}
\begin{rmk}
  In the light of Lemma \ref{lem:conj-sec-quant},  Example \ref{eq:n-mode-diagonal} and Proposition \ref{prop-autonne}, a pair of matrices $(A,\Lambda)$ satisfying $M(A,\Lambda)>0$ produces a product state up to a conjugation by a second quantization unitary if and only if there exists a single unitary $U$ such that $A = UD_{\bm{\alpha}}U^T$ and $\Lambda = UD_{\bm{\lambda}}U^{\dagger}$, where $D_{\bm{\alpha}}$ and $D_{\bm{\lambda}}$ are diagonal matrices.
\end{rmk}

\begin{thm}\label{sec:compl-entangl-pure}
 Let $A=[\alpha_{ij}]$ be a complex $n\times n$ symmetric matrix satisfying the following conditions:
\begin{enumerate}
\item\label{item:30} $\alpha_{ij}\neq 0$ for all $i\neq j$;
\item \label{item:31} $\norm{A}<\frac{1}{2}$.
\end{enumerate}
 Then the associated pure gaussian state $\ket{\psi_A}$ (or equivalently $\rho(A,0) =\ketbra{\psi_A}$) with $\psi_A$ as in (\ref{eq:56}) is completely entangled.
\end{thm}
\begin{proof}
  The matrix of $A$ is not block diagonal in any decomposition  $\CH=\CH_0\oplus \CH_1$ as in the definition of completely entangled states. Now Theorem \ref{sec:compl-entangl-pure-2} completes the proof.
\end{proof}

The following corollary provides a generalization of photon number entangled states (PNES) in Example \ref{sec:compl-entangl-pure-1}. Such states are studied in references  \cite{AdSeIl2004, AdIl2007, AdIl2008, LoFu2003, UsPrRa2007} under the name fully symmetric states.
\begin{cor}\label{sec:compl-entangl-pure-3}
  Let $\theta\in \BC$ be such that $\abs{\theta}<\frac{1}{2(n-1)}$and  $A$ be the $n\times n$ matrix with all diagonal entries equal to $0$ and all non diagonal entries equal to $\theta$, i.e.,
\begin{align*}
A = \theta \bmqty{0 &1&1&\cdots&1\\1&0&1&\cdots&1\\
\vdots&\cdots&\cdots&\cdots&\vdots\\1&1&\cdots&1&0}.
\end{align*}
Then $\ket{\psi_A}$ is a completely entangled zero mean pure gaussian state which is invariant under the action of the permutation group $S_n$ on the set of all the modes.
\end{cor}
\begin{proof}
  Observe that $\norm{A}<\frac{1}{2}$. So there exists a gaussian state $\ket{\psi_A}$ given by equation (\ref{eq:56}). Furthermore, $PAP^T = A$ for  any permutation matrix $P$, and hence $\Gamma(P)\ketbra{\psi_A}\Gamma(P)^{\dagger}=\ketbra{\psi_A}$ by Lemma \ref{lem:conj-sec-quant}. In other words, $\ketbra{\psi_A}$ is invariant under the action of the permutation group $S_n$ on the set of all the modes.
\end{proof}
\begin{rmk}
  As a special case of Corollary \ref{sec:compl-entangl-pure-3}, we have a completely entangled  $3$-mode pure gaussian state which is invariant under the action of $S_3$ on the modes when \[A=\theta\bmqty{0 & 1& 1\\
       1 & 0 & 1\\
       1 & 1 & 0 }, \abs{\theta}<\frac{1}{2}.\] 
By Example \ref{3-mode-0-diagonal}, we have in this case, 
\begin{align*}
\ket{\psi_A}=\sqrt{c(A,0)} \sum\limits_{k=0}^{\infty}2^{k}\theta^{k}\sum\limits_{\substack{t_{1}+t_2+t_{3}=2k\\ \underset{i}{\operatorname{max}}\ t_{i}\leq k}} \frac{
           \sqrt{t_{1}!t_{2}!t_{3}!}}{(k-t_{1})!(k-t_{2})!(k-t_{3})!} \ket{t_{1},t_{2},t_{3}}.
\end{align*}
\end{rmk}

\section{Tomography of Gaussian States}\label{sec:tomography}

In this section, we make some remarks on the tomography of an unknown gaussian state in $\GCn$ through the estimation of its $\CE_2(\CH)$-parameters $(c,\bm{\alpha}, A, \Lambda)$ by using finite set valued-measurements. For tomography based on the estimation of the mean annihilation and position-momentum covariance matrix parameters with countable set-valued measurements, we refer to \cite{Par-Sen-2015}.

Let $\bm{\alpha} = (\alpha_1,\alpha_2,\dots,\alpha_n)$, $A=[a_{jk}], \Lambda = [\lambda_{jk}]$. Recall the notations defined in (\ref{eq:126}), we have the following relations from (\ref{eq:61})
\begin{equation}
\label{eq:62.1}
\begin{aligned}
 \mel{\Omega}{\rho}{\Omega}&= c, &\mel{\chi_{jj}}{\rho}{{\Omega}}&=\sqrt{2}c(\frac{\alpha_{j}^{2}}{2} + a_{jj}),\\
\mel*{\chi_j}{\rho}{\Omega}&=c\alpha_{j}, & \mel{\chi_{jk}}{\rho}{\Omega}&=2c(\frac{\alpha_{j}\alpha_{k}}{2} + a_{jk}),\\
&& \mel{\chi_j}{\rho}{\chi_k}&=c(\alpha_{j}\bar{\alpha}_{k} + \lambda_{jk}),
\end{aligned}
\end{equation}
where $1\leq j,k \leq n$ and the equations involving $\chi_{jk}$ are valid only for $j\neq k$. Our aim is to estimate the $\CE_2(\CH)$-parameters by making measurements in the state $\rho$. 
To estimate $c$, consider the projection $\mathcal{P}_0 = \ketbra{\Omega}$ and the yes-no measurement 
\begin{align*}
\mathcal{M}_0 = \{\mathcal{P}_0, I-\mathcal{P}_0\}.
\end{align*}
Measurement of $\mathcal{M}_0$ in the state $\rho$ yields a classical random variable $X_0$ on the two point set $\mathcal{M}_0$ with values in $\{0,1\}$ and 
\begin{align*}
\operatorname{Pr}(X_0=1) = \tr \rho\mathcal{P}_0 = c.
\end{align*}
Hence by the law of large numbers,  estimates of $c$ can be obtained by making measurements $\mathcal{M}_0^{\otimes^k}$ in $\rho^{\otimes^k}$, $k\in \mathbb{N}$.
 The parameters $\bm{\alpha}, A$ and $\Lambda$ are functions of the scalars $\mel{u}{\rho}{v}$, where $u, v$ vary over the set 
\begin{align*}
\mathbf{B}= \{\ket{\Omega}, \ket{\chi_j}, \ket{\chi_{jk}}|1\leq j\leq k\leq n \}.
\end{align*} 
Towards estimating these parameters, recall the polarisation formula
\begin{equation}
\label{eq:62}
\mel{u}{\rho}{v} = \mel*{\frac{u+v}{\sqrt{2}}}{\rho}{\frac{u+v}{\sqrt{2}}}-i \mel*{\frac{u+iv}{\sqrt{2}}}{\rho}{\frac{u+iv}{\sqrt{2}}}-\frac{1-i}{2}(\mel{u}{\rho}{u}+\mel{v}{\rho}{v}).
\end{equation}
For each $j, 1\leq j \leq n$, let $\mathcal{P}_{j} = \ketbra{\chi_j}$, $\mathcal{P}_{j0} = \ketbra{\psi_j}$, $\mathcal{P}'_{j0} = \ketbra{\psi_j'},$ where $\ket{\psi_j} = \frac{\ket{\chi_j}+\ket{\Omega}}{\sqrt{2}}$ and $\ket{\psi_j'} = \frac{\ket{\chi_j}+i\ket{\Omega}}{\sqrt{2}}$. 
Consider the yes-no measurements 
\begin{equation}
\label{eq:63}
\begin{aligned}
\mathcal{M}_j &= \{\mathcal{P}_j, I-\mathcal{P}_j\}, &&\\
\mathcal{M}_{j0} &= \{\mathcal{P}_{j0}, I-\mathcal{P}_{j0}\},  &\ \mathcal{M}_{j0}'& = \{\mathcal{P}_{j0}', I-\mathcal{P}_{j0}'\}, 
\end{aligned}
\end{equation}
To estimate $\alpha_j$, take $u =\chi_j$ and $v =\Omega$ in (\ref{eq:62}). Each term on the right hand side of $(\ref{eq:62})$ can be estimated using the procedure described to estimate $c$ but using the measurements $\mathcal{M}_{j0}, \mathcal{M}_{j0}', \mathcal{M}_{j}$ and $\mathcal{M}_{0}$ respectively. 
Thus $\bm{\alpha}$ can be estimated using the measurements in (\ref{eq:63}).

For each $j, k, 1\leq j\leq k \leq n$, let $\mathcal{P}_{jk0} = \ketbra{\psi_{jk}}$, $\mathcal{P}'_{jk0} = \ketbra{\psi_{jk}'},$ 
 where $\ket{\psi_{jk}} = \frac{\ket{\chi_{jk}}+\ket{\Omega}}{\sqrt{2}}$ and $\ket{\psi_{jk}'} = \frac{\ket{\chi_{jk}}+i\ket{\Omega}}{\sqrt{2}}$.
Now consider the measurements 
\begin{align}
\label{eq:64}
\mathcal{M}_{jk0} &= \{\mathcal{P}_{jk0}, I-\mathcal{P}_{jk0}\},  &\ \mathcal{M}_{jk0}'& = \{\mathcal{P}_{jk0}', I-\mathcal{P}_{jk0}'\}. 
\end{align}
Assume that we have already estimated $c$ and $\bm{\alpha}$. Now by (\ref{eq:62.1}), to estimate $A$ and $\Lambda$, it is enough to estimate the scalars $\mel{\chi_{jk}}{\rho}{\Omega}$ and $\mel{\chi_j}{\rho}{\chi_k}$, $1\leq j\leq k\leq n$. To this end, let $\mathcal{P}$ denote the projection onto the subspace spanned by $\mathbf{B}$ and $\mathcal{M}$ be the von-Neumann measurement 
\begin{align*}
\mathcal{M} = \left\{\ketbra{\zeta}\ |\ \zeta\in \mathbf{B}\right\}\cup \{I -\mathcal{P}\},
\end{align*}  
which contains $N := \frac{(n+1)(n+2)}{2}+1$ mutually orthogonal projections. Now we label the elements of $\mathcal{M}$ using numbers from $0$ to $N-1$, as follows 
\begin{equation}
\label{eq:66}
\begin{aligned}
  \ketbra{\Omega}&\mapsto 0,  \phantom{...............}  \ketbra{\chi_r}\mapsto r, 1 \leq r \leq n,\\
\ketbra{\chi_{jk}}&\mapsto n+ \frac{(2n-j)(j-1)}{2}+k, 1 \leq j\leq k \leq n,\\
 I -\mathcal{P} &\mapsto N-1.
\end{aligned}
\end{equation}
 Thus we get a classical random variable $X$ on the $N$ point sample space $\mathcal{M}$, taking values $0,1,2,\dots, N-1$ with respective probabilities
\begin{equation}
\label{eq:65}
\begin{aligned}
 &\operatorname{Pr}(X=0) =  \mel{\Omega}{\rho}{\Omega}, \  \operatorname{Pr}(X=r) = \mel{\chi_r}{\rho}{\chi_r}, 1 \leq r \leq n,\\
&\operatorname{Pr}(X=n+ \frac{(2n-j)(j-1)}{2}+k)  =   \mel{\chi_{jk}}{\rho}{\chi_{jk}}, 1 \leq j\leq k \leq n,\\
&\operatorname{Pr}(X=N-1) = 1-\sum\limits_{t=0}^{N-2}\operatorname{Pr}(X=t).
\end{aligned}
\end{equation}
Hence the scalars $\mel{\chi_j}{\rho}{\Omega}$ and $\mel{\chi_j}{\rho}{\chi_k}$ can be approximated by using the polarisation formula (\ref{eq:62}) and making the measurements $\mathcal{M}^{\otimes^k}$ in $\rho^{\otimes^k}$, $k\in \mathbb{N}$.

Denoting by $Q$ the projection $\mathcal{P}_0+\mathcal{P}_1+ \mathcal{P}_2$ where $\mathcal{P}_j$ is the projection on the $j$-particle subspace and tomographing the finite dimensional density operator $(\tr \rho Q)^{-1}Q\rho Q$, it is possible to reduce the number of measurements. We leave this problem open for the present.

\section*{Conclusions}
\begin{enumerate}
    \item A Klauder-Bargmann integral representation of all gaussian symmetries in an $n$-mode boson Fock space is obtained.
    \item The notion of generating function of a bounded operator in the boson Fock space $\GCn$ over the $n$-dimensional Hilbert space $\BC^n$ is introduced and a $\dagger$-closed multiplicative semigroup $\CE_2(\CH)$ with $\CH=\BC^n$ is constructed. 
The semigroup $\CE_2(\CH)$ is closed under the weak operator topology and a state is gaussian if and only if it is an element of $\CE_2(\CH)$. Furthermore, $\CE_2(\CH)$ contains all the gaussian symmetries in $\GCn$.
    \item Using the properties of the semigroup $\CE_2(\CH)$, the set of all $n$-mode gaussian states is parametrized  by a set of scalars derived from the matrix entries of the gaussian state at $0, 1,$ and $2$-particle vectors of a particle basis in $\GCn$. The exact relations between these new parameters and the conventional set of means and covariances of position and momentum observables are obtained.
    \item A positive element $Z$ in the semigroup $\CE_2(\CH)$ is factorised as $Z_1^\dagger Z_1$, where $Z_{1} = \sqrt{c}\Gamma(\sqrt{\Lambda})\exp{\sum_{r=1}^{n} \lambda_ra_r+\sum_{r,s=1}^{n} \alpha_{rs}a_{r}a_{s}}$ on the dense linear manifold generated by all exponential vectors.
\item 
 An explicit particle basis expansion of an arbitrary mean zero pure gaussian state vector along with a density matrix formula for a general mean zero gaussian state is  obtained in terms of the $\CE_2(\CH)$-parameters. 
 
    \item 
A whole class of completely entangled $n$-mode pure gaussian states is constructed. This yields examples of such entangled states which are also
invariant under the action of the permutation group $S_n$ on the set of all the $n$ modes. 
    \item The new parametrization enables the tomography of an unknown $n$-mode gaussian state by $O(n^2)$ measurements with a finite number of outcomes. 
\end{enumerate}
\nocite{werner-wolf-2001, simon-2000, weedbrook-et-al-2012, WANG20071}
\textbf{AUTHOR'S CONTRIBUTION}

\vspace{0.4cm}
Both the authors contributed equally to this work

\begin{acknowledgments}We sincerely thank the anonymous referee for several constructive suggestions which helped us to improve the paper immensely.
The first author thanks Gayathri Varma and Ajit Iqbal Singh for several fruitful discussions, and the Indian Statistical Institute, Delhi centre for the postdoctoral fellowship and the facilities provided. The second author thanks the Indian Statistical Institute for providing a friendly environment and all the facilities for research in the preparation of this article under an emeritus professorship.
\end{acknowledgments}
\textbf{DATA AVAILABILITY}
 
 The data that supports the findings of this study are available within the article.
\bibliography{bibliography}

\end{document}